\documentclass[11pt,leqno]{article}

\usepackage{amsmath, amssymb, amsthm}
\usepackage{mathtools}
\usepackage{geometry}
\usepackage{setspace}
\usepackage{natbib}
\usepackage{hyperref}
\usepackage{booktabs}
\usepackage{graphicx}
\usepackage{caption}
\usepackage{subcaption}
\usepackage[flushleft]{threeparttable}
\usepackage{tabularx}
\usepackage{array}
\usepackage{todonotes}

\makeatletter
\def\tagform@#1{\maketag@@@{\hbox to 3em{\hss(\ignorespaces#1\unskip\@@italiccorr)}}}
\let\veqno\@@leqno
\renewcommand{\eqref}[1]{\textup{(\ref{#1})}}
\makeatother

\bibpunct{(}{)}{;}{a}{,}{,}

\newtheoremstyle{aerplain}%
  {1em}{1em}%
  {\itshape}%
  {}%
  {\bfseries}%
  {.}%
  {.5em}%
  {\thmname{#1}\thmnumber{ #2}\thmnote{ (#3)}}%

\newtheoremstyle{aerdef}%
  {1em}{1em}{\upshape}{}{\bfseries}{.}{.5em}%
  {\thmname{#1}\thmnumber{ #2}\thmnote{ (#3)}}

\theoremstyle{aerplain}
\newtheorem{theorem}{Theorem}
\newtheorem{proposition}{Proposition}
\newtheorem{lemma}{Lemma}
\newtheorem{corollary}{Corollary}

\theoremstyle{aerdef}
\newtheorem{assumption}{Assumption}

\newtheorem{remark}{Remark}

\hypersetup{colorlinks=true, linkcolor=blue, citecolor=blue,
            urlcolor=blue, pdfborder={0 0 0}}

\DeclareMathOperator*{\argmax}{arg\,max}
\DeclareMathOperator{\osc}{osc}
\newcommand{\R}{\mathbb{R}}

\newcommand{\E}{\mathbb{E}}
\renewcommand{\Pr}{\mathrm{Pr}}
\newcommand{\1}{\mathbf{1}}

\newcommand{\convd}{\xrightarrow{d}}
\newcommand{\convp}{\xrightarrow{p}}

\begin{document}


\begin{center}
  {\large\bfseries PPML and Heavy-Tailed Trade and Factor Flows:\\
   Why Standard Inference Fails and How to Fix It%
   {\renewcommand{\thefootnote}{\fnsymbol{footnote}}%
    \footnote{First arXiv date: September 16, 2026. We thank Yassine Sbai Sassi for very helpful discussions.}}%
   \setcounter{footnote}{0}\par}
  \vspace{1.25em}
  {\normalsize \textsc{Peter H. Egger}%
    \footnote{Professor of Economics, ETH Z\"urich. Email:
      \texttt{pegger@ethz.ch}.}
   \textsc{and} \textsc{Ting Ji}%
    \footnote{Associate Professor of Economics, Central University of Finance and Economics.
      Email: \texttt{jiting@cufe.edu.cn}.}
   \textsc{and} \textsc{Yulong Wang}%
    \footnote{Associate Professor of Economics, Lehigh University.
      Email: \texttt{yuw925@lehigh.edu}.}\par}
\end{center}

\vspace{0.75em}

\begin{abstract}
    The Poisson pseudo-maximum likelihood (PPML) estimator is widely used for estimating bilateral gravity equations. Its consistency requires only a correctly specified conditional mean. Conventional inference, however, also requires finite-variance scores and Gaussian limits. We show that these conditions fail: bilateral flows are Pareto-tailed, PPML scores have a stable limit under a structural gravity data-generating process, and sandwich confidence intervals are too narrow. We retain PPML for point estimation but replace sandwich inference with an $m$-out-of-$n$ bootstrap robust to heavy tails. Across three bilateral data settings, the correction is large and overturns conventionally significant gravity coefficients.
\\

    
\end{abstract}

\vspace{1.5em}

\section{Introduction}
\label{sec:intro}

The Poisson pseudo-maximum likelihood (PPML) estimator of \citet{santossilva&tenreyro2006} has transformed empirical trade research. 
It targets the conditional mean of bilateral sales or factor flows between countries $i$ and $j$, $y_{ij}$, as a function of a vector of observable explanatory variables, $x_{ij}$, directly as $\E[y_{ij}|x_{ij}]=\exp(x_{ij}'\beta_0)$, where $\beta_0$ is an unknown vector of parameters of interest.
In doing so, it avoids the heteroskedasticity bias of log-linear ordinary least squares (OLS) and handles zero bilateral flows naturally. The consistency of the estimator $\hat\beta$ of $\beta_0$ under a correct specification of the conditional mean is well established.

What has received much less attention is whether the statistical inference accompanying PPML estimates is reliable. 
By inference we mean the standard errors, confidence intervals, and hypothesis tests routinely reported alongside the point estimate. 
The conventional sandwich variance estimator and the associated Gaussian approximation of $\hat\beta$ require the PPML score to have a finite variance, which in turn requires the outcome variable $y_{ij}$ to have a finite second moment.

The question of whether bilateral sales or factor flows have a finite second moment is empirical. We address it directly using various datasets.
For instance, we use the OECD Inter-Country Input-Output (ICIO) database, which provides directed bilateral sales (the value shipped from an origin to a destination) for all ordered pairs of 79 countries, covering both intermediate and final use.
The answer to the above question based on these data is unambiguous: irrespective of whether we use total, only final-goods, or only intermediate-goods sales, the upper tail of bilateral sales is well described by a Pareto distribution with a tail index $\hat\alpha$ that is robustly below 2. 
This finding suggests that the underlying distribution of $y_{ij}$ has a very heavy tail and likely unbounded second moment. 
The Pareto fit is also visually tight across all considered aggregations of sales flows. 
Heavy-tailed sales shocks are an empirical regularity at every level of disaggregation in ICIO, irrespective of whether domestic sales are included or not. 
We document such patterns also using data on international emigration flows and foreign inward direct investment stocks.

Establishing the inferential consequences of these heavy tails requires more care than in the classic independent and identically distributed (i.i.d.) setting.
Bilateral sales or factor flows have an intrinsically dyadic design: (i)~the gravity covariates contain country-specific quantities (GDP, multilateral resistance terms, fixed effects) that repeat across pairs, and (ii)~under the structural gravity model of trade of \citet{EatonKortum2002} and \citet{AndersonVanWincoop2003}, country
GDPs and price indices are jointly determined by the world matrix of bilateral trade costs through general equilibrium (see also \citealp{fally2015structural}; \citealp{allen2020universal}). Similar models exist for migration and factor flows (see \citealp{artucc2015trade}, \citealp{dix-carneiro2014labor}, \citealp{caliendo2019dynamics}, \citealp{artuc2025jobs}). Analogous processes have been postulated for flows of goods and factors between subnational units such as regions or cities.
Conditional on this deterministic dyadic design, the maintained DGP has independent pair-level shocks. The PPML score across country pairs (or country-sector pairs) is therefore an independent but non-identically distributed triangular array: the general equilibrium (GE) structure enters as deterministic heterogeneity in the score multiplier, not as cross-pair stochastic dependence. 
The standard i.i.d.\ Gaussian asymptotics for PPML inference do not capture this combination of heavy tails and dyadic design heterogeneity.

The contribution of this paper is threefold. First, we formulate a general-equilibrium-consistent DGP for bilateral trade in which the only stochastic primitive is a pair-level i.i.d.\ multiplicative shock with a regularly varying tail.
The conditional mean is the GE-determined sales level as in \citet{AndersonVanWincoop2003}, and the resulting (concentrated) score decomposes as $s_{ij}=(\varepsilon_{ij}-1)\,\kappa_{ij}$, where $\varepsilon_{ij}-1$ is the centered residual and  $\kappa_{ij}$ is a deterministic dyadic array.
Conditional
on this array, the scores are independent across pairs $ij$ but not
identically distributed; GE effects create deterministic heterogeneity
in $\kappa_{ij}$, not stochastic dependence among scores.

Second, with $n$ denoting the sample size, we establish formally that under this DGP the PPML estimator converges to a multivariate $\alpha$-stable distribution at the rate $n^{1-1/\alpha}$, not to a Gaussian at the rate $\sqrt{n}$, where $\alpha$ is the Pareto exponent characterizing the tail heaviness of $y_{ij}$. We also establish that the conventional variance estimator diverges in probability, at the rate $n^{2/\alpha-1}$.
These results imply the failure of the conventional inference method under heavy tail.

Third, we propose a remedy that retains PPML for point estimation but replaces the Gaussian approximation with a nonparametric $m$-out-of-$n$ bootstrap adapted from \citet{Athreya1987} and \citet{ChiangSasakiWang2023}.
The procedure is self-normalizing, requires no estimation of the tail index $\alpha$ or the scale of the stable limit, and remains valid in both the Gaussian and stable regimes.
We implement the procedure in a Stata package, \texttt{ppmlmn}, which we develop alongside this paper.\footnote{The package is available at \url{https://github.com/yulongwang06/Tail-PPML} and can be installed in Stata with \texttt{net install}.}
We apply it to gravity regressions on three bilateral datasets, all at the country-pair level: bilateral sales flows of goods and services including international trade (ICIO 2022, for intermediate plus final goods, and for final and intermediate goods separately), bilateral
migration (\citealp{AbelCohen2019}), and bilateral stocks of foreign direct inward investment (FDI; \citealp{Steenbergen2022}). The PPML estimator is widely used with the mentioned three types of data and many others (see the following subsection for specific references).

The bootstrap confidence interval (CI) is substantially wider than the heteroskedasticity-robust (HC0) sandwich CI in every specification. The bootstrap-to-HC0 width ratio ranges from about $2.3$ to $2.7$ across the country-pair, migration, and FDI cross-sections. The ratio is of similar magnitude in a final-goods specification at the origin-sector by country-pair level, which we report as a robustness check.
The practical message is that standard inference understates uncertainty in PPML gravity regressions by more than $125$ percent, more than doubling the width of conventional intervals, often enough to overturn the significance of standard gravity covariates, and that a simple bootstrap correction restores valid inference.

\paragraph{Literature.}
Our paper contributes to two literature, in international trade and in econometrics.

First, the PPML estimator with classic standard error has been widely used in the trade literature.
Applications subsequent to \citet{santossilva&tenreyro2006} using PPML include the following selective examples. \citet{anderson2010changing}; \citet{egger&larch&staub&winkelmann2011}, \citet{melitz2014native}, \citet{fally2015structural}, \citet{breinlich2018selling}, \citet{baier2019widely}, \citet{barjamovic2019trade}, \citet{de2019britain}, \citet{WeidnerZylkin2021},  \citet{ahmad2023brexit}, \citet{french2024effects}, \citet{hieu2024effects}, \citet{bergstrand2025tails}, \citet{freeman2025unlocking}, \citet{nagengast2025staggered}, \citet{shepherd2025bias}, and \citet{shepherd2026surprising} all used bilateral sales or trade flows as outcome. \citet{peeters2012gravity}, \citet{albert2022immigration}, \citet{hoffmann2024drought}, \citet{morten2024effects}, \citet{artuc2025jobs}, \citet{kantor2025moonshot}, \citet{morales2025high}, \citet{ghose2026trade}, and \citet{eckert&peters2026} all used bilateral migration flows or stocks or worker mobility as outcome. And \citet{pica2011s}, \citet{alviarez2019multinational}, \citet{broner2023bilateral}, \citet{gu2023climate}, and \citet{kim2026multinational} all used direct investment and other measures of multinational firm activity as outcomes.
We apply our proposed method to a broad set of datasets to illustrate its wide empirical relevance.

Second, a complementary literature on nonlinear panels with many fixed effects studies how the incidental-parameter problem biases the PPML point estimator and how to remove that bias.
\citet{FernandezVal2016} develop analytical and jackknife bias corrections for nonlinear models with two-way unobserved heterogeneity, and \citet{WeidnerZylkin2021} specialize the analysis to three-way gravity, characterizing the bias of PPML and its consistency. Our profile-remainder bounds extend the negligibility of this fixed-effect estimation error to the heavy-tailed regime; they reduce to the finite-variance conditions of those papers when the shock has finite variance.
A parallel approach differences the fixed effects out at the moment level rather than estimating them: \citet{Charbonneau2017} removes nonlinear two-way fixed effects through conditional moment restrictions, an idea \citet{Jochmans2017} adapts to construct fixed-effect-free moment conditions for gravity.
Building on the latter, \citet{YangZhang2023} extend the fixed-effect-free generalized method of moments (GMM) estimation to three-way gravity models, obtaining an $N$-consistent, asymptotically normal estimator that requires no bias correction.
The construction is elegant on the incidental-parameter dimension but, by relying on quadruple product moments, it presumes finite sixteenth moments of the disturbance, so its $U$-statistic central-limit theorem and the derived $\sqrt{n}$ convergence rate both fail once $\alpha<2$.
All of these operate squarely in the finite-variance regime and target the bias of the coefficient, not the tail of the score, and are therefore silent on the heavy-tail failure documented in our Theorem~\ref{thm:stable}.

Third, a parallel strand concerns variance estimation and resampling inference rather than the bias of the point estimate, and it is the strand to which our remedy speaks most directly.
\citet{Zylkin2024} shows that a re-sampling bootstrap doubles, and hence corrects the bias of PPML point estimates while delivering less biased standard errors.
This procedure is the ordinary $n$-out-of-$n$ bootstrap, which is inconsistent for statistics in the infinite-variance domain of attraction \citep{Athreya1987}; restoring consistency requires resampling fewer than $n$ observations \citep{BickelGotzeVanZwet1997} or the closely related subsampling of \citet{PolitisRomano1994}.
The $m$-out-of-$n$ resampling we propose and adopt in Section~\ref{sec:remedy} is precisely this device.
On the clustering side, \citet{Pfaffermayr2023} corrects the downward bias of two-way cluster-robust standard errors in cross-sectional gravity by projecting out equi-correlated country shocks and applying a small-sample working-variance adjustment, in the spirit of the multiway-clustering variance of \citet{CameronGelbachMiller2011}.
This analysis maintains finite cluster-score variances throughout, whereas Appendix~\ref{app:cluster} discusses that this variance diverges under $\alpha\in(1,2)$, leaving the cluster-robust interval correctly scaled but anchored to the wrong limiting law.
Therefore, cluster subsampling along the lines of \citet{ChiangSasakiWang2023} would be required.
Taken together, these two strands address how the estimation of many fixed effects distorts inference, but neither confronts the non-Gaussian limit that arises when bilateral flows (or stocks) themselves lack a finite second moment.

\paragraph{Organization.} The remainder of the paper is organized as follows.
Section~\ref{sec:pareto} documents the Pareto tail of three bilateral datasets: goods and services sales flows (ICIO, at two levels of disaggregation), aggregate migration flows, and aggregate foreign direct investment stocks.
Section~\ref{sec:dgp} formulates the GE-consistent DGP and states the dyadic-score structure that drives our asymptotic theory.
Section~\ref{sec:failure} explains why a Pareto tail with $\alpha<2$ invalidates standard PPML inference, states our main limit theorem, and quantifies the resulting size distortions
through simulation. Section~\ref{sec:remedy} develops the $m$-out-of-$n$ bootstrap procedure and establishes its validity.
Section~\ref{sec:application} applies the procedure to all three datasets and discusses the empirical findings.
Section~\ref{sec:conclusion} concludes. All proofs are in Appendix~\ref{app:proofs}.

\section{The Pareto Tail of Bilateral Economic Flows and Stocks}
\label{sec:pareto}

This section documents empirically that bilateral economic flows and stocks of outcomes commonly studied in international economics have a Pareto tail with index below~2 across three distinct data settings of the gravity literature: bilateral sales, measured through the ICIO database in aggregate
and decomposed into final- and intermediate-sales components;
bilateral emigration flows; and bilateral foreign inward direct investment stocks. 
In every dataset, and whether or not domestic flows are included with goods and services sales, the estimated tail index lies below the finite-variance threshold of~2.

\subsection{Data}
\label{sec:data}

We use the OECD Inter-Country Input-Output (ICIO) database, which records, for 79 countries in aggregate, the value of output produced in an origin country and absorbed in a
destination country, broken down by 50 production sectors on the producing side and by use on the absorbing side.
Throughout, a bilateral flow is directed: it is the value shipped from the origin to the destination, or equivalently, the origin's outward sales to, or the destination's inward purchases from, the partner.
It is therefore asymmetric (the origin-to-destination value differs from the destination-to-origin value).
From this directed flow matrix we construct, all at the country-pair
level, three aggregates of the data that differ only in which use
categories are summed on the absorbing side:

\begin{enumerate}
  \item[(i)] \emph{Aggregate sales}: the directed origin-to-destination
    value of gross output sales, summed over all 50 producing sectors
    and over both final and intermediate use. This is the total sales of
    the origin to the destination.

  \item[(ii)] \emph{Final-use sales}: the same directed value
    restricted to output absorbed as final demand (HFCE, NPISH, GFCE,
    GFCF, INVNT).

  \item[(iii)] \emph{Intermediate-use sales}: the same directed value
    restricted to output absorbed as intermediate inputs into the
    destination's production: all remaining (industry) use categories.
\end{enumerate}

All three flows share the dimension of analysis (origin country,
destination country, year), giving $79\times 79 = 6{,}241$ ordered
pairs per year, or $6{,}162$ once the $79$ domestic ($i=j$) cells are
dropped. The sector and use dimensions are summed over and do not
index the observations. The use categories partition gross output, so
the final- and intermediate-use values sum exactly to the aggregate
for every ordered pair,
$y_{ij}=y_{ij}^{\mathrm{fin}}+y_{ij}^{\mathrm{int}}$: the
decomposition reweights the same total rather than drawing a separate
sample or changing the unit of observation. We carry all three
measures through the tail diagnostics (Table~\ref{tab:hill_raw}) and
the gravity regressions (Section~\ref{sec:application}), so that the
heavy-tail evidence can be read both at the aggregate sales level and
separately for each of the two economically distinct types of flows, final
demand and intermediate inputs.

In all regression applications, we report separately results for data including versus excluding domestic flows.
We make this choice for two reasons.
First, domestic flows are an order of magnitude larger than international ones, since a country sells far more to itself than to any single partner.
Therefore, the diagonal conflates the determinants of cross-border gravity with those of internal distribution, which the standard gravity covariates (distance,
contiguity, common language, colonial ties) alone are not built to capture; absorbing it requires an internal-trade dummy whose coefficient (above five log points) merely reflects this level gap.
Second, and central to this paper, the domestic observations are precisely the most extreme draws of the heavy-tailed outcome and carry extreme leverage, so including them maximizes the leverage of the extreme observations and makes the heavy-tail correction largest.
Reporting both samples lets us read the cross-border gravity estimates of interest off the international-only specification, while documenting that the inferential problem we study is, if anything, more severe once the diagonal in the sales-flow matrix is present.

We focus on the year 2022, the most recent year available in the ICIO 2024 release. 
Two accounting rows on the production side (TLS = taxes less subsidies, VA = value added) and two accounting columns on the use side (DPABR = direct purchases abroad by residents, GGFC =
government final consumption when treated as a demand category) are dropped, leaving 50 production sectors and 50 input-use sectors.
Bilateral distance, contiguity, common official language, and colonial-relationship indicators as regressors used in a log-linear parameterized trade-cost function come from the CEPII GeoDist database \citep{MayerZignago2011}; nominal GDP comes from the World Bank; all country-specific regressors will be absorbed by fixed effects.

We complement the trade data with two further bilateral datasets that are standard gravity settings in their own right.
For bilateral migration we use the migrant-flow estimates of \citet{AbelCohen2019}.
These are directed flows, not stocks: each observation is the number of people who emigrated from the origin country to the destination country over a five-year
interval (equivalently, the destination's immigration inflow from the origin), reconstructed for 232 countries over 1990--2015 from the demographic accounting of changes in foreign-born stocks; we use their pseudo-Bayesian closed-accounting series.
As with trade, the measure is directed and asymmetric, and we use the origin's emigration outflow to each destination.
For bilateral foreign direct investment we use the World Bank Harmonized Bilateral FDI database \citep{Steenbergen2022}, a directed source-by-host panel that harmonizes UNCTAD, OECD, IMF CDIS, and national sources for 247 economies.
In contrast to the migration measure, the FDI outcome is a stock, not a flow: it is the inward FDI position held in the destination (host) economy and sourced from the origin economy, in current US dollars.

\subsection{Tail Diagnostics}
\label{sec:diagnostics}

We apply two standard diagnostics. The \citet{Hill1975}
estimator,\footnote{Under a Pareto-type tail assumption, it has been well established that $\sqrt{k}(\hat{\alpha}(k)-\alpha)\overset{d}{\to}\mathcal{N}(0,\alpha^2) \text{ as } k\to\infty$ and $k/n\to0$ \citep[e.g.,][Chapter 3]{de2006extreme}. The 95\% confidence interval for $\alpha$ can then be constructed as $\hat{\alpha}(k) \pm 1.96\times\hat{\alpha}(k)/\sqrt{k}$.  }
\begin{equation}
  \hat\alpha(k) \;=\; \biggl[\,k^{-1}
    \sum_{i=1}^{k}\log(y_{(i)}/y_{(k+1)})\biggr]^{-1},
  \label{eq:hill}
\end{equation}
where $y_{(1)}\geq y_{(2)}\geq\cdots$ are the order statistics of
the sample, is plotted against $k$ to identify the threshold at which the
Pareto approximation stabilizes, and we also report the log-rank
estimator of \citet{arkolakis2010market} and \citet{GabaixIbragimov2011}.\footnote{A related trade literature estimates the tail parameter of firm-level sales distributions by the quantile--quantile (QQ) regression estimator of \citet{KratzResnick1996}, as in \citet{head2014welfare} and \citet{FontagneOrefice2018}.} 
As a visual goodness-of-fit check we use a tail probability--probability (P--P) plot: for the top $k$ exceedances above a threshold $u$ we fit a Pareto shape by the Hill (maximum-likelihood) estimator and plot the empirical conditional cumulative distribution function (CDF) of those exceedances against the fitted Pareto CDF $1-(u/y)^{\hat\alpha}$. 
If the upper tail is Pareto the points lie on the $45^\circ$ line; systematic departures reveal where a single-index Pareto under- or over-states the curvature of the tail. 
We report Hill estimates at four threshold values $k$ for robustness.

\subsection{Tail Evidence Across Datasets}

Table~\ref{tab:hill_raw} reports Hill estimates on raw bilateral outcomes for our five measures: directed country-pair sales in aggregate and split into their final- and intermediate-use
components, directed emigration flows, and directed inward FDI stocks.
For the three sales measures, where domestic ($i=j$) flows exist, we report estimates separately for specifications that include and exclude them. Migration and FDI are cross-border by construction, so only the
international tail is defined for them.

\begin{table}[!ht]
  \centering
  \caption{Hill Tail-Index Estimates on Raw Bilateral Flows}
  \label{tab:hill_raw}
  \begin{threeparttable}
  \begin{tabularx}{\textwidth}{@{}l*{4}{>{\centering\arraybackslash}X}@{}}
    \toprule
    & \multicolumn{4}{c}{Threshold $k$ (number of observations)} \\
    \cmidrule(lr){2-5}
    \multicolumn{5}{l}{\textit{Panel A. Aggregate directed sales (final and intermediate), ICIO 2022 ($n\approx 6{,}200$)}} \\
    & 50 & 100 & 200 & 500 \\
    \quad With domestic flows    & 0.70 & 0.54 & 0.61 & 0.69 \\
    \quad International only      & 1.81 & 1.57 & 1.24 & 1.06 \\[4pt]
    \multicolumn{5}{l}{\textit{Panel B. Final-goods directed sales, ICIO 2022 ($n\approx 6{,}200$)}} \\
    & 50 & 100 & 200 & 500 \\
    \quad With domestic flows    & 0.77 & 0.60 & 0.62 & 0.71 \\
    \quad International only      & 1.45 & 1.28 & 1.19 & 1.04 \\[4pt]
    \multicolumn{5}{l}{\textit{Panel C. Intermediate-goods directed sales, ICIO 2022 ($n\approx 6{,}200$)}} \\
    & 50 & 100 & 200 & 500 \\
    \quad With domestic flows    & 0.70 & 0.54 & 0.60 & 0.68 \\
    \quad International only      & 1.85 & 1.60 & 1.23 & 1.07 \\[4pt]
    \multicolumn{5}{l}{\textit{Panel D. Directed emigration flows, 2015 ($n\approx 35{,}000$)}} \\
    & 100 & 200 & 500 & 1{,}000 \\
    \quad Cross-border           & 1.59 & 1.32 & 1.06 & 0.86 \\[4pt]
    \multicolumn{5}{l}{\textit{Panel E. Directed inward FDI stocks, 2015 ($n\approx 9{,}000$)}} \\
    & 50 & 100 & 200 & 500 \\
    \quad Cross-border           & 1.26 & 1.35 & 1.03 & 0.81 \\
    \bottomrule
  \end{tabularx}
  \begin{tablenotes}[flushleft]\footnotesize
    \item \emph{Notes}: Hill estimator $\hat\alpha(k)$ from \eqref{eq:hill} on positive directed bilateral flows.
      Panels~A--C are ICIO 2022 directed country-pair sales summed over the 50 producing sectors: total (A) and its final-use (B) and intermediate-use (C) components, which sum to the total for every ordered pair.
      Panel~D uses \citet{AbelCohen2019} directed emigration flows and Panel~E \citet{Steenbergen2022} directed inward FDI stocks, both in the 2015 cross-section.
      ``With domestic flows'' adds the $i=j$ diagonal; migration and FDI are cross-border by construction.
      Threshold ranges differ across panels because larger $n$ makes higher-$k$ estimates feasible.
      Every estimate is below the finite-variance threshold of~2.
  \end{tablenotes}
  \end{threeparttable}
\end{table}

Three patterns emerge. First, the Hill estimate of $\alpha$ is robustly below~2 in every cell of Table~\ref{tab:hill_raw}: the finite-variance threshold $\alpha=2$ that underwrites Gaussian inference is empirically violated in all three settings.\footnote{These diagnostics apply to the unconditional cross-pair distribution of $y_{ij}$, whereas Assumption~\ref{ass:dgp}(a) restricts the centered shock. Table~\ref{tab:hill_resid} in Section \ref{sec:results-discussion} repeats them on the multiplicative residual from the fitted two-way specifications and reaches the same conclusion.} 
Second, the tail heaviness is comparable across very different bilateral outcomes. 
International sales have Hill estimates in the $\hat\alpha\approx 1.0$--$1.9$ range, in aggregate and in both the final- and intermediate-use components; migration is slightly
heavier ($\hat\alpha\approx 0.9$--$1.6$); and FDI is the most extreme, with $\hat\alpha$ falling to $0.8$--$1.0$ in the deep tail, close to the boundary at which even the first moment ceases to exist. 
The final- and intermediate-use decomposition shows that the heavy tail is intrinsic to both margins: neither isolates a finite-variance component, and the intermediate-use tail is, if
anything, marginally heavier than the final-use tail in the deep tail. 
Third, for sales, including domestic flows always produces a heavier tail: domestic flows are systematically larger than international ones because of an absence of many frictions to
economic transactions at national borders, and the largest domestic flows are the most extreme observations in the entire dataset, so excluding them removes the heaviest part of the right tail (\citealp{yotov2022role}, however, demonstrates that omitting domestic flows may bias the parameters of interest in gravity models).

Several cells, most persistently FDI and the deepest-threshold migration and domestic-sales estimates, fall at or below one.
This bears on the scope of the theory that follows.
PPML consistency (Theorem~\ref{thm:stable}(i)) is maintained throughout $\alpha\in(1,2)$, where the shock has a finite mean and the PPML conditional-mean model is well defined.
The lower bound $\alpha>1$ is not merely technical, since a pure Pareto right tail with $\alpha\le1$ has no finite mean and the conditional-mean DGP is not literally valid without truncation.
We therefore read the less-than-one estimates not as a literal calibration of $\alpha$ but as evidence of tail risk more severe than the distributional theorem quantifies: where it applies, the breakdown of Gaussian sandwich inference is only sharper.

Figure~\ref{fig:pareto_icio} provides a visual confirmation.
For each measure it plots the empirical conditional CDF of the top 500 tail exceedances against the CDF of a Pareto fitted to those same exceedances.
Under a Pareto upper tail the points lie on the $45^\circ$ line. They do, closely, in all five panels: the root-mean-squared error (RMSE) of the deviation from the diagonal ranges from $0.020$ (final-goods sales) to $0.037$ (inward FDI stocks).
The mild systematic bows show that a single-index Pareto is a good but not exact description, with the most extreme observations heavier than the fitted body; for FDI these are the special-purpose-entity hubs. From the econometric perspective, these plots motivate the assumption of a Pareto-type tail, which is formally referred to as regular variation.
See Section \ref{sec:stoch} for more details.

\begin{figure}[!ht]
  \centering
  \includegraphics[width=\textwidth]{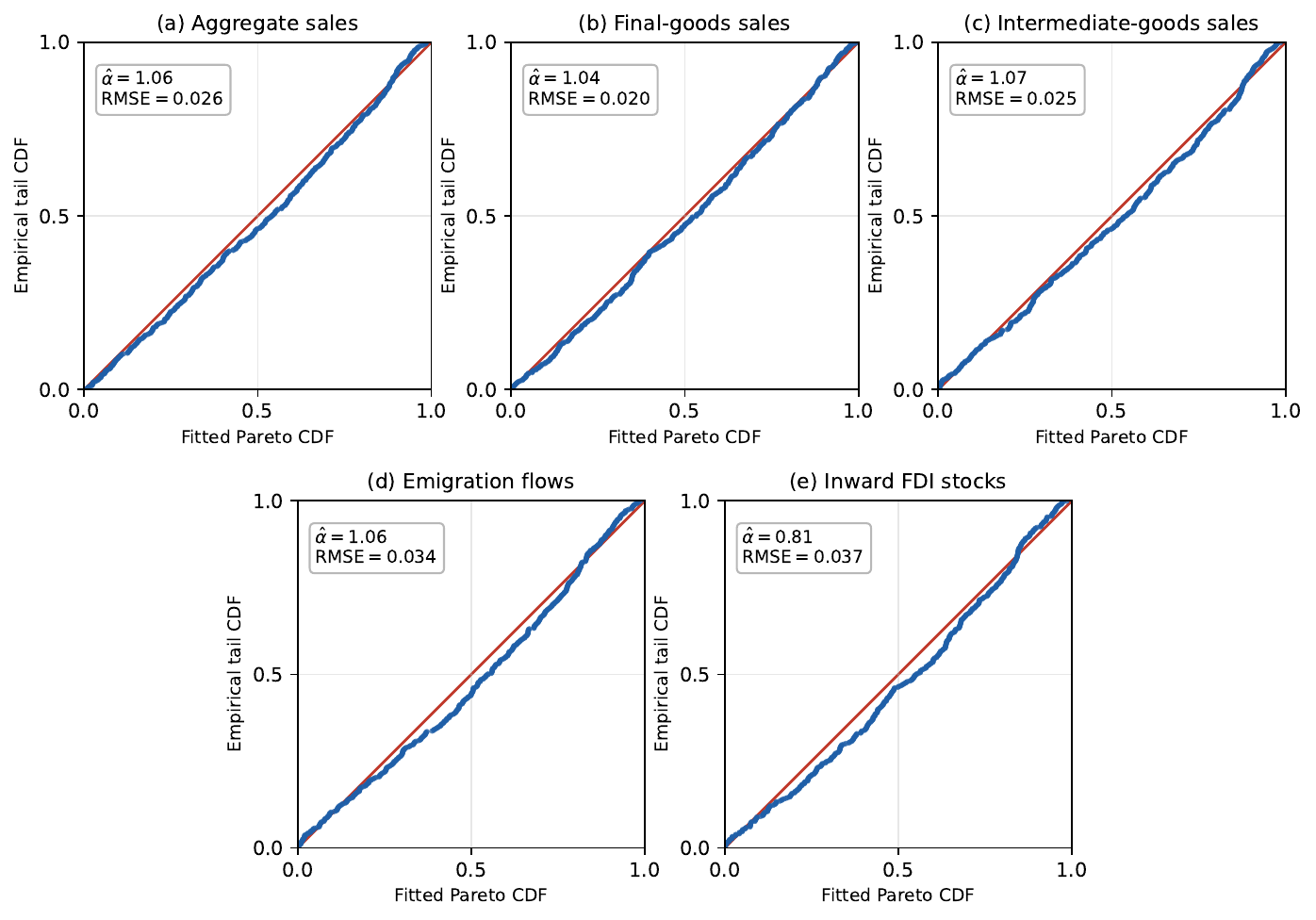}
  \caption{Tail probability--probability plots for the five measures of Table~\ref{tab:hill_raw}: directed country-pair sales in aggregate (a), final-use (b), and intermediate-use (c), ICIO 2022 international flows only; directed emigration flows (d, \citealp{AbelCohen2019}); and directed inward FDI stocks (e, \citealp{Steenbergen2022}), the latter two in the 2015 cross-section.
  For each measure the Pareto shape $\hat\alpha$ is fit by the Hill estimator on the top $k=500$ exceedances above the threshold $u$, and the empirical conditional CDF of those exceedances is plotted against the fitted Pareto CDF $1-(u/y)^{\hat\alpha}$; points on the $45^\circ$ line indicate a Pareto tail.
  Each panel reports $\hat\alpha$ and the root-mean-squared deviation from the diagonal, matching the $k=500$ column of Table~\ref{tab:hill_raw}.}
  \label{fig:pareto_icio}
\end{figure}

\section{A General-Equilibrium-Consistent Data-generating Process}
\label{sec:dgp}

This section formulates the data-generating process (DGP) under which
the asymptotic theory of the next two sections is developed. The
DGP adopts the canonical structural gravity model of
\citet{AndersonVanWincoop2003} formulated econometrically as in \citet{EggerStaub2016}, and adds
a single stochastic primitive: a pair-level i.i.d.\ multiplicative
shock with a regularly varying tail. The resulting score
decomposition $s_{ij}=(\varepsilon_{ij}-1)\,\kappa_{ij}$ isolates
the source of the heavy-tail behavior we analyze in
Section~\ref{sec:failure}. The generic structure of this model has been shown to coincide with that of \citet{EatonKortum2002} and other customary quantitative gravity models by \citet{arkolakis&costinot&rodriguezclare2012}.

\subsection{A Structural Aggregate Gravity Model}
\label{sec:ge-model}

Let $N$ denote the number of countries, with exogenous
populations $\{L_i\}_{i=1}^N$ and exogenous log trade costs
$\{\ln\tau_{ij}\}_{i,j=1}^N$, where $\tau_{ij}\geq 1$ and
$\tau_{ii}=1$. Let $\zeta_0<-1$ denote the trade elasticity. The
general-equilibrium per-capita-income vector
$W=(W_1,\ldots,W_N)$ solves the trade-balance fixed-point
system
\begin{equation}
  W_iL_i \;=\; \sum_{j=1}^N
    \frac{L_iW_i^{\zeta_0}\,\tau_{ij}^{\zeta_0}\,W_jL_j}
         {\sum_{h=1}^N L_hW_h^{\zeta_0}\,\tau_{hj}^{\zeta_0}},
  \qquad i=1,\ldots,N,
  \label{eq:gefixedpoint}
\end{equation}
with $W_1$ fixed as the numeraire. Existence and uniqueness of $W$
follow from standard arguments (\citealp{AlvarezLucas2007}; \citealp{allen2020universal}; \citealp{allen2024equilibrium}) under
positivity of $L_i$ and finiteness of $\tau_{ij}$. Equilibrium
$W_i$ depends on the entire matrix $\{\tau_{hj}\}_{h,j=1}^N$, not
just on the trade costs of country $i$, so two pair flows
$y_{ij}$ and $y_{i'j'}$ that share neither country index are
nonetheless linked through the common equilibrium prices.

Equilibrium bilateral sales are deterministic functions of
$(L,\tau,\zeta_0)$ given by
\begin{equation}
  \mu_{ij}(L,\tau,\zeta_0) \;=\;
    \frac{L_iW_i^{\zeta_0}\,\tau_{ij}^{\zeta_0}\,W_jL_j}
         {\sum_{h=1}^N L_hW_h^{\zeta_0}\,\tau_{hj}^{\zeta_0}}.
  \label{eq:mu-gravity}
\end{equation}
Taking logs of \eqref{eq:mu-gravity} and collecting terms,
\begin{equation}
  \ln\mu_{ij} \;=\;
    \phi^{\mathrm{ex}}_i+\phi^{\mathrm{im}}_j+\zeta_0\ln\tau_{ij},
  \quad
  \phi^{\mathrm{ex}}_i\equiv\ln(L_iW_i^{\zeta_0}),
  \quad
  \phi^{\mathrm{im}}_j\equiv
    \ln\frac{W_jL_j}{\textstyle\sum_h L_hW_h^{\zeta_0}\tau_{hj}^{\zeta_0}}.
  \label{eq:gravity-loglin}
\end{equation}
This is the canonical structural gravity equation for aggregate bilateral exports. The exporter fixed effect $\phi^{\mathrm{ex}}_i$ and the importer fixed effect
$\phi^{\mathrm{im}}_j$ are nonlinear functions of the entire
$(L,\tau)$ matrix through the GE solution
$W$.

\subsection{Stochastic Specification}
\label{sec:stoch}

Realized bilateral sales differ from their structural means
\eqref{eq:mu-gravity} due to measurement error, quality and
preference shocks, and idiosyncratic firm-level draws that
aggregate to the country-pair level. We model this through a
multiplicative pair-level shock,
\begin{equation}
  y_{ij}\;=\;\mu_{ij}\cdot\varepsilon_{ij},
  \qquad i,j=1,\ldots,N,\;\;i\neq j.
  \label{eq:dgp}
\end{equation}
The asymptotic frame lets $n\to\infty$ with $n=N(N-1)$. The
$N$-by-$N$ trade-cost matrix is treated as a triangular array
$\{\tau^{(N)}_{ij}\}_{i,j\leq N}$, and the corresponding GE
solution $W^{(N)}$ is determined by \eqref{eq:gefixedpoint} at
sample size $N$. To keep the array structure visible where it matters, we
retain the $N$ superscript in the conditions that quantify over or take limits in
$N$ (the boundedness, scale, and spectral-measure requirements of
Assumption~\ref{ass:dgp} and the limit theory), and suppress it elsewhere.

Denote the centered shock by $\eta_{ij}\equiv\varepsilon_{ij}-1\geq -1$.
Let $\beta_0\in \mathcal B$ denote the coefficients on the dyadic covariates $x_{ij}\in\mathbb{R}^p$, with $\mathcal B$ a compact convex subset of $\mathbb{R}^p$ and $\beta_0\in\operatorname{int}(\mathcal B)$, and let
\begin{equation}
  \kappa_{ij}\;\equiv\;\mu_{ij}\,\tilde x_{ij}\,\in\,\mathbb{R}^p,
  \label{eq:kappa-def}
\end{equation}
where $\tilde x_{ij}$ is the residual of $x_{ij}$ after partialling out the country fixed effects (FE) in the $\mu_{ij}$-weighted projection onto the FE column space (see Appendix \ref{app:profiling} for more details). The FE design vector $d_{ij}\in\R^{2N}$ is the exporter--importer incidence vector, defined by $d_{ij}'\phi=\phi^{\mathrm{ex}}_i+\phi^{\mathrm{im}}_j$: it is the $0/1$ selector with a one in the exporter-$i$ slot and a one in the importer-$j$ slot and zeros elsewhere. Every $d_{ij}$ is orthogonal to the direction $(\mathbf 1_N',-\mathbf 1_N')'$, so $\phi$ is identified only up to that direction; Appendix~\ref{app:profiling} records the quotient conventions under which all fixed-effect objects are well defined and representative-free.

The single primitive of this DGP is the following.

\begin{assumption}[GE-consistent dyadic DGP]
\label{ass:dgp}
The data $\{y_{ij}:i\neq j\}_{i,j=1}^N$ are generated by \eqref{eq:dgp}, where $\mu_{ij}$ is the GE-determined conditional mean in \eqref{eq:mu-gravity}.
The following four conditions hold:

\begin{enumerate}
  \item[(a)] (\emph{Heavy-tailed one-sided shock})
    $\{\varepsilon_{ij}\}_{i\neq j}$ i.i.d.;\quad $\E[\varepsilon_{ij}]=1$;\quad
    $\varepsilon_{ij}\ge0$ a.s.;\quad $\eta_{ij}=\varepsilon_{ij}-1$;
    \begin{equation}
      \Pr(\eta_{ij}>t)=c_+\,t^{-\alpha}\,\mathcal{L}(t),\qquad
      c_+>0,\quad \alpha\in(1,2),\quad
      \mathcal{L}(t)\to c_\ast\in(0,\infty)\ \text{as }t\to\infty;
      \label{eq:tail-balance}
    \end{equation}
    $\Pr(\eta_{ij}<-t)=0\ (t>1)$;\quad $n\,\Pr(\eta_{ij}>a_n)\to1$.

  \item[(b)] (\emph{Bounded, identified design})
    $\mu_{ij}(\beta,\phi)=\exp(x_{ij}'\beta+\phi^{\mathrm{ex}}_i+\phi^{\mathrm{im}}_j)$;\quad
    $\bar\phi(\beta)$ solves $\sum_{i\neq j}\bigl(\mu_{ij}^0-\mu_{ij}(\beta,\bar\phi(\beta))\bigr)d_{ij}=0$
    and denotes throughout its unique representative in the
    orthogonal complement of $(\mathbf 1_N',-\mathbf 1_N')'$, so that
    \eqref{eq:pseudo-bdd} is well posed;\quad
    $\phi_0=\bar\phi(\beta_0)$;\quad
    $H_{\phi\phi,n}(\beta)=\sum_{i\neq j}\mu_{ij}(\beta,\bar\phi(\beta))\,d_{ij}d_{ij}'$.
    For constants $0<c\le C<\infty$ and all $i\neq j$, $N$, $\beta\in\mathcal B$,
    \begin{equation}
      \|x_{ij}\|\le C,\quad \|\tilde x_{ij}\|\le C,\quad
      \max_i\bigl|\bar\phi^{\mathrm{ex}}_i(\beta)\bigr|\le C,\quad
      \max_j\bigl|\bar\phi^{\mathrm{im}}_j(\beta)\bigr|\le C,
      \label{eq:pseudo-bdd}
    \end{equation}
    \begin{equation}
      c\le\lambda_{\min}\!\bigl(N^{-1}H_{\phi\phi,n}(\beta)\bigr)
      \le\lambda_{\max}\!\bigl(N^{-1}H_{\phi\phi,n}(\beta)\bigr)\le C
      \quad\text{on the identified subspace},
      \label{eq:dense-design}
    \end{equation}
    \begin{equation}
      H_n=n^{-1}\textstyle\sum_{i\neq j}\mu_{ij}\,\tilde x_{ij}\tilde x_{ij}'
      \to H\succ0.
      \label{eq:Hn-limit}
    \end{equation}

  \item[(c)] (\emph{Score spectral measure})
    $\bar A_n=n^{-1}\sum_{(i,j)}\|\kappa_{ij}\|^\alpha\to\bar A\in(0,\infty)$;
    \begin{equation}
      \frac{\sum_{(i,j)}\delta_{\kappa_{ij}/\|\kappa_{ij}\|}\,
            \|\kappa_{ij}\|^\alpha}
           {\sum_{(i,j)}\|\kappa_{ij}\|^\alpha}
      \;\Rightarrow\;\Gamma\ \text{on}\ \mathbb{S}^{p-1};\qquad
      \int_{\mathbb{S}^{p-1}}\theta\theta'\,\Gamma(d\theta)\succ0,
      \label{eq:spec-meas}
    \end{equation}
    where pairs with $\kappa_{ij}=0$ are omitted from the numerator
    sum (their weight is zero).

  \item[(d)] (\emph{Negligible fixed-effect profiling})
    With $R_n(\beta)$, $H_n^{\ast}(\beta)$, $\hat\Omega_n^{\mathrm{prof}}$ of
    Appendix~\ref{app:profiling} and $s_{ij}=\eta_{ij}\kappa_{ij}$,
    \begin{equation}
    \begin{gathered}
      \sup_{\beta\in\mathcal B}\bigl\|R_n(\beta)\bigr\|=o_p(a_n),
      \qquad
      \frac{n}{a_n^{2}}\Bigl\|\hat\Omega_n^{\mathrm{prof}}(\beta_0)
      -n^{-1}\textstyle\sum_{i\neq j}s_{ij}s_{ij}'\Bigr\|\convp0 ,\\
      \sup_{\|\beta-\beta_0\|\le\delta_n}
      \bigl\|H_n^{\ast}(\beta)-H\bigr\|\convp0
      \quad\text{for every sequence }\delta_n\downarrow0 .
      \end{gathered}
      \label{eq:hl-remainder}
    \end{equation}
\end{enumerate}
\end{assumption}

Assumption~\ref{ass:dgp} collects the maintained conditions. 
Part~(a) restricts the centered shock $\eta=\varepsilon-1\ge-1$. Its one-sided regularly varying right tail places it in the $\alpha$-stable domain, with norming $a_n$ defined by $n\Pr(\eta>a_n)\to1$, so that $a_n=(c_+c_\ast n)^{1/\alpha}(1+o(1))\asymp n^{1/\alpha}$. The lower bound $\alpha>1$ gives a finite mean and PPML consistency. The upper bound $\alpha<2$ gives the non-Gaussian limit. 
Both bounds have a firm-level reading: in the \citet{Melitz2003} member of the \citet{arkolakis&costinot&rodriguezclare2012} class with a Pareto distribution, $\alpha$ equals the productivity exponent divided by the elasticity of substitution minus one, so $\alpha>1$ is the integrability condition that literature already imposes and $\alpha<2$ holds across its standard calibration range.
Appendix~\ref{app:primitive-a} gives the argument, shows that entry cutoffs left-truncate the productivity draw and so move $\mu_{ij}$ without changing $\alpha$, and shows that aggregation across the exporters serving a pair leaves the index unchanged. 
The condition $\mathcal L(t)\to c_\ast$ includes the Hall class \citep{Hall1982} and, more generally, any tail whose slowly varying factor converges to a positive constant. The Hall class covers the Pareto, Student-$t$, Fr\'echet, and stable families and is the standard setting for tail-index estimation. The condition permits deviations from an exact Pareto tail while keeping the norming a clean power of $n$. 

Part~(b) takes the regressors and the true and pseudo-true fixed effects to be uniformly bounded. With compact $\mathcal B$, this delivers $\mu_{ij}(\beta,\bar\phi(\beta))\in[c,C]$, a bounded multiplier $\kappa_{ij}=\mu_{ij}\tilde x_{ij}$, the well-conditioned fixed-effect Gram \eqref{eq:dense-design}, and a non-degenerate concentrated Hessian $H_n\to H\succ0$; the Gram condition identifies the fixed effects and makes the concentrated objective strictly concave. Bounded fixed effects are the bounded country heterogeneity implied by bounded log trade costs and populations and the GE stability of \citet{AlvarezLucas2007}; the Sinkhorn characterization of the two-way effects there also bounds the pseudo-true $\bar\phi(\beta)$ uniformly over $\mathcal B$. This is the dyadic analogue of the bounded-index condition of \citet[Assumption~4.1(iv)--(v)]{FernandezVal2016}, which for the Poisson likelihood confines the conditional mean to $0<b_{\min}\le\mu\le b_{\max}$ near the truth. \citet{WeidnerZylkin2021} adopt the same regularity for two-way gravity PPML. We state the global ($\beta\in\mathcal B$) version, which the consistency argument requires. Part~(c) is the dyadic spectral-measure convergence; there $\Rightarrow$ denotes weak convergence and $\mathbb S^{p-1}$ the unit sphere in $\R^p$. It holds automatically under i.i.d.\ sampling and for GE economies whose cost matrices densify regularly.

Part~(d) is the only high-level condition. Estimating $2N$ country effects from $n=N(N-1)$ flows injects sampling noise into the concentrated score, Hessian, and score covariance. Part~(d) requires that noise to be small relative to the heavy-tailed signal, whose stable scale $a_n$ exceeds the $\sqrt n$ of the finite-variance case. In the finite-variance benchmark, the three conditions reduce to the standard incidental-parameter negligibility conditions of \citet{FernandezVal2016} and \citet{WeidnerZylkin2021}. Appendix~\ref{app:primitive-d} verifies~(d) from the Gram conditioning \eqref{eq:dense-design} of part~(b) and one further primitive condition, a non-spiked residualized design. Proposition~\ref{prop:primitive-d} shows that all three conditions hold throughout $\alpha\in(1,2)$ whenever the residualized design satisfies \eqref{eq:design-op} with $\gamma<1$.

\begin{remark}[Pareto exponent and stable index]\label{rem:pareto-stable}
The tail index $\alpha_P$ of part~(a) and the index $\alpha_S$ of the stable limit coincide on the relevant domain. The generalized central limit theorem \citep{GnedenkoKolmogorov1954,SamorodnitskyTaqqu1994} gives $\alpha_S=\min(\alpha_P,2)$. For $\alpha_P\in(1,2)$ the limit is therefore non-Gaussian $\alpha$-stable with $\alpha_S=\alpha_P$ at rate $n^{1/\alpha_P}$, and we write $\alpha$ for both. When the shock has finite variance, the normalization is $\sqrt n$ and the limit is Gaussian. 
Recall that in Section~\ref{sec:pareto} the estimated $\hat\alpha$ lie in $(0.5,1.9)$ across all cells of Table~\ref{tab:hill_raw}, decisively in the non-Gaussian regime.
\end{remark}

\begin{remark}[Including the domestic diagonal]\label{rem:diagonal}
The analysis extends to the full $n=N^2$ sample (including $i=j$) provided the domestic shocks share the tail index $\alpha$, $\sup_N\max_i\mu_{ii}<\infty$, and the $N^2$-observation array satisfies the analogues of Assumption~\ref{ass:dgp}(b)--(d); note that adding the diagonal changes the residualization. The boundedness condition binds, since domestic flows dwarf any single bilateral flow. Large domestic cells can materially increase the finite-sample score scale, although their asymptotic contribution to $\bar A_n$ depends on their leverage and on the resulting residualization, the diagonal contributing $N$ of the $N^2$ cells.
\end{remark}

\section{Why PPML Fails under Heavy Tails}
\label{sec:failure}

Section~\ref{sec:pareto} established that bilateral sales have a Pareto tail with index below 2 at every level of disaggregation, and Section~\ref{sec:dgp} formulated a GE-consistent DGP in which this heavy tail is inherited from a pair-level shock $\varepsilon_{ij}$. 
This section shows that the heavy tail invalidates the standard inference of PPML. 
We derive the non-Gaussian asymptotic distribution of the PPML estimator under Assumption~\ref{ass:dgp}, show that the sandwich variance estimator diverges, and quantify the resulting size distortions through a simulation calibrated to the ICIO tail.

\subsection{Setup and Consistency}
\label{sec:setup}

Throughout this section and the proof in Appendix~\ref{app:stable-proof} we work with the two-way fixed-effect PPML estimator, profiling out the $2N$ country effects. The parameter of interest is $\hat\beta\in\mathbb{R}^p$, the fixed-dimensional block on the dyadic covariates. 
Under Assumption~\ref{ass:dgp} the conditional mean carries the fixed effects,
\begin{equation}
  y_{ij}=\mu_{ij}\,\varepsilon_{ij},\qquad \mu_{ij}=\exp\bigl(x_{ij}'\beta_0+d_{ij}'\phi_0\bigr),\qquad \E[\varepsilon_{ij}]=1,
  \label{eq:model}
\end{equation}
where $\phi_0$ collects the country fixed effects (with $d_{ij}$ the incidence vector defined above), so $\E[y_{ij}|x_{ij}]=\mu_{ij}$. 
Reindexing the $n=N(N-1)$ ordered pairs by $\ell=1,\dots,n$, the estimator maximizes the two-way fixed-effect Poisson objective
\begin{equation}
  (\hat\beta,\hat\phi)=\argmax_{\beta\in \mathcal{B},\;\phi}\;\frac{1}{n}\sum_{\ell=1}^{n} \Bigl[\,y_\ell\,(x_\ell'\beta+d_\ell'\phi)-\exp(x_\ell'\beta+d_\ell'\phi)\Bigr].
  \label{eq:ppml_obj}
\end{equation}
The maximizer $\hat\phi$ is unique only up to the normalization direction $(\mathbf 1_N',-\mathbf 1_N')'$; every object reported below (fitted means, residuals, test statistics) is invariant to the representative.
Concentrating out the fixed effects yields $\hat\phi(\beta)$ and the profiled estimator is written as $\hat\beta=\argmax_{\beta\in\mathcal B}Q_n(\beta,\hat\phi(\beta))$. 
By the first-order conditions, $\hat\beta$ solves the partialled-out score equation
\begin{equation}
  \sum_{\ell=1}^{n}\bigl(y_\ell-\hat\mu_\ell\bigr)\,\hat{\tilde x}_\ell(\hat\beta)=0,\qquad \hat\mu_\ell=\exp\bigl(x_\ell'\hat\beta+d_\ell'\hat\phi(\hat\beta)\bigr),
  \label{eq:score-foc}
\end{equation}
where $\hat{\tilde x}_\ell(\hat\beta)$ is the $\hat\mu$-weighted residual of $x_\ell$ on the fixed-effect design (Appendix~\ref{app:profiling}).
Crucially, the fitted mean $\hat\mu_\ell$ retains the estimated country effects; partialling out acts only on the score direction, not on the conditional mean.
The heavy-tail analysis operates on this concentrated score, and its validity requires that the estimation error in the $2N$ country fixed effects, the incidental-parameter problem of two-way gravity models, be negligible relative to the heavy-tailed variation of the score itself.
This is Assumption~\ref{ass:dgp}(d); Appendix~\ref{app:primitive-d} verifies it from primitive design conditions, which reduce to the standard incidental-parameter conditions of \citet{FernandezVal2016} and \citet{WeidnerZylkin2021} in the finite-variance benchmark.

The PPML estimator is consistent under correct specification of the conditional mean. 
The dyadic objective \eqref{eq:ppml_obj} is jointly concave in $(\beta,\phi)$, and concentrating preserves concavity in $\beta$; combined with the dense, well-conditioned design and the positive-definite Hessian limit of Assumption~\ref{ass:dgp}(b), this delivers $\hat\beta\convp\beta_0$. 
We give the formal argument as Step~3 of the proof of Theorem~\ref{thm:stable} in Appendix~\ref{app:stable-proof}.

\subsection{The Stable Limiting Distribution}
\label{sec:stable}

Asymptotic normality requires more than consistency. Writing the $n$ score contributions as
\begin{equation}
  s_\ell=(\varepsilon_\ell-1)\,\kappa_\ell,\qquad \kappa_\ell=\mu_\ell\,\tilde x_\ell,
  \label{eq:score-l}
\end{equation}
the classical central limit theorem (CLT) for $n^{-1/2}\sum_\ell s_\ell$ requires a finite second moment of $s_\ell$, which is equivalent to a finite
second moment of $\varepsilon_\ell-1$. 
When the shock has a regularly varying tail with index $\alpha<2$ as in Assumption~\ref{ass:dgp}(a), $\E[(\varepsilon_\ell-1)^2] =\infty$, the sandwich variance estimator is inconsistent, and the Gaussian approximation for $\hat\beta$ is invalid.

Section~\ref{sec:pareto} documents Hill tail indices $\hat\alpha\approx 1.0$--$1.9$ on raw bilateral sales across the ICIO specifications, well below the finite-variance threshold of~2,
so the regularly varying shock of Assumption~\ref{ass:dgp}(a) is the empirically operative case. The same is true for migration flows and FDI stocks for pairs of countries.

The correct limiting theory under infinite variance is the generalized central limit theorem \citep{GnedenkoKolmogorov1954,SamorodnitskyTaqqu1994}: if the
summands are regularly varying with index $\alpha\in(1,2)$, then suitably normalized sums converge to an $\alpha$-stable law.
Adapting this to the dyadic score requires a triangular-array version of the CLT, in which the spectral measure of the limit is determined by the empirical distribution of the deterministic multipliers $\{\kappa_\ell/\|\kappa_\ell\|\}$ on the unit sphere weighted by $\|\kappa_\ell\|^\alpha$, rather than by an i.i.d.\ sampling distribution. This is the role of
Assumption~\ref{ass:dgp}(c).

The main theorem of the paper states that under Assumption~\ref{ass:dgp}, the (concentrated) PPML estimator converges at rate $n^{1-1/\alpha}$ to a multivariate $\alpha$-stable distribution, the (unscaled) score second-moment matrix diverges at rate $n^{2/\alpha-1}$, and the conventional sandwich confidence interval does not generally have asymptotic coverage at the nominal level.

\begin{theorem}[Stable Limit and Variance Divergence for PPML under a GE-Consistent DGP]\label{thm:stable}
Suppose Assumption~\ref{ass:dgp} holds.
Let $a_n$ be the normalizing sequence defined in Assumption~\ref{ass:dgp}(a).
Let $\hat\beta\in\mathbb{R}^p$ denote the concentrated PPML estimator of the structural-coefficient block, with the $2N$ country fixed effects partialled out.
Then as $n\to\infty$ with $n=N(N-1):$
\begin{enumerate}
  \item[(i)] \emph{(Consistency)} $\hat\beta\convp\beta_0$.

  \item[(ii)] \emph{(Stable limit)} With
    $\Lambda=C_\alpha\bar A\,\Gamma$ the finite spectral measure on $\mathbb S^{p-1}$ built from Assumption~\ref{ass:dgp}(c), where $C_\alpha=\Gamma_E(1-\alpha)\cos(\pi\alpha/2)$ and $\Gamma_E$ is the Euler gamma function,
    \begin{equation}
      \frac{n}{a_n} \bigl(\hat\beta-\beta_0\bigr)\;\convd\; H^{-1}\,S_\alpha(\Lambda),
      \label{eq:main-limit}
    \end{equation}
    with $S_\alpha(\Lambda)$ the multivariate $\alpha$-stable law with
    spectral measure $\Lambda$.

  \item[(iii)] \emph{(Score-covariance divergence)} With
    $\hat\Omega_n\equiv n^{-1}\sum_{(i,j)} \hat s_{ij}\hat s_{ij}'$ and $\hat s_{ij}=(y_{ij}-\hat\mu_{ij})\,\hat{\tilde x}_{ij}(\hat\beta)$ the PPML scores at $(\hat\beta,\hat\phi(\hat\beta))$ (Appendix~\ref{app:profiling}), there is a positive-semidefinite $\alpha/2$-stable random matrix $\Xi$, jointly distributed with
    $S_\alpha(\Lambda)$, such that
    \begin{equation}
      \frac{n}{a_n^2}\,\hat\Omega_n\;\convd\;\Xi ;
      \label{eq:meat-rate}
    \end{equation}
    in particular $\hat\Omega_{n,kk}\convp\infty$ for every $k$.

  \item[(iv)] \emph{($t$-statistic limit)} Let $\widehat{\mathrm{se}}_k=\bigl([\hat H_n^{-1}\hat\Omega_n\hat H_n^{-1}]_{kk}/n\bigr)^{1/2}$ with $\hat H_n=n^{-1}\sum_{(i,j)}\hat\mu_{ij}\,
    \hat{\tilde x}_{ij}(\hat\beta)\hat{\tilde x}_{ij}(\hat\beta)'$. For each $k$,
    \begin{equation}
      \hat t_k\;=\;\frac{\hat\beta_k-\beta_{0,k}}{\widehat{\mathrm{se}}_k}\;\convd\;J^\star_{\alpha,k}\;\equiv\;\frac{[H^{-1}S_\alpha(\Lambda)]_k}{\bigl([H^{-1}\Xi H^{-1}]_{kk}\bigr)^{1/2}} .
      \label{eq:tstat-limit}
    \end{equation}
    The distribution function of $J^\star_{\alpha,k}$ is continuous, and $J^\star_{\alpha,k}\neq\mathcal{N}(0,1)$ in distribution for every $\alpha\in(1,2)$.
\end{enumerate}
\end{theorem}

Three features of Theorem~\ref{thm:stable} matter for empirical gravity work. 
First, since $a_n=(c_+c_*n)^{1/\alpha}(1+o(1))$ under Assumption~\ref{ass:dgp}(a), the estimator converges at rate $n^{1-1/\alpha}$ rather than $\sqrt n$: at $\alpha=1.4$ the rate is
$n^{0.286}$, so a tenfold increase in the number of country pairs raises the precision of an estimated trade-cost elasticity by a factor of about $1.9$, against $3.2$ under $\sqrt n$ asymptotics. 
Second, aggregation across the $n=N(N-1)$ pairs does not average the heavy tail away: the largest single score contribution is of the same order $a_n$ as the centered aggregate
score fluctuation, so the limiting law inherits the full tail index $\alpha$ of the pair-level shock.
Third, the general-equilibrium structure enters inference through the spectral measure $\Lambda$: market sizes, trade costs, and multilateral resistance generate the deterministic dyadic
heterogeneity in the score multipliers $\kappa_{ij}$, and $\Lambda$ records where on the sphere that heterogeneity concentrates, while the tail index is inherited unchanged from the shock.

Parts~(iii) and~(iv) carry the inferential consequences. 
The sandwich variance estimator does not converge: $\hat\Omega_n$ diverges at rate $a_n^2/n$ (equivalently $n^{2/\alpha-1}$), so reported standard errors are not consistent estimates of a finite asymptotic variance but realizations of a random scale. 
The $t$-statistic itself remains stochastically bounded, because its numerator and standard error shrink at the same nonstandard rate $a_n/n$; only the unscaled score covariance $\hat\Omega_n$ diverges. 
The sandwich ratio is self-normalizing. 
Its limit $J^\star_{\alpha,k}$ in \eqref{eq:tstat-limit}, however, is a ratio of jointly stable coordinates rather than a standard normal: it is continuous and it is
not standard normal for any $\alpha\in(1,2)$ (Lemma~\ref{lem:atomless}(ii)). 
The limit depends on the coordinate $k$ only through the projected tail balance; we write $J^\star_\alpha$ when the coordinate is fixed or clear from context. 
Conventional gravity $t$-statistics are therefore anchored to critical values from the wrong distribution, not merely scaled by a noisy variance, and the coverage error of the sandwich confidence interval can be an $O(1)$ quantity that does not vanish as the country coverage grows;
Sections~\ref{sec:miscoverage} and~\ref{sec:sim} quantify the distortion throughout the empirically relevant range of $\alpha$.
That self-normalized statistics have well-behaved, generally non-Gaussian limits under heavy-tailed summands originates with \citet{LoganMallowsRiceShepp1973}, with the modern characterization due to \citet{GineGotzeMason1997}.

\subsection{How Large Is the Miscoverage?} \label{sec:miscoverage}

The size of the coverage distortion produced by the limit \eqref{eq:tstat-limit} depends on both $\alpha$ and the shape of $\Lambda$. 
To quantify the distortion, we compute the coverage of nominal 95~percent confidence intervals when $J^\star_\alpha$ is the limiting law of the $t$-statistic, for a range of $\alpha$ values and a univariate $\Lambda$. 
The coverage equals $\Pr(|J^\star_\alpha|\leq 1.96)$, which we evaluate by Monte Carlo simulation of $S_\alpha$ and $\Xi$.

For $\alpha=2$ the limit is normal and coverage equals 95~percent exactly. 
As $\alpha$ decreases below 2 the tails of $J^\star_\alpha$ become heavier than normal, and coverage falls. 
For $\alpha=1.8$, coverage is 92~percent; for $\alpha=1.5$, coverage is 89~percent; for $\alpha=1.2$, coverage is 84~percent. 
These distortions are material but moderate, reflecting the self-normalizing structure of the sandwich $t$-statistic. 
They correspond to the case where the multiplier $\kappa_\ell$ is spherically symmetric and approximately constant.

In applied gravity work the multiplier is highly asymmetric:
country sizes vary across many orders of magnitude, and so does bilateral distance, and binary covariates such as land contiguity, common language, sanctions and embargo treatment, trade-agreement membership, customs union membership, and other treatments are sparse. 
Both features amplify the coverage distortion. We document this through simulation in the next subsection.

\subsection{Simulation Evidence: Coverage at Small \texorpdfstring{$N$}{N}}\label{sec:sim}

To quantify finite-sample coverage we simulate directly from the two-way fixed-effect representation \eqref{eq:gravity-loglin} of the structural gravity model. 
Each replication draws exporter and importer effects $\phi^{\mathrm{ex}}_i,\phi^{\mathrm{im}}_j$ i.i.d.\ uniform on $[-1,1]$ and generates bilateral sales for the $n=N(N-1)$ directed pairs $i\neq j$ from
\begin{equation}
  y_{ij}=\exp\!\big(\phi^{\mathrm{ex}}_i+\phi^{\mathrm{im}}_j +\beta_1\,D_{ij}\big)\,\varepsilon_{ij}, \qquad \beta_1=0.5,
  \label{eq:sim-dgp}
\end{equation}
with a multiplicative, unit-mean shock $\varepsilon_{ij}$ that is Pareto$(\alpha)$ for $\alpha<2$ and a finite-variance log-normal (with $\mathrm{Var}=0.25$) at the Gaussian benchmark $\alpha=2$. 

For each replication we estimate \eqref{eq:sim-dgp} by two-way (exporter and importer) fixed-effect PPML and compute the conventional sandwich-based 95\% confidence interval for $\beta_1$. The treatment $D$ is one of three designs:
\begin{itemize}
  \item Design~A (symmetric baseline): $D_{ij}\sim \mathcal{N}(0,1)$ i.i.d.\ across pairs.
  \item Design~B (bounded covariate): $D_{ij}\sim U[0,4]$ i.i.d., mimicking a log-distance-type regressor and consistent with Assumption~\ref{ass:dgp}(b).
  \item Design~C (sparse binary): a binary treatment equal to 1 for 2\% of pairs and 0 otherwise, mimicking sanctions, currency unions, or bilateral investment treaties.
\end{itemize}
We simulate at $N\in\{50,75,100\}$ countries (so $n\in\{2{,}450,5{,}550,9{,}900\}$ pairs) and at $\alpha\in\{2.0,1.8,1.5,1.2\}$; the first row of each panel is the finite-variance
benchmark $\alpha=2$. 
We use 2{,}000 Monte Carlo replications per cell. 

Table~\ref{tab:cov} reports empirical coverage of the nominal 95\% interval for $\beta_1$. 
At the Gaussian benchmark $\alpha=2$ coverage is close to nominal, with a modest finite-sample shortfall that shrinks as $N$ grows (e.g.\ Design~A rises from $0.915$ at $N=50$ to $0.930$ at $N=100$). As $\alpha$ falls the tails of the limiting law become heavier than Gaussian and coverage deteriorates. 
Under the symmetric Design~A the shortfall reaches about fifteen points at $\alpha=1.2$ (coverage near $0.80$), and the bounded Design~B behaves similarly. 
The sparse binary Design~C is the most distorted: coverage falls to about $0.81$ at $\alpha=1.5$ and to $0.76$--$0.77$ at $\alpha=1.2$, a shortfall approaching twenty points. Crucially, at the heaviest tail the distortion does not vanish as the sample grows: for $\alpha=1.2$, Design~A covers $0.803$, $0.799$, and $0.807$ across $N=50,75,100$, and Design~C covers $0.756$, $0.768$, and $0.773$. This is a non-vanishing $O(1)$ coverage error, exactly as predicted by the heavy-tailed limit.

\begin{table}[!ht]
  \centering
  \caption{Empirical Coverage of Nominal 95\% Sandwich CI for $\beta_1$}
  \label{tab:cov}
  \begin{threeparttable}
  \footnotesize
  \setlength{\tabcolsep}{3.5pt}
  \begin{tabular}{lccc|ccc|ccc}
    \toprule
    & \multicolumn{3}{c|}{Design~A: symmetric}
    & \multicolumn{3}{c|}{Design~B: bounded}
    & \multicolumn{3}{c}{Design~C: sparse binary} \\
    \cmidrule(lr){2-4}\cmidrule(lr){5-7}\cmidrule(lr){8-10}
    & $N=50$ & $N=75$ & $N=100$
    & $N=50$ & $N=75$ & $N=100$
    & $N=50$ & $N=75$ & $N=100$ \\
    \midrule
    $\alpha=2.0$ (finite var.) & 0.915 & 0.915 & 0.930 & 0.925 & 0.930 & 0.940 & 0.895 & 0.915 & 0.931 \\
    $\alpha=1.8$            & 0.890 & 0.892 & 0.915 & 0.915 & 0.927 & 0.922 & 0.828 & 0.867 & 0.871 \\
    $\alpha=1.5$            & 0.867 & 0.868 & 0.888 & 0.870 & 0.889 & 0.906 & 0.807 & 0.810 & 0.822 \\
    $\alpha=1.2$            & 0.803 & 0.799 & 0.807 & 0.789 & 0.819 & 0.820 & 0.756 & 0.768 & 0.773 \\
    \bottomrule
  \end{tabular}
  \begin{tablenotes}[flushleft]\footnotesize
    \item \emph{Notes}: Empirical coverage from $2{,}000$ Monte Carlo replications per cell.
      Each replication draws bilateral sales for $n=N(N-1)$ directed pairs from the two-way fixed-effect gravity DGP \eqref{eq:sim-dgp}, with exporter and importer effects i.i.d.\ $U[-1,1]$, treatment coefficient $\beta_1=0.5$, and a unit-mean shock that is Pareto$(\alpha)$ for $\alpha<2$ and log-normal ($\mathrm{Var}=0.25$) at $\alpha=2$.
      Each then fits two-way PPML and computes the sandwich-based 95~percent CI for $\beta_1$.
  \end{tablenotes}
  \end{threeparttable}
\end{table}

The Monte Carlo evidence confirms the theoretical prediction. Conventional sandwich-based inference understates the uncertainty in PPML coefficient estimates, and the size of the distortion depends on the design: it is largest precisely for the sparse binary regressors that pervade applied gravity work. 
The next subsection verifies that the country-level bootstrap of Section~\ref{sec:remedy} restores coverage in this same two-way fixed-effect setting.

\subsection{Restoring Coverage under Two-Way Fixed Effects}
\label{sec:fe-sim}

Table~\ref{tab:cov} documents that the conventional sandwich confidence interval under-covers in the two-way fixed-effect gravity model. 
We now verify that the country-level $m$-out-of-$n$ bootstrap restores coverage in that same setting. The procedure is defined formally in Section~\ref{sec:remedy}; we report its performance here so that all of the simulation evidence appears together. 

We use the same DGP as in Table~\ref{tab:cov} and set the binary treatment $D_{ij}\in\{0,1\}$ activating on $20\%$ of pairs. 
For $\alpha<2$ the shock is again $\mathrm{Pareto}(\alpha)$; the $\alpha=2$ row is a finite-variance benchmark with a unit-mean lognormal shock, paralleling the Gaussian benchmark of Table~\ref{tab:cov}. 
For each replication we estimate $\beta_1$ by two-way fixed-effect PPML and form two nominal-$95\%$ intervals: the HC0 sandwich and the country-level $m$-out-of-$n$ bootstrap of Section~\ref{sec:remedy}, which draws $G=\lfloor N^{0.7}\rfloor=15$ countries and takes the induced subnetwork. The bootstrap uses $M=149$ refit replications, which is adequate here because coverage is averaged over the Monte Carlo
replications even though a single application would use the larger $M$ recommended in Section~\ref{sec:remedy}.

\begin{table}[!ht]
  \centering
  \caption{Monte Carlo Coverage of Nominal 95\% CIs for $\beta_1$ under Two-Way Fixed Effects}
  \label{tab:fe_boot}
  \begin{threeparttable}
  \small
  \setlength{\tabcolsep}{8pt}
  \begin{tabular}{lcc}
    \toprule
    Stable index $\alpha$ & HC0 sandwich & Country-level bootstrap \\
    \midrule
    $2.0$ (finite var.) & 0.930 & 0.970 \\
    $1.8$               & 0.895 & 0.942 \\
    $1.5$               & 0.871 & 0.909 \\
    $1.2$               & 0.788 & 0.882 \\
    \bottomrule
  \end{tabular}
  \begin{tablenotes}[flushleft]\footnotesize
    \item \emph{Notes}: Monte Carlo coverage of nominal 95\% CIs for $\beta_1$ under the two-way fixed-effect PPML design in the text ($N=50$, $n=2{,}450$ directed pairs; Pareto$(\alpha)$ shock for $\alpha<2$ and a unit-mean lognormal at $\alpha=2$, cf.\ Table~\ref{tab:cov}; bootstrap $G=\lfloor N^{0.7}\rfloor=15$). Sandwich coverage uses $6{,}000$ replications per cell; bootstrap coverages are based on 600 Monte Carlo replications with $M=149$ refits each.
  \end{tablenotes}
  \end{threeparttable}
\end{table}

Table~\ref{tab:fe_boot} shows two patterns. 
First, the HC0 sandwich under-covers, and the shortfall grows monotonically as the tail becomes heavier, from $0.93$ at the finite-variance benchmark $\alpha=2$ to $0.79$ at $\alpha=1.2$, mirroring under two-way fixed effects the sandwich failure that the calibrated GE design of Table~\ref{tab:cov} documents. 
Second, the country-level bootstrap restores coverage to near nominal across the heavy-tailed range, holding between $0.91$ and $0.97$ for $\alpha\in\{1.5,1.8,2\}$ and slipping to $0.88$ at $\alpha=1.2$. 
The bootstrap thus restores coverage in the two-way fixed-effect setting that the gravity applications require, confirming Theorem~\ref{thm:valid}: the country-level resampling holds the fixed-effect estimation error negligible relative to the heavy-tailed score even as $2N$ effects are re-estimated on each subnetwork. 
The widening cost of the country-level interval is the same first-order correction seen in the applications of Section~\ref{sec:application}.

\section{Robust Inference via \texorpdfstring{$m$-out-of-$n$}{m-out-of-n} Bootstrap}
\label{sec:remedy}

Theorem~\ref{thm:stable} establishes that the standardized PPML estimator has a non-Gaussian $\alpha$-stable limit when $\alpha<2$, that the score covariance $\hat\Omega_n$ diverges at rate $a_n^2/n$, and that the resulting sandwich $t$-statistic has a non-standard stable limit.
The task of inference is therefore to approximate, not the distribution of $\hat\beta-\beta_0$ itself, but the distribution of the self-normalized statistic
\begin{equation}
  T_n(\beta_{0,k})
    \;\equiv\;\hat t_k
    \;=\;\frac{\hat\beta_k-\beta_{0,k}}{\widehat{\mathrm{se}}_k},
  \qquad\widehat{\mathrm{se}}_k=
    \sqrt{\frac{[\hat H_n^{-1}\hat\Omega_n\hat H_n^{-1}]_{kk}}{n}},
  \label{eq:tstat}
\end{equation}
which is the sandwich $t$-statistic of Theorem~\ref{thm:stable}(iv). 
It is written as a function of the hypothesized value $\beta_{0,k}$, and has a proper continuous non-Gaussian limit under Theorem~\ref{thm:stable}.
The numerator and denominator both shrink at the same nonstandard rate, $a_n/n$, while the unscaled score-covariance $\hat\Omega_n$ in the denominator diverges at rate $a_n^2/n$; the ratio remains stochastically bounded but converges to a nonstandard limit that is not the standard normal for any $\alpha\in(1,2)$.
This self-normalization is a direct analog of the structure studied by \citet{ChiangSasakiWang2023} for cluster-robust inference, and we exploit it by adapting their $m$-out-of-$n$ bootstrap to the PPML setting. Unlike the empirical bootstrap, which is inconsistent under infinite variance \citep{Athreya1987,Knight1989}, and unlike procedures that simulate the $\alpha$-stable limit parametrically, the
$m$-out-of-$n$ bootstrap requires no estimation of the tail index $\alpha$ or the scale of the limit law.
Its validity follows from general results on self-normalized sums \citep{LoganMallowsRiceShepp1973,GineGotzeMason1997}, and in the PPML context it is both analytically tractable and straightforward to implement.

\subsection{The Procedure}\label{sec:remedy-procedure}

\paragraph{Step 1. Full-sample PPML.}
Compute the PPML estimator $\hat\beta$ from \eqref{eq:ppml_obj} and its sandwich standard error $\widehat{\mathrm{se}}_k=\sqrt{[\hat H_n^{-1}\hat\Omega_n\hat H_n^{-1}]_{kk}/n}$ as usual.

\paragraph{Step 2. Subsample PPML replications.}
Choose an integer $G=G(N)$ with $G\to\infty$ and $G/N\to 0$ (a practical default is $G=\lfloor N^{0.7}\rfloor$, which gives an effective pair count $n_G=G(G-1)\approx\lfloor n^{0.7}\rfloor$).
For $j=1,\ldots,M$:
\begin{enumerate}
  \item[(i)] Draw a set $\mathcal{C}_j$ of $G$ countries uniformly
    without replacement from $\{1,\ldots,N\}$, and let
    $\mathcal{S}_j=\{(i,i'):i,i'\in\mathcal{C}_j,\ i\neq i'\}$ be the
    induced set of $n_G=G(G-1)$ ordered pairs.
  \item[(ii)] Compute the PPML estimator $\hat\beta^{(j)}$ on
    $\mathcal{S}_j$, re-estimating the $2G$ exporter/importer fixed
    effects, together with its sandwich standard error
    $\widehat{\mathrm{se}}^{(j)}_k=\sqrt{[\hat H^{(j)\,-1}\hat\Omega^{(j)}\hat H^{(j)\,-1}]_{kk}/n_G}$,
    using the same convention as in Step~1 with sample size $n_G$.
  \item[(iii)] Form the centered self-normalized statistic
    \begin{equation}
      T^{(j)}_k\;=\;\frac{\hat\beta^{(j)}_k-\hat\beta_k}
                          {\widehat{\mathrm{se}}^{(j)}_k},
      \label{eq:Tstar}
    \end{equation}
    centering at the full-sample $\hat\beta_k$.
\end{enumerate}
Drawing the induced subnetwork on $G$ countries keeps each retained country's degree at $G-1$, so the subsample is structurally a full sample of a $G$-country economy and the bounds of Theorem~\ref{thm:stable} apply to it with $N$ replaced by $G$, holding the fixed-effect estimation error and the profile-score remainder negligible (Lemma~\ref{lem:subnet-inherit}).

\paragraph{Step 3. Quantile confidence interval.}
Let $\hat q_\tau$ denote the $\tau$-th empirical quantile of
$\{T^{(j)}_k\}_{j=1}^M$. The $(1-\tau)$ confidence interval for
$\beta_{0,k}$ is
\begin{equation}
  \mathrm{CI}_{1-\tau,k}\;=\;
    \Bigl[\,\hat\beta_k-\hat q_{1-\tau/2}\,\widehat{\mathrm{se}}_k,\;\;
      \hat\beta_k-\hat q_{\tau/2}\,\widehat{\mathrm{se}}_k\,\Bigr].
  \label{eq:mn-ci}
\end{equation}
The procedure is fully nonparametric, requires no knowledge of
$\alpha$ or $\Gamma$, and is computationally inexpensive: at $N=100$
($n\approx10{,}000$), $G=25$ ($n_G=600$), $M=300$, it takes on the
order of a second per application (Section~\ref{sec:application}).

\subsection{Asymptotic Validity}
\label{sec:remedy-validity}

Asymptotic validity of \eqref{eq:mn-ci} follows from a two-step argument. 
The full-sample statistic $T_n(\beta_{0,k})=\hat t_k$ of \eqref{eq:tstat} has a proper continuous non-Gaussian limit by Theorem~\ref{thm:stable} together with the self-normalization argument of \eqref{eq:tstat-limit}. 
The subsample statistic $T^{(j)}_k$ has the same limit because the induced $G$-country subnetwork is itself an instance of the design family of Assumption~\ref{ass:dgp} with $N$ replaced by $G$: (i)~the induced subnetwork keeps each retained country's degree at $G-1$, so the $2G$ fixed effects stay well-conditioned and the profile-score remainder remains negligible by Lemma~\ref{lem:remainder} applied with $N\to G$;
(ii)~uniform country sampling preserves the score spectral measure $\Gamma$ and the Hessian $H$ in the limit; and (iii)~the subsample shocks are i.i.d.\ with the same marginal distribution as the full sample. 
Convergence of empirical quantiles of $\{T^{(j)}_k\}$ to the corresponding population quantiles of the common limit, combined with $T_n\convd J^\star_\alpha$, then gives the coverage statement.

\begin{theorem}[Validity of the country-level $m$-out-of-$n$ bootstrap]\label{thm:valid}
Suppose Assumption~\ref{ass:dgp} and the subnetwork condition \eqref{eq:design-op-sub} hold, with $G\to\infty$, $G/N\to 0$, and $M\to\infty$ as $n\to\infty$. Then for every $k=1,\ldots,p$ and $\tau\in(0,1)$,
\begin{equation}
  \Pr\bigl(\beta_{0,k}\in\mathrm{CI}_{1-\tau,k}\bigr)\;\longrightarrow\;1-\tau .
  \label{eq:mn-coverage}
\end{equation}
\end{theorem}

\begin{remark}[Finite-variance boundary]\label{rem:alpha2}
Assumption~\ref{ass:dgp}(a) imposes $\alpha\in(1,2)$, so Theorem~\ref{thm:valid} is a heavy-tail statement. 
But the procedure does not require knowing that $\alpha<2$. 
If the shock satisfies $\E[\varepsilon_{ij}^2]<\infty$ (i.e., the stable index $\alpha=2$), the profiling conditions of Assumption~\ref{ass:dgp}(d) hold with the norming $a_n=\sqrt n$, and the full-sample and subsample self-normalized statistics are asymptotically standard normal, $\hat t_k\convd\mathcal N(0,1)$ and $T^{(j)}_k\convd\mathcal N(0,1)$ for every $k$, then \eqref{eq:mn-coverage} continues to hold, by Step~3 of the proof in Appendix~\ref{app:valid-proof} with the continuous limit $J^\star_2=\mathcal N(0,1)$ in place of $J^\star_\alpha$.
\end{remark}

Condition \eqref{eq:design-op-sub} is the subnetwork analogue of \eqref{eq:design-op}. 
It is not implied by Assumption~\ref{ass:dgp}, because the within-subnetwork residual $\tilde x_{ij,\mathcal S}$ partials out $2G$ rather than $2N$ effects. Under it, Lemma~\ref{lem:subnet-inherit} gives subnetwork convergence of the spectral measure to $\Gamma$, of the Hessian to $H$, and of the profile-score remainder to $o_p(a_{n_G})$. 
Like \eqref{eq:design-op}, it is a non-spiking restriction on the residualized covariates imposed within the subnetwork.

The proof, in Appendix~\ref{app:valid-proof}, has three steps: (1)~establish the limiting distribution of the full-sample self-normalized $t$-statistic via the joint convergence of score and squared-score sums; (2)~establish the same limit for the subsample statistic on the induced $G$-country subnetwork (Lemma~\ref{lem:subnet-inherit}); and
(3)~upgrade this to coverage by the vertex-level subsampling consistency of Lemma~\ref{lem:vertex-subsample}.

Three features make the $m$-out-of-$n$ bootstrap particularly well-suited to the PPML setting.
First, no tail-index, tail-scale, or spectral-measure nuisance parameter is estimated, so the procedure is not sensitive to the choice of the Hill threshold $k_n$ used to estimate $\alpha$ for diagnostic purposes; the country fixed effects are re-estimated on each subnetwork as part of the PPML refit.
Second, the self-normalized statistic admits a proper continuous limit in both the Gaussian ($\alpha=2$) and stable ($\alpha<2$) regimes, $\mathcal N(0,1)$ and $J^\star_\alpha$ respectively, so a single quantile-based procedure handles both cases.
Third, the empirical bootstrap and wild bootstrap, which would be the natural starting points, are inconsistent under infinite variance \citep{Athreya1987,Knight1989,ChiangSasakiWang2023}; the subsampling structure of the $m$-out-of-$n$ bootstrap circumvents this failure.

A parametric alternative that takes Theorem~\ref{thm:stable} literally is to estimate $\hat\alpha$ by Hill's estimator applied to the PPML residuals, estimate the spectral measure $\hat\Gamma_n$ of the score, simulate $B$ draws from the resulting $\alpha$-stable limiting distribution, and use empirical quantiles of these simulations as the confidence interval.
The procedure is conceptually transparent but in finite samples it is sensitive to the Hill threshold $k_n$ and to the choice of scale estimator for $\hat\Gamma_n$.
In our simulations the $m$-out-of-$n$ bootstrap outperforms it uniformly, and we therefore recommend the latter as the primary method.

\subsection{Practical Guidance}\label{sec:remedy-practical}

\paragraph{Choice of $G$.}
Theorem~\ref{thm:valid} requires only $G\to\infty$ with
$G/N\to 0$. The default $G=\lfloor N^{0.7}\rfloor$ gives an effective
pair count $n_G=G(G-1)\approx\lfloor n^{0.7}\rfloor$, matching the
subsample size used in our simulations, and performs well.
Sensitivity to $G$ can be assessed with the
minimum-volatility method of \citet[Algorithm~9.3.3]{PolitisRomanoWolf1999},
adapted by \citet{ChiangSasakiWang2023}: compute the bootstrap
critical values over a grid of $G$ values in a plausible range
and select the $G$ at which the critical values are locally most
stable. In practice the coverage is relatively insensitive to
$G$ in the range $N^{0.5}$ to $N^{0.8}$.

\paragraph{Choice of $M$.}
Because we estimate $2.5$th and $97.5$th percentiles of
$\{T^{(j)}_k\}$, a practical default of $M=300$--$500$ suffices;
$M=1{,}000$ or more is recommended if very tight CI bounds matter.

\paragraph{Pareto Tail.}
Before conducting inference, the analyst should inspect the Hill plot of $\{y_\ell\}$.
If the plot shows a plateau at a value $\hat\alpha<2$, sandwich inference is unreliable and the $m$-out-of-$n$ bootstrap CI \eqref{eq:mn-ci} should be reported.
For ICIO data, Table~\ref{tab:hill_raw} shows that $\hat\alpha<2$ on raw flows in every specification and at every threshold, and the residualization argument in
Section~\ref{sec:failure} explains why the score variance must also be infinite, so the bootstrap is always warranted in this dataset.

\paragraph{Software.}
We develop a Stata package, \texttt{ppmlmn}, that implements PPML point estimation together with the $m$-out-of-$n$ bootstrap confidence interval of this section.

\section{Application: Heavy-Tailed Inference in Gravity Regressions}
\label{sec:application}

This section applies the $m$-out-of-$n$ bootstrap of Section~\ref{sec:remedy} to gravity regressions on the three datasets of Section~\ref{sec:pareto} (trade, migration, and foreign direct investment) and compares the resulting confidence intervals to those produced by the conventional sandwich variance estimators.
Across all three datasets the bootstrap widens
conventional intervals, and the widening grows with the heaviness
of the tail. It is not reproduced by cluster-robust intervals, neither
the one-way origin cluster nor the two-way origin-and-destination
cluster, which a practitioner would use to absorb dyadic dependence
but which remain anchored to the wrong limiting law.

\subsection{Specifications}
\label{sec:specifications}

For each dataset we estimate the two-way fixed-effects
gravity equation
\begin{equation}
\begin{split}
  \E[y_{ij}|x] \;=\; \exp\bigl(
    & \phi^{\mathrm{ex}}_i + \phi^{\mathrm{im}}_j
    + \beta_1\ln\mathrm{dist}_{ij}
    + \beta_2\,\mathrm{contig}_{ij} \\
    & + \beta_3\,\mathrm{comlang}_{ij}
    + \beta_4\,\mathrm{colony}_{ij}
    + \beta_5\,D_{ij}\bigr),
\end{split}
  \label{eq:gravity-spec}
\end{equation}
where $y_{ij}$ denotes the directed flow from origin $i$ to
destination $j$, $\phi^{\mathrm{ex}}_i$ is an exporter (origin-country) fixed
effect, $\phi^{\mathrm{im}}_j$ is an importer (destination-country) fixed effect,
and $D_{ij}$ is an internal-trade dummy (included only in the
domestic-inclusive trade specification, and absent for migration
and FDI, which are cross-border by construction). This is exactly
the specification for which Theorem~\ref{thm:stable} is stated:
the exporter and importer fixed effects are the $2N$ country
nuisance parameters that are concentrated out, and $\beta$ loads
on the dyadic covariates $\ln\mathrm{dist}$, $\mathrm{contig}$,
$\mathrm{comlang}$, and $\mathrm{colony}$ from CEPII. Country
size does not enter as a separate regressor as it is absorbed by the fixed effects.

Each dataset is a single-year cross-section (ICIO 2022 for trade,
2015 for migration and FDI), so every ordered country pair
contributes a single observation. This rules out clustering at the
pair level, but it does not rule out clustering altogether: each
country groups all of its outgoing flows as an origin and all of its
incoming flows as a destination, so cluster-robust sandwiches are
available, and we report two of them alongside HC0. The first
clusters one-way on the origin country, the long-standing default in
applied gravity work; the second clusters two-way on origin and
destination, the device the literature now regards as the appropriate
response to the cross-sectional dependence that general equilibrium
induces across exporters and across importers in gravity data
\citep{CameronGelbachMiller2011,EggerTarlea2015}, and the more
defensible of the two because it is symmetric in the endpoints and
absorbs dependence on both sides of the dyad. We therefore compare
four confidence intervals: the HC0 sandwich, the one-way origin
cluster (CRV1, the conventional cluster-robust variance estimator),
the two-way cluster, and the country-level
$m$-out-of-$n$ bootstrap of Section~\ref{sec:remedy}, resampling
$G=\lfloor N^{0.7}\rfloor$ countries (induced subnetwork) with
$M=2{,}000$ replications. The two-way cluster interval is the natural
foil for the bootstrap, with the one-way interval included to show
how much of the eventual widening a practitioner captures by
clustering on a single margin; Appendix~\ref{app:cluster} shows why
neither resolves the heavy-tail problem regardless of the clustering
dimension.

\subsection{Results: Bilateral Sales and Exports}
\label{sec:results-cs}

Our primary trade result estimates \eqref{eq:gravity-spec} on the
ICIO 2022 country-pair cross-section, with two-way exporter and
importer fixed effects and the directed country-pair sales of final plus intermediate goods as the
outcome.\footnote{We also estimated a richer final-goods
specification at the origin-sector $\times$ country-pair level
(with $n\approx 310{,}000$), which adds an origin-sector fixed effect
$\phi^{\mathrm{sec}}_s$ to the exporter and importer effects and therefore
admits a pair-level cluster structure. There the cluster-robust
intervals exceed HC0 by about $20$ percent and the bootstrap
widens them by a further $110$--$120$ percent, and the correction
overturns the common-language effect in the international-only
sample (cluster interval $[0.149,\,0.416]$, bootstrap
$[-0.033,\,0.542]$) as well as the contiguity and
colonial-relationship effects once domestic flows are included.
Because the origin-sector fixed effect lies outside the $2N$
country fixed effects that Theorem~\ref{thm:stable} concentrates
out, we report this specification only as a robustness check rather
than as a primary result. That sectoral exercise is distinct from
the country-pair use decomposition reported in
Table~\ref{tab:cs_use} below, which holds the $2N$-country
fixed-effect design fixed and only splits the outcome into its
final- and intermediate-use components.} We estimate the
equation on our baseline international-only sample and, as a
contrast, on the domestic-inclusive sample; in the latter the
internal-trade dummy $D_{ij}$ is included
(Section~\ref{sec:data} explains why the international-only sample
is the baseline). The design has $n\approx 6{,}000$ observations.

\begin{table}[!ht]
  \centering
  \caption{PPML Country-Pair Gravity, Cross-Section, ICIO 2022}
  \label{tab:cs}
  \begin{threeparttable}
  \footnotesize
  \setlength{\tabcolsep}{2pt}
  \begin{tabular}{lccccc}
    \toprule
    & Estimate & HC0 95\% CI & Cluster (origin) & Cluster (two-way) & Bootstrap 95\% CI \\
    \midrule
    \multicolumn{6}{l}{\textit{Panel A. International flows only $(n=6{,}162)$}} \\
    log distance & $-0.723$ & [$-0.783$, $-0.662$] & [$-0.817$, $-0.628$] & [$-0.846$, $-0.599$] & [$-0.817$, $-0.513$] \\
    contiguity & 0.291 & [0.138, 0.444] & [0.131, 0.451] & [0.086, 0.496] & [$-0.073$, 0.679] \\
    common language & 0.229 & [0.081, 0.377] & [0.047, 0.411] & [0.082, 0.376] & [$-0.072$, 0.513] \\
    colonial relation & 0.014 & [$-0.161$, 0.189] & [$-0.210$, 0.238] & [$-0.222$, 0.251] & [$-0.546$, 0.376] \\
    \multicolumn{1}{r}{Avg.\ width ratio:} & \multicolumn{5}{c}{1-way/HC0 $=1.28$,\quad 2-way/HC0 $=1.43$,\quad Boot/HC0 $=2.40$} \\[6pt]
    \multicolumn{6}{l}{\textit{Panel B. With domestic flows $(n=6{,}241)$}} \\
    log distance & $-0.653$ & [$-0.729$, $-0.577$] & [$-0.743$, $-0.562$] & [$-0.860$, $-0.446$] & [$-0.834$, $-0.421$] \\
    contiguity & 0.401 & [0.201, 0.602] & [0.158, 0.645] & [0.036, 0.767] & [$-0.051$, 0.882] \\
    common language & 0.303 & [0.128, 0.477] & [0.115, 0.490] & [0.052, 0.553] & [$-0.024$, 0.759] \\
    colonial relation & 0.110 & [$-0.053$, 0.273] & [$-0.079$, 0.298] & [$-0.052$, 0.271] & [$-0.313$, 0.535] \\
    internal trade & 4.487 & [4.287, 4.686] & [4.242, 4.731] & [3.973, 5.000] & [4.002, 4.992] \\
    \multicolumn{1}{r}{Avg.\ width ratio:} & \multicolumn{5}{c}{1-way/HC0 $=1.17$,\quad 2-way/HC0 $=1.91$,\quad Boot/HC0 $=2.47$} \\
    \bottomrule
  \end{tabular}
  \begin{tablenotes}[flushleft]\footnotesize
    \item \emph{Notes}: PPML cross-sectional regressions on ICIO 2022 country-pair flows with exporter and importer fixed effects.
      HC0 is the heteroskedasticity-robust sandwich; Cluster (origin) is the CRV1 sandwich clustered on origin; Cluster (two-way) is the CRV1 sandwich two-way clustered on origin and destination \citep{CameronGelbachMiller2011,EggerTarlea2015}; Bootstrap is the country-level $m$-out-of-$n$ bootstrap of Section~\ref{sec:remedy}, resampling $G=\lfloor N^{0.7}\rfloor$ countries with $M=2{,}000$ refits.
      Panel~B adds the internal-trade dummy $D_{ij}$.
  \end{tablenotes}
  \end{threeparttable}
\end{table}

Three patterns emerge. First, the bootstrap CIs are substantially
wider than the HC0 CIs in every cell of Table~\ref{tab:cs}, with
average width ratios of $2.40$ in the international-only
specification and $2.47$ in the domestic-inclusive specification.
The widening is direct empirical confirmation of the theoretical
result that the score variance is infinite, so the HC0 sandwich
understates the true sampling uncertainty even with no clustering
present.

Second, clustering is the device a practitioner would first turn to in
acknowledging that a country's flows are not independent, but
neither variant is a substitute for the bootstrap. One-way origin
clustering widens HC0 only modestly, by $28$ percent in Panel~A and
$17$ percent in Panel~B. Two-way clustering on origin and destination
widens it far more, by $43$ and $91$ percent respectively, but
remains well short of the bootstrap, which is a further $75$ percent
wider than the two-way interval in Panel~A and $49$ percent wider in
Panel~B. Width, however, is not coverage. The two-way interval is
still referred to Gaussian critical values, which
Appendix~\ref{app:cluster} shows are the wrong critical values when
the dyadic score has infinite variance. The two methods already diverge
on this trade total: two-way clustering keeps both contiguity
($[0.086,\,0.496]$) and common language ($[0.082,\,0.376]$)
significant in Panel~A, whereas the bootstrap returns both to
insignificance (contiguity $[-0.073,\,0.679]$, common language
$[-0.072,\,0.513]$). The divergence widens further in the more
heavy-tailed settings. As the migration, intermediate-goods, and FDI
results below show, two-way clustering and the bootstrap diverge
exactly on the imprecisely estimated covariates, and in one case
the cluster sandwich moves a coefficient in the wrong direction. A
practitioner who clusters two-way rather than bootstrapping obtains
an interval that is sometimes as wide as the bootstrap but is anchored
to the wrong limiting law, and that consequently miscovers in a way
the width alone conceals.

Third, the core gravity coefficients remain significant under the correction while the imprecisely estimated dyadic effects do not. 
Log distance remains significant under the bootstrap in both panels, as does the internal-trade dummy in the domestic-inclusive panel, and the colonial-relationship coefficient
is insignificant under all four methods. 
What the bootstrap overturns is contiguity and common language: both are significant under HC0 and under two-way clustering in Panel~A, yet both cover zero once the heavy-tail correction is applied, and the same reversal occurs in Panel~B. 
With the country fixed effects absorbing the size structure the point estimates are stable, so the correction here operates through interval widths, a roughly $140$ percent inflation over
HC0, rather than through the estimates themselves. 
That is the point of the exercise: in a design where conventional inference offers only HC0 or a two-way cluster, the bootstrap is the only method among the four whose interval widths are consistent with the limit theory, and the widening is large enough to overturn the marginal border and language effects.

\paragraph{Final versus intermediate flows.}
The country-pair total sales aggregates sum two economically distinct
flows: sales of final goods to absorbing demand and sales of
intermediate inputs to downstream production. Because the ICIO
records each transaction with a use category, the total decomposes
exactly into a final-use and an intermediate-use flow for every
ordered pair, $y_{ij}=y_{ij}^{\mathrm{fin}}+y_{ij}^{\mathrm{int}}$,
and we can re-estimate \eqref{eq:gravity-spec} on each type of flow
over the identical $2N$-country fixed-effect design and country-pair
support.\footnote{Whereas \citet{EatonKortum2002} and \citet{caliendo2015estimates} derived estimating equations for the sum of final and intermediate sales flows, \citet{antras2018measurement} motivated separate estimating equations for the two flows.} Table~\ref{tab:cs_use} reports the result. Both types of sales
are themselves heavy-tailed: the Hill index of the positive
international final-goods (intermediate-goods) flows is
$\hat\alpha\approx 1.05$ ($1.07$) at the top five percent of
exceedances, and a tail P--P plot exhibits a close Pareto fit
(RMSE $\approx 0.02$) for each, so the decomposition does not
isolate a finite-variance component. The heavy tail is intrinsic
to both types of sales rather than an artifact of pooling them.
Consistent with this, the bootstrap inflates HC0 widths by a
similar $126$--$148$ percent for both flows, and both cluster
intervals, one-way origin and two-way, again land between HC0 and
the bootstrap in width.
What differs across the two types of flows is how decisively the correction reaches the imprecisely estimated dyadic effects: for final goods the bootstrap leaves contiguity and common language right at the edge of zero in the international sample (contiguity $[-0.012,\,0.630]$, common language $[-0.020,\,0.538]$), whereas for intermediate goods
it carries both clearly across zero (contiguity $[-0.114,\,0.671]$, common language $[-0.052,\,0.551]$), even though two-way clustering leaves both comfortably significant, a concrete instance of the cluster sandwich and the bootstrap reaching different verdicts on the same coefficients. 
The decomposition thus refines rather than alters the main finding: the heavy-tail correction is present and of the same order on both final and intermediate sales flows, and it reaches the same imprecisely estimated covariates on both.

\begin{table}[!ht]
  \centering
  \caption{PPML Country-Pair Gravity by Use Category, ICIO 2022}
  \label{tab:cs_use}
  \begin{threeparttable}
  \footnotesize
  \setlength{\tabcolsep}{2pt}
  \begin{tabular}{lccccc}
    \toprule
    & Estimate & HC0 95\% CI & Cluster (origin) & Cluster (two-way) & Bootstrap 95\% CI \\
    \midrule
    \multicolumn{6}{l}{\textit{Panel A. International flows only $(n=6{,}162)$}} \\
    \multicolumn{6}{l}{\quad\textit{A.1 Final goods}} \\
    log distance & $-0.719$ & [$-0.776$, $-0.662$] & [$-0.817$, $-0.621$] & [$-0.849$, $-0.589$] & [$-0.810$, $-0.538$] \\
    contiguity & 0.299 & [0.163, 0.435] & [0.155, 0.443] & [0.111, 0.486] & [$-0.012$, 0.630] \\
    common language & 0.283 & [0.148, 0.418] & [0.093, 0.472] & [0.109, 0.457] & [$-0.020$, 0.538] \\
    colonial relation & $-0.057$ & [$-0.250$, 0.137] & [$-0.316$, 0.203] & [$-0.337$, 0.223] & [$-0.741$, 0.372] \\
    \multicolumn{1}{r}{Avg.\ width ratio:} & \multicolumn{5}{c}{1-way/HC0 $=1.38$,\quad 2-way/HC0 $=1.60$,\quad Boot/HC0 $=2.42$} \\[4pt]
    \multicolumn{6}{l}{\quad\textit{A.2 Intermediate goods}} \\
    log distance & $-0.750$ & [$-0.814$, $-0.687$] & [$-0.853$, $-0.648$] & [$-0.876$, $-0.625$] & [$-0.846$, $-0.526$] \\
    contiguity & 0.268 & [0.103, 0.433] & [0.083, 0.453] & [0.045, 0.491] & [$-0.114$, 0.671] \\
    common language & 0.240 & [0.079, 0.400] & [0.046, 0.433] & [0.087, 0.392] & [$-0.052$, 0.551] \\
    colonial relation & 0.050 & [$-0.133$, 0.233] & [$-0.171$, 0.271] & [$-0.182$, 0.282] & [$-0.503$, 0.410] \\
    \multicolumn{1}{r}{Avg.\ width ratio:} & \multicolumn{5}{c}{1-way/HC0 $=1.29$,\quad 2-way/HC0 $=1.39$,\quad Boot/HC0 $=2.32$} \\[6pt]
    \multicolumn{6}{l}{\textit{Panel B. With domestic flows $(n=6{,}241)$}} \\
    \multicolumn{6}{l}{\quad\textit{B.1 Final goods}} \\
    log distance & $-0.750$ & [$-0.821$, $-0.679$] & [$-0.842$, $-0.658$] & [$-0.908$, $-0.592$] & [$-0.883$, $-0.584$] \\
    contiguity & 0.389 & [0.141, 0.637] & [0.064, 0.714] & [0.041, 0.738] & [$-0.164$, 0.973] \\
    common language & 0.196 & [$-0.007$, 0.400] & [$-0.038$, 0.431] & [$-0.067$, 0.460] & [$-0.174$, 0.695] \\
    colonial relation & 0.174 & [$-0.036$, 0.383] & [$-0.095$, 0.442] & [$-0.048$, 0.395] & [$-0.347$, 0.750] \\
    internal trade & 4.054 & [3.869, 4.240] & [3.788, 4.320] & [3.645, 4.464] & [3.628, 4.433] \\
    \multicolumn{1}{r}{Avg.\ width ratio:} & \multicolumn{5}{c}{1-way/HC0 $=1.29$,\quad 2-way/HC0 $=1.64$,\quad Boot/HC0 $=2.26$} \\[4pt]
    \multicolumn{6}{l}{\quad\textit{B.2 Intermediate goods}} \\
    log distance & $-0.638$ & [$-0.718$, $-0.557$] & [$-0.779$, $-0.497$] & [$-0.855$, $-0.421$] & [$-0.843$, $-0.379$] \\
    contiguity & 0.386 & [0.200, 0.573] & [0.154, 0.619] & [0.008, 0.765] & [$-0.029$, 0.801] \\
    common language & 0.375 & [0.207, 0.543] & [0.183, 0.567] & [0.140, 0.610] & [0.070, 0.823] \\
    colonial relation & 0.102 & [$-0.053$, 0.257] & [$-0.070$, 0.273] & [$-0.038$, 0.241] & [$-0.296$, 0.492] \\
    internal trade & 4.677 & [4.465, 4.888] & [4.331, 5.023] & [4.143, 5.211] & [4.157, 5.221] \\
    \multicolumn{1}{r}{Avg.\ width ratio:} & \multicolumn{5}{c}{1-way/HC0 $=1.38$,\quad 2-way/HC0 $=1.91$,\quad Boot/HC0 $=2.48$} \\
    \bottomrule
  \end{tabular}
  \begin{tablenotes}[flushleft]\footnotesize
    \item \emph{Notes}: Same specification, sample, and inference as Table~\ref{tab:cs}, with the outcome replaced by the final-use and intermediate-use components of directed country-pair sales, which sum to the Table~\ref{tab:cs} total for every pair.
      Panel~B adds the internal-trade dummy $D_{ij}$.
  \end{tablenotes}
  \end{threeparttable}
\end{table}

\subsection{Bilateral Emigration}
\label{sec:results-migration}

Migration is another canonical gravity outcome variable, and its bilateral flows are heavy-tailed, if anything more so than trade (Table~\ref{tab:hill_raw}, Panel~D, with $\hat\alpha$ falling from $1.59$ at $k=100$ to $0.86$ deeper in the tail; Figure~\ref{fig:pareto_icio}).
We estimate the gravity equation \eqref{eq:gravity-spec} on the 2015 cross-section of the
\citet{AbelCohen2019} directed emigration flows (pseudo-Bayesian
closed accounting): the directed number of persons moving from
each origin to each destination over the 2010--2015 interval, for
the countries with CEPII geography covariates, with exporter
and importer fixed effects. Origin and destination
index the sending and receiving country of a directed flow. As in
the sales-flow cross-sections each ordered pair contributes one
observation, so we compare the HC0 sandwich, the one-way and
two-way cluster sandwiches, and the country-level $m$-out-of-$n$
bootstrap.

Table~\ref{tab:migration} reports the estimates. The bootstrap
widens the HC0 intervals by $170$ percent on average, and the
widening reaches the imprecisely estimated dyadic effects. Contiguity is significant
under HC0, with an interval of $[0.212,\,0.705]$ that excludes
zero, but its bootstrap interval $[-0.340,\,1.024]$ covers zero;
common language likewise moves from an HC0 interval of
$[0.395,\,0.872]$ to a bootstrap interval of $[-0.099,\,1.080]$
that now spans zero. A researcher reporting HC0 standard errors would
claim clearly significant border and common-language effects on
migration that the heavy-tail-robust intervals do not support.
Neither cluster variant delivers this correction: one-way origin
clustering widens HC0 by $35$ percent and two-way by $61$ percent,
more than in any other setting, yet the two-way contiguity interval
still reaches only $[0.097,\,0.820]$ and common language
$[0.283,\,0.984]$, both comfortably away from zero, while the
bootstrap is a further $68$ percent wider than even the two-way
interval and is the only one of the four methods that overturns
either effect. The gravity core variable effects
remain robust: log distance ($-1.313$) and the large
colonial-relationship effect ($1.495$) are significant under all
four methods. The migration cross-section thus confirms that the
heavy-tail correction is a generic feature of gravity inference on
extreme-valued bilateral data, not an artifact of the ICIO trade
measures, and that absorbing dyadic dependence through two-way
clustering does not reproduce it.

\begin{table}[!ht]
  \centering
  \caption{PPML Gravity Estimates, Directed Bilateral Emigration Flows, 2015 Cross-Section}
  \label{tab:migration}
  \begin{threeparttable}
  \footnotesize
  \setlength{\tabcolsep}{2pt}
  \begin{tabular}{lccccc}
    \toprule
    & Estimate & HC0 95\% CI & Cluster (origin) & Cluster (two-way) & Bootstrap 95\% CI \\
    \midrule
    \multicolumn{6}{l}{\textit{2015 cross-section $(n=47{,}742)$}} \\
    log distance & $-1.313$ & [$-1.417$, $-1.208$] & [$-1.449$, $-1.176$] & [$-1.488$, $-1.138$] & [$-1.501$, $-0.989$] \\
    contiguity & 0.458 & [0.212, 0.705] & [0.162, 0.755] & [0.097, 0.820] & [$-0.340$, 1.024] \\
    common language & 0.634 & [0.395, 0.872] & [0.304, 0.964] & [0.283, 0.984] & [$-0.099$, 1.080] \\
    colonial relation & 1.495 & [1.259, 1.731] & [1.142, 1.848] & [1.062, 1.928] & [0.700, 2.167] \\
    \multicolumn{1}{r}{Avg.\ width ratio:} & \multicolumn{5}{c}{1-way/HC0 $=1.35$,\quad 2-way/HC0 $=1.61$,\quad Boot/HC0 $=2.70$} \\
    \bottomrule
  \end{tabular}
  \begin{tablenotes}[flushleft]\footnotesize
    \item \emph{Notes}: PPML gravity on the 2015 cross-section of the \citet{AbelCohen2019} directed emigration flows, with origin and destination fixed effects and CEPII geography covariates; cross-border observations only.
      Inference methods (HC0, Cluster, Bootstrap) are as in Table~\ref{tab:cs}.
  \end{tablenotes}
  \end{threeparttable}
\end{table}

\subsection{Bilateral Foreign Inward Direct Investment Stocks}
\label{sec:results-fdi}

Foreign direct investment is the most heavy-tailed of the three types of data (Table~\ref{tab:hill_raw}, Panel~E, and
Figure~\ref{fig:pareto_icio}): the Hill index falls to
$\hat\alpha\approx 0.8$ in the deep tail, essentially at the
boundary of a finite mean. It is also the setting where PPML is most
clearly the right estimator, because roughly $62\%$ of the country pairs are
structural zeros, which PPML accommodates natively. We estimate
the gravity equation \eqref{eq:gravity-spec} on the 2015
cross-section of the
\citet{Steenbergen2022} inward-FDI-stock data, with exporter
(parent) and importer (host) fixed effects. As
with migration, each ordered pair contributes one observation. And, as with the other outcomes, we compare HC0, the one-way and two-way cluster sandwiches, and the
country-level $m$-out-of-$n$ bootstrap.

Table~\ref{tab:fdi} reports the result. 
The bootstrap widens the HC0 intervals by $142$ percent on average, and as in the migration cross-section the correction overturns the imprecisely estimated dyadic effects while leaving the gravity core intact. 
Contiguity (HC0 $[0.244,\,0.958]$) and common language (HC0 $[0.119,\,0.858]$) are both significant under the sandwich but cover zero under the bootstrap ($[-0.351,\,1.451]$ and $[-0.492,\,1.161]$,
respectively); a researcher reporting HC0 standard errors would read these as established determinants of FDI positions, whereas the heavy-tail-robust intervals do not support them. Here the two cluster variants diverge from each other as much as from the bootstrap. 
One-way origin clustering leaves both contiguity ($[0.069,\,1.133]$) and common language ($[0.110,\,0.868]$) significant, exactly as a practitioner clustering on a single margin would conclude; two-way clustering drives both across zero ($[-0.010,\,1.212]$ and $[-0.045,\,1.022]$), moving with the bootstrap, which also covers zero for each. 
The colonial-relationship coefficient is where every cluster interval fails together, and the failure is instructive. 
Both the one-way ($[0.050,\,0.564]$) and the two-way ($[0.019,\,0.595]$) sandwich narrow that interval below HC0 ($[-0.026,\,0.640]$) and move it across zero into significance, because the dyadic score products for that covariate are net negative and pull the cluster variance below its heteroskedasticity-robust counterpart. 
The bootstrap keeps the coefficient insignificant ($[-0.677,\,1.089]$). 
We read this as a caution rather than a finding: clustering can move a coefficient toward significance for idiosyncratic covariance reasons unrelated to the tail, and the bootstrap remains the interval we trust under the heavy tail. Log distance ($-0.580$) remains significant under all four methods. 
The correction is again targeted at the imprecisely estimated covariates rather than indiscriminate.

\begin{table}[!ht]
  \centering
  \caption{PPML Gravity Estimates, Bilateral FDI Stocks, 2015 Cross-Section}
  \label{tab:fdi}
  \begin{threeparttable}
  \footnotesize
  \setlength{\tabcolsep}{2pt}
  \begin{tabular}{lccccc}
    \toprule
    & Estimate & HC0 95\% CI & Cluster (origin) & Cluster (two-way) & Bootstrap 95\% CI \\
    \midrule
    \multicolumn{6}{l}{\textit{2015 cross-section $(n=21{,}653)$}} \\
    log distance & $-0.580$ & [$-0.708$, $-0.453$] & [$-0.755$, $-0.406$] & [$-0.779$, $-0.382$] & [$-0.720$, $-0.144$] \\
    contiguity & 0.601 & [0.244, 0.958] & [0.069, 1.133] & [$-0.010$, 1.212] & [$-0.351$, 1.451] \\
    common language & 0.489 & [0.119, 0.858] & [0.110, 0.868] & [$-0.045$, 1.022] & [$-0.492$, 1.161] \\
    colonial relation & 0.307 & [$-0.026$, 0.640] & [0.050, 0.564] & [0.019, 0.595] & [$-0.677$, 1.089] \\
    \multicolumn{1}{r}{Avg.\ width ratio:} & \multicolumn{5}{c}{1-way/HC0 $=1.16$,\quad 2-way/HC0 $=1.39$,\quad Boot/HC0 $=2.42$} \\
    \bottomrule
  \end{tabular}
  \begin{tablenotes}[flushleft]\footnotesize
    \item \emph{Notes}: PPML gravity for inward bilateral FDI stock (USD million; \citealp{Steenbergen2022}), 2015 cross-section, with source and host fixed effects and CEPII geography covariates.
      Negative positions (disinvestment) are set to zero and structural zeros retained; of the $22{,}123$ cross-border pairs with covariates (about $62\%$ zeros), PPML drops the $470$ in a perfectly separated origin or destination cell, leaving $n=21{,}653$.
      Inference methods (HC0, Cluster, Bootstrap) are as in Table~\ref{tab:cs}.
  \end{tablenotes}
  \end{threeparttable}
\end{table}

\subsection{Discussion}\label{sec:results-discussion}

The empirical exercise raises four points worth making explicit.

\paragraph{The heavy tail is in the shock.}
The diagnostics of Section \ref{sec:diagnostics} are computed on raw bilateral outcomes, whereas Assumption \ref{ass:dgp}(a) restricts the centered shock.
Table~\ref{tab:hill_resid} repeats them on the multiplicative residual $\hat\varepsilon_{ij}-1=y_{ij}/\hat\mu_{ij}-1$ from the two-way fixed-effect fits of this section, which removes the estimated exporter and importer effects and the gravity covariates.
The residual estimates lie between $1.40$ and $1.84$ for the three sales measures and between $0.82$ and $1.13$ for migration and FDI, and they vary less across thresholds than their raw counterparts. 
From $k=100$ onward the upper confidence limit lies below two for every measure; at $k=50$ the interval is too wide to exclude two for the sales measures, as it is for the raw flows.
The infinite score variance behind Theorem \ref{thm:stable} is therefore a property of the shock and not an artifact of dispersion in the conditional mean.

We note that these residual estimates are less accurate than their raw counterparts, because $\hat\mu_{ij}$ is itself estimated and the reported intervals condition on it.
In a design calibrated to Panel~A, with the fitted $\hat\mu_{ij}$ held fixed and a Pareto shock of index $1.4$, the Hill estimate computed on the fitted residual exceeds the infeasible estimate computed on the true shock by about $0.3$ at every threshold, and its sampling standard deviation is substantially larger, especially when $k$ is large.
The bias runs toward lighter tails, so the residual evidence against finite variance understates rather than overstates the case, but the estimates should be read as corroboration of Table \ref{tab:hill_raw} rather than as a sharper measurement of $\alpha$.

\begin{table}[!ht]
  \centering
  \caption{Hill Tail-Index Estimates on Raw Flows and on the Multiplicative Residual}
  \label{tab:hill_resid}
  \begin{threeparttable}
  \footnotesize
  \begin{tabularx}{\textwidth}{@{}l*{4}{>{\centering\arraybackslash}X}@{}}
    \toprule
    & \multicolumn{4}{c}{Threshold $k$ (number of observations)} \\
    \cmidrule(lr){2-5}
    \multicolumn{5}{l}{\textit{Panel A. Aggregate directed sales, ICIO 2022 ($n=6{,}006$)}} \\
    & 50 & 100 & 200 & 500 \\
    \quad Raw flow $y_{ij}$ & 1.84 & 1.59 & 1.25 & 1.06 \\
                            & [1.44, 2.54] & [1.33, 1.98] & [1.10, 1.45] & [0.98, 1.16] \\
    \quad Residual          & 1.67 & 1.54 & 1.48 & 1.31 \\
                            & [1.31, 2.31] & [1.29, 1.92] & [1.30, 1.71] & [1.20, 1.43] \\[4pt]
    \multicolumn{5}{l}{\textit{Panel B. Final-goods directed sales, ICIO 2022 ($n=6{,}006$)}} \\
    & 50 & 100 & 200 & 500 \\
    \quad Raw flow $y_{ij}$ & 1.67 & 1.29 & 1.17 & 1.03 \\
                            & [1.31, 2.31] & [1.08, 1.60] & [1.03, 1.36] & [0.95, 1.13] \\
    \quad Residual          & 1.64 & 1.41 & 1.40 & 1.23 \\
                            & [1.28, 2.26] & [1.18, 1.75] & [1.23, 1.63] & [1.13, 1.34] \\[4pt]
    \multicolumn{5}{l}{\textit{Panel C. Intermediate-goods directed sales, ICIO 2022 ($n=6{,}006$)}} \\
    & 50 & 100 & 200 & 500 \\
    \quad Raw flow $y_{ij}$ & 1.87 & 1.82 & 1.22 & 1.06 \\
                            & [1.46, 2.59] & [1.52, 2.26] & [1.07, 1.42] & [0.98, 1.16] \\
    \quad Residual          & 1.84 & 1.55 & 1.53 & 1.32 \\
                            & [1.44, 2.54] & [1.29, 1.93] & [1.34, 1.78] & [1.21, 1.44] \\[4pt]
    \multicolumn{5}{l}{\textit{Panel D. Directed emigration flows, 2015 ($n=30{,}757$ positive)}} \\
    & 100 & 200 & 500 & 1{,}000 \\
    \quad Raw flow $y_{ij}$ & 1.60 & 1.26 & 1.05 & 0.84 \\
                            & [1.33, 1.98] & [1.11, 1.47] & [0.96, 1.15] & [0.79, 0.89] \\
    \quad Residual          & 0.96 & 1.03 & 1.10 & 1.07 \\
                            & [0.81, 1.20] & [0.90, 1.19] & [1.01, 1.20] & [1.01, 1.14] \\[4pt]
    \multicolumn{5}{l}{\textit{Panel E. Directed inward FDI stocks, 2015 ($n=8{,}048$ positive)}} \\
    & 50 & 100 & 200 & 500 \\
    \quad Raw flow $y_{ij}$ & 1.25 & 1.31 & 1.04 & 0.80 \\
                            & [0.98, 1.74] & [1.10, 1.63] & [0.91, 1.21] & [0.73, 0.87] \\
    \quad Residual          & 1.05 & 1.13 & 0.96 & 0.87 \\
                            & [0.83, 1.46] & [0.95, 1.41] & [0.85, 1.12] & [0.80, 0.95] \\
    \bottomrule
  \end{tabularx}
  \begin{tablenotes}[flushleft]\footnotesize
    \item \emph{Notes}: Hill estimator $\hat\alpha(k)$ from \eqref{eq:hill}, applied to the raw directed flow and to the multiplicative residual $\hat\varepsilon_{ij}-1=y_{ij}/\hat\mu_{ij}-1$, where $\hat\mu_{ij}$ is the fitted conditional mean from the two-way fixed-effect PPML specifications of Section~\ref{sec:application}.
      Brackets give $95$ percent confidence intervals from the asymptotic normality of the Hill estimator, $\sqrt{k}\,(\hat\alpha(k)/\alpha-1)\Rightarrow \mathcal{N}(0,1)$, so that the interval is $[\hat\alpha(k)/(1+1.96/\sqrt{k}),\,\hat\alpha(k)/(1-1.96/\sqrt{k})]$.
      Both rows are computed on the PPML estimation sample, that is after the covariate merge and the removal of separated observations, so the raw rows differ in some cells from the international-only rows of Table~\ref{tab:hill_raw}, which uses the full sample of positive flows.
      All measures are cross-border only.
  \end{tablenotes}
  \end{threeparttable}
\end{table}

\paragraph{The correction is large and pervasive.}
Across all of our cross-sectional data settings for gravity models, the bootstrap substantially widens conventional intervals, by roughly $130$--$170$ percent over HC0 in the country-pair, migration, and FDI cross-sections, all of which have a Pareto tail index below two, so the Gaussian sandwich understates uncertainty in each. 
The country-pair total sales decomposes exactly into final- and intermediate-sales flows, and the correction is present and of the same order on both types of flows (Table~\ref{tab:cs_use}), so it is not an artifact of pooling economically distinct transactions. 
The richer final-goods specification, which adds origin-sector fixed effects and so admits pair-level clustering, makes the same point with an additional layer: even after cluster-robust standard errors absorb the within-pair correlation, the bootstrap widens the cluster intervals by a further $110$--$120$ percent, because the cluster sandwich still relies on Gaussian critical values. 
We caution against reading the relative magnitudes across datasets as a clear ranking in the tail index: the size of the bootstrap-to-HC0 ratio reflects both tail heaviness and
design: how many observations share a cluster and how heavily conventional inference already leans on Gaussian critical values.
What the theory delivers, and what is robust here, is the within-design statement: under $\alpha<2$ the Gaussian sandwich miscovers and the $m$-out-of-$n$ bootstrap restores valid interval widths.

\paragraph{Clustering is not a substitute.}
A natural objection is that the dependence the bootstrap captures could be handled by clustering, since a country's flows are plainly not independent. 
We address the objection in its strongest form by reporting both the one-way origin cluster and the two-way origin-and-destination cluster that the gravity literature regards as the appropriate response to exporter- and importer-side cross-sectional dependence \citep{CameronGelbachMiller2011,EggerTarlea2015}. 
The cluster intervals reported throughout
Tables~\ref{tab:cs}--\ref{tab:fdi} show that neither substitutes for the bootstrap. 
The one-way cluster widens HC0 by $16$ to $38$ percent and the two-way cluster by $39$ to $91$ percent across the settings, and in the domestic trade panels the two-way interval is for some covariates as wide as the bootstrap. But width is not coverage, and on the imprecisely estimated covariates the methods reach different verdicts. 
Two-way clustering overturns contiguity and common language in FDI, agreeing with the bootstrap there, where one-way clustering leaves both significant; yet two-way clustering still leaves migration contiguity and intermediate-trade contiguity comfortably significant where the bootstrap drives them to or across zero; and for the FDI colonial coefficient both cluster variants narrow the interval below HC0 and move it across zero into significance, exactly the wrong direction.
This is what Appendix~\ref{app:cluster} predicts: when the dyadic score has infinite variance, the cluster-robust variance is itself driven by a divergent quantity, so the cluster interval is rescaled but is still referred to Gaussian critical values and miscovers. 
The choice of clustering dimension changes the rescaling but not the underlying limit; only the bootstrap tracks the correct $\alpha$-stable limit.

\paragraph{Role of the GE structure.}
The GE-consistent DGP introduces dyadic dependence into $\{\kappa_{ij}\}$ but does not change the qualitative asymptotic picture. 
The rate $n^{1-1/\alpha}$ and the $\alpha$-stability of the limit are determined by the pair-level shock alone; the GE structure affects only the spectral measure $\Gamma$, into which the dyadic heterogeneity is absorbed. 
Under Assumption~\ref{ass:dgp} this heterogeneity is purely deterministic: the shocks $\{\eta_{ij}\}$ are i.i.d.\ across pairs and the scores $s_{ij}=\eta_{ij}\kappa_{ij}$ are
therefore independent conditional on the fixed array $\{\kappa_{ij}\}$, so no stochastic cross-pair term enters the limit theory.

The maintained DGP holds equilibrium prices at the nonstochastic fundamentals $(L,\tau,\zeta_0)$. 
Were the realized shocks instead allowed to perturb the equilibrium, the induced feedback would be asymptotically negligible. 
Each country's price block in \eqref{eq:gefixedpoint} averages over its $N-1$ partners, so a single pair's shock changes any other pair's conditional mean
$\mu_{ij}$, and hence its multiplier $\kappa_{ij}$, by a relative $O(1/N)$. 
Because the stable limit enters only through the spectral measure $\Gamma$, a continuous functional of the multiplier directions $\{\kappa_{ij}/\|\kappa_{ij}\|\}$ on
$\mathbb{S}^{p-1}$, an $O(1/N)$ perturbation of the array shifts $\Gamma$ by $O(1/N)\to0$ and leaves the limit law unchanged; the tail index $\alpha$ and the rate $n^{1-1/\alpha}$, fixed by the shock, are unaffected. 
The $m$-out-of-$n$ bootstrap tracks the correct $\Gamma$ through subsampling without requiring the analyst to take a stand on whether i.i.d.\ or GE structure is operative for the DGP of the data.

\section{Conclusion}
\label{sec:conclusion}

Standard PPML inference in gravity models is unreliable because bilateral flows have a Pareto tail with index well below~2. 
We document this across three bilateral datasets: international trade (ICIO 2022, at the country-pair and final-goods levels), migration, and foreign direct investment. 
The Pareto tail index is well below~2 in every case, and heaviest for FDI, where it approaches~1.
Under a general-equilibrium-consistent DGP in which the only stochastic primitive is a pair-level i.i.d.\ multiplicative shock with regularly varying tails, Theorem~\ref{thm:stable}
shows that the PPML estimator converges at rate $n^{1-1/\alpha}$ to a multivariate $\alpha$-stable distribution, the score covariance $\hat\Omega_n$ diverges at rate $n^{2/\alpha-1}$, and the resulting sandwich $t$-statistic has a non-standard self-normalized stable limit rather than $\mathcal{N}(0,1)$.

The remedy is simple: retain PPML for point estimation, and replace the sandwich-based Gaussian CI with the $m$-out-of-$n$ bootstrap of \citet{ChiangSasakiWang2023} applied to the
self-normalized $t$-statistic. 
The procedure requires no estimate of the tail index or the scale of the stable limit and is valid in both the Gaussian and stable regimes.
The bootstrap CI is materially wider than the HC0 sandwich CI in every specification (the bootstrap-to-HC0 width ratio ranges from about $2.3$ to $2.7$ across the country-pair, migration, and FDI cross-sections, and is of a similar magnitude in a final-goods robustness specification), and in several cases reverses the substantive conclusion of standard inference for the imprecisely estimated dyadic covariates such as contiguity and common language. 
The widening is first-order in every setting we examine, all of which have a tail index below two. We recommend that practitioners report bootstrap confidence intervals whenever heavy-tail diagnostics indicate $\alpha<2$, which Table~\ref{tab:hill_raw} shows to be the empirically operative case for the bilateral flows that typically populate gravity regressions.

The results presented in this paper extend naturally beyond data situations encountered in international economics. 
For instance, tail behavior is widely documented for citation flow data regarding patents or academic publications, where gravity-type flow models are often used in estimation. 
Also, medical expenditure is widely acknowledge to exhibit heavy tail, where the PPML model is suitable to capture heavy tail and zero atom simultaneously. 

\newpage
\appendix

\renewcommand{\thetheorem}{\thesection.\arabic{theorem}}
\renewcommand{\thelemma}{\thesection.\arabic{lemma}}
\renewcommand{\theequation}{\thesection.\arabic{equation}}
\renewcommand{\theassumption}{\thesection.\arabic{assumption}}
\renewcommand{\theproposition}{\thesection.\arabic{proposition}}
\setcounter{theorem}{0}
\setcounter{proposition}{0}
\setcounter{lemma}{0}
\setcounter{equation}{0}
\setcounter{assumption}{0}

\section{Proofs}
\label{app:proofs}

\subsection{Notation and Proof of Theorem~\ref{thm:stable}}
\label{app:profiling}
\label{app:stable-proof}

\medskip\noindent\textbf{Notation.}
We write $f_n\asymp g_n$, equivalently $f_n=\Theta(g_n)$, when $0<\liminf f_n/g_n\le\limsup f_n/g_n<\infty$, and $f_n=\Theta_p(g_n)$ for the in-probability analogue; $\|\cdot\|_{\mathrm{op}}$ denotes the spectral norm (largest singular value), $\|\cdot\|_q$ the $\ell_q$ norm on vectors, and $\|\cdot\|_{p\to q}$ the $\ell_p\to\ell_q$ operator norm on matrices; $\osc(v)=\max_iv_i-\min_iv_i$ the oscillation seminorm. We abbreviate ``with probability approaching one'' as w.p.a.\ one, and write $c_\mu:=c$ and $C_\mu:=C$ for the bounds $\mu_{ij}(\beta)\in[c,C]$ implied by Assumption~\ref{ass:dgp}(b).

Index dyads by $\ell=(i,j)$, $i\neq j$, with $n=N(N-1)\asymp N^2$.
We use the redundant fixed-effect parameterization throughout: $d_\ell=d_{ij}\in\R^{2N}$ is the exporter--importer incidence vector, so the conditional mean is $\mu_\ell(\beta,\phi)=\exp(x_\ell'\beta+d_\ell'\phi)$ with truth $\mu_\ell=\mu_\ell(\beta_0,\phi_0)$. 
Every $d_\ell$ is orthogonal to the null direction $e:=(\mathbf 1_N',-\mathbf 1_N')'\in\R^{2N}$, $d_\ell'e=0$, so $\phi$ is identified only up to $\mathrm{span}(e)$, and every fixed-effect Gram or Hessian matrix below ($H_{\phi\phi,n}$, $\hat H_{\phi\phi,n}(\beta)$, $\bar H_{\phi,n}$, and their pseudo-true and subnetwork analogues) is singular exactly along $e$
and positive definite on $e^\perp$ by \eqref{eq:dense-design}. 
For symmetric $A$ with $Ae=0$ we write $A\restriction_{e^\perp}$ for the operator induced on $e^\perp$, so that $\lambda_{\min}(A\restriction_{e^\perp})=\min_{0\neq v\in e^\perp}v'Av/(v'v)$. 
This is the sense in which \eqref{eq:dense-design} bounds eigenvalues on the identified subspace; the same convention applies on subnetworks, with $e$ replaced by its $2G$-dimensional analogue.
All inverses of such matrices are understood as inverses on $e^\perp$, equivalently Moore--Penrose pseudoinverses; expressions such as $H_{\beta\phi,n}H_{\phi\phi,n}^{-1}d_\ell$ are then well defined and invariant to the representative of $\phi$, because $H_{\beta\phi,n}e=0$ and $d_\ell\in e^\perp$. 
Residuals, fitted indices $d_\ell'\hat\phi(\beta)$, and every object entering the theory are representative-free. No normalization affects these invariant objects; individual lemmas select convenient auxiliary representatives (the Sinkhorn and orthogonal representatives introduced below). 

We introduce notations used for profiling.
Write $H_{\beta\phi,n}=\sum_\ell\mu_\ell x_\ell d_\ell'$ and $H_{\phi\phi,n}=\sum_\ell\mu_\ell d_\ell d_\ell'$. 
The partialled-out regressor is $\tilde x_\ell=x_\ell-H_{\beta\phi,n}H_{\phi\phi,n}^{-1}d_\ell$.
Equivalently, because $d_\ell'\phi$ is additive, the projection reduces coordinate by coordinate ($k=1,\dots,p$) to a $\mu$-weighted two-way fit,
\begin{equation}  \label{eq:xtilde-additive}
  \tilde x_{ij}^{(k)}\;=\;x_{ij}^{(k)}-a_i^{(k)}-b_j^{(k)},
  \qquad
  (a^{(k)},b^{(k)})\;=\;\operatorname*{arg\,min}_{a,b}\sum_{i\neq j}\mu_{ij}\bigl(x_{ij}^{(k)}-a_i-b_j\bigr)^2.
\end{equation}
The minimizing pair is unique up to the shift $(a^{(k)}_i+t,\,b^{(k)}_j-t)$, under which the residual $\tilde x^{(k)}_{ij}$ is invariant; a normalization such as $b^{(k)}_N=0$ may be imposed computationally but plays no role below.
The minimizers solve the coupled normal equations
\begin{equation}  \label{eq:xtilde-foc}
  a_i^{(k)} =\frac{\sum_{j\neq i}\mu_{ij}\bigl(x_{ij}^{(k)}-b_j^{(k)}\bigr)}{\sum_{j\neq i}\mu_{ij}},
  \qquad
  b_j^{(k)} =\frac{\sum_{i\neq j}\mu_{ij}\bigl(x_{ij}^{(k)}-a_i^{(k)}\bigr)}{\sum_{i\neq j}\mu_{ij}}.
\end{equation}  
By construction, it satisfies the residual orthogonality
\begin{equation}\label{eq:orth}
  \sum_\ell\mu_\ell\,\tilde x_\ell\,d_\ell'\;=\;0,\qquad\text{equivalently}\qquad\sum_\ell\mu_\ell\,\tilde x_\ell\,g_\ell\;=\;0 \ \text{ for every additive }g_\ell=a_i+b_j.
\end{equation}
Let $\hat\phi_0=\hat\phi(\beta_0)$ solve $S_{\phi,n}(\beta_0,\hat\phi_0)=0$ and set $\nu_n=\hat\phi_0-\phi_0$, with exporter/importer components $\nu_n= (\xi_1,\dots,\xi_N,\omega_1,\dots,\omega_{N})'$, so that $d_\ell'\nu_n=\xi_i+\omega_j$. Define the multiplicative fixed-effect errors
\begin{equation}  
  u_i\;=\;e^{\xi_i}-1,\qquad w_j\;=\;e^{\omega_j}-1.
  \label{eq:mult-fe}
\end{equation}
The truth-level score is $s_\ell=\eta_\ell\mu_\ell\tilde x_\ell$.
Write $S_i=\sum_{j}\eta_{ij}\mu_{ij}$ and $\tilde S_j=\sum_{i}\eta_{ij}\mu_{ij}$ for the exporter/importer score blocks, and $m_i=\sum_j\mu_{ij}$, $\tilde m_j=\sum_i\mu_{ij}$.
The PPML score blocks are $S_{\beta,n}(\beta,\phi)=\sum_\ell\{y_\ell-\mu_\ell(\beta,\phi)\}x_\ell$ and $S_{\phi,n}(\beta,\phi)=\sum_\ell\{y_\ell-\mu_\ell(\beta,\phi)\}d_\ell$, and the
profile score is $S_n^{\mathrm{profile}}(\beta)\equiv S_{\beta,n}(\beta,\hat\phi(\beta))$, with $\hat\phi(\beta)$ solving $S_{\phi,n}(\beta,\hat\phi(\beta))=0$;
a finite solution exists w.p.a.\ one, uniformly over $\beta\in\mathcal B$, by Corollary~\ref{cor:sample-existence}.
We write $\mu_\ell^0\equiv\mu_\ell$ for the truth-level mean and $\bar Q_n(\beta,\phi)=n^{-1}\sum_\ell[\mu_\ell^0(x_\ell'\beta+d_\ell'\phi)-\exp(x_\ell'\beta+d_\ell'\phi)]$
for the population PPML objective. The concentrated objective is $\hat Q_n(\beta)=Q_n(\beta,\hat\phi(\beta))$. With the fitted weights $\hat\mu_\ell(\beta)=\exp(x_\ell'\beta+d_\ell'\hat\phi(\beta))$, write $\hat H_{\beta\phi,n}(\beta)=\sum_\ell\hat\mu_\ell(\beta)x_\ell d_\ell'$ and
$\hat H_{\phi\phi,n}(\beta)=\sum_\ell\hat\mu_\ell(\beta)d_\ell d_\ell'$, and let
\[
  \hat{\tilde x}_\ell(\beta) \;=\;x_\ell-\hat H_{\beta\phi,n}(\beta)\,\hat H_{\phi\phi,n}(\beta)^{-1}d_\ell
\]
denote the fitted-weight ($\hat\mu(\beta)$-weighted) projection residual. 
It satisfies $\sum_\ell\hat\mu_\ell(\beta)\hat{\tilde x}_\ell(\beta)d_\ell'=0$. 
It is to be distinguished from the pseudo-true-weight residual $\tilde x_\ell(\beta)$ defined below. 
The feasible concentrated Hessian at $\beta$ is
\begin{equation}
  H_n^{\ast}(\beta)\;=\;\frac1n\sum_\ell\hat\mu_\ell(\beta)\,\hat{\tilde x}_\ell(\beta)\hat{\tilde x}_\ell(\beta)' .
  \label{eq:feasible-hessian}
\end{equation}
This matrix is the exact derivative of the profile score: $-\,\partial\{n^{-1}S_n^{\mathrm{profile}}(\beta)\}/\partial\beta'=H_n^{\ast}(\beta)$. 
Indeed, the implicit function theorem applied to $S_{\phi,n}(\beta,\hat\phi(\beta))=0$ gives $\partial\hat\phi(\beta)/\partial\beta' =-\hat H_{\phi\phi,n}(\beta)^{-1}\hat H_{\phi\beta,n}(\beta)$; the Jacobian is invertible on the quotient w.p.a.\ one by Corollary~\ref{cor:sample-existence} and the margin bounds of Lemma~\ref{lem:margins} and Corollary~\ref{cor:uniform-beta}. 
The chain rule and the orthogonality above then give the displayed derivative. 
The estimated score process is
\[
  \hat s_\ell(\beta)\;=\;\bigl\{y_\ell-\hat\mu_\ell(\beta)\bigr\}\,\hat{\tilde x}_\ell(\beta),
\]
always written with its argument: $\hat s_\ell(\beta_0)$ is its value at the truth-profiled point $(\beta_0,\hat\phi_0)$ and $\hat s_\ell(\hat\beta)$ at the estimates. 
The score covariance at the PPML estimates is $\hat\Omega_n=n^{-1}\sum_\ell\hat s_\ell(\hat\beta)\hat s_\ell(\hat\beta)'$, and the truth-profiled score covariance is
$\hat\Omega_n^{\mathrm{prof}}(\beta_0)=n^{-1}\sum_\ell\hat s_\ell(\beta_0)\hat s_\ell(\beta_0)'$.
For $\beta\in\mathcal B$, $\mu_\ell(\beta)=\mu_\ell(\beta,\bar\phi(\beta))$ are the pseudo-true weights and $\tilde x_\ell(\beta)$ the $\mu_\ell(\beta)$-weighted projection residual of $x_\ell$ onto the fixed-effect column space, satisfying $\sum_\ell\mu_\ell(\beta)\tilde x_\ell(\beta)d_\ell'=0$; at $\beta_0$ these are $\mu_\ell^0$ and $\tilde x_\ell$. 
The profile-score remainder is
\begin{equation}
  R_n(\beta)\;\equiv\;\sum_\ell\bigl(y_\ell-\mu_\ell(\beta)\bigr)\tilde x_\ell(\beta)\;-\;S_n^{\mathrm{profile}}(\beta),\qquad R_n\equiv R_n(\beta_0),
  \label{eq:Rn-def}
\end{equation}
the gap between the $\mu(\beta)$-partialled score and the profile score.
Throughout, $a_n$ is the heavy-tail norming of the score sum, defined by $n\Pr(\eta>a_n)\to1$; under Assumption~\ref{ass:dgp}(a), $\mathcal L(t)\to c_\ast$, so $a_n=(c_+c_\ast\,n)^{1/\alpha}(1+o(1))\asymp n^{1/\alpha}\asymp N^{2/\alpha}$, and the tail rates below are pure powers of $N$. 
Finally, $\hat H_n\equiv H_n^{\ast}(\hat\beta)$ is the concentrated Hessian at the estimates.

The recurring score and residual objects are summarized below; an unadorned symbol is evaluated at the truth, an argument $\beta$ indicates the pseudo-true (plain) or fitted (hatted) construction; the estimated score always carries its argument.
\begin{center}
{\small
\begin{tabular}{llll}
\toprule
Symbol & Weights & Residualization & Evaluated at\\
\midrule
$\mu_\ell,\ \tilde x_\ell,\ s_\ell=\eta_\ell\mu_\ell\tilde x_\ell$
  & truth & truth & $(\beta_0,\phi_0)$\\
$\mu_\ell(\beta),\ \tilde x_\ell(\beta)$
  & pseudo-true & pseudo-true & $(\beta,\bar\phi(\beta))$\\
$\hat\mu_\ell(\beta),\ \hat{\tilde x}_\ell(\beta),\ \hat s_\ell(\beta)$
  & fitted & fitted & $(\beta,\hat\phi(\beta))$\\
$\hat{\tilde x}_\ell(\beta_0),\ \hat s_\ell(\beta_0)$
  & fitted & fitted & $(\beta_0,\hat\phi_0)$\\
$\hat{\tilde x}_\ell(\hat\beta),\ \hat s_\ell(\hat\beta)$
  & fitted & fitted & $(\hat\beta,\hat\phi(\hat\beta))$\\
\bottomrule
\end{tabular}}
\end{center}
The Hessian-type matrices are: on $\R^{2N}$, the fixed-effect blocks $H_{\phi\phi,n}$ and $\hat H_{\phi\phi,n}(\beta)$ and the integral-averaged $\bar H_{\phi,n}$ of Lemma~\ref{lem:fe-rate}, with inverses on $e^\perp$; on $\R^{p}$, the population $H_n\to H$ of Assumption~\ref{ass:dgp}(b), the feasible concentrated $H_n^{\ast}(\beta)$ of \eqref{eq:feasible-hessian} with
$\hat H_n=H_n^{\ast}(\hat\beta)$ and segment average $\bar H_n^{\ast}$ of \eqref{eq:linearization-app}, and the pseudo-true and hybrid matrices $H^{\mathrm{pt}}_n(\beta)$ and $H^{\mathrm{hyb}}_n(\beta)$ local to Proposition~\ref{prop:primitive-d}.

Two representatives of the fixed-effect error $\nu_n\in\R^{2N}$ appear below; by the quotient conventions above, $\nu_n$ is identified only up to the null direction $e$, i.e.\ $\xi_i\mapsto\xi_i+t$, $\omega_j\mapsto\omega_j-t$.
Lemma~\ref{lem:margins} works with the Sinkhorn representative $\nu^{\mathrm{sk}}$, normalized by $\sum_je^{\omega_j}=N-1$, under which the fixed effects are exactly the Sinkhorn scalings that rescale the base weights $(\mu_{ij})$ to the observed margins \citep{Sinkhorn1967,Idel2016}.
Lemma~\ref{lem:fe-rate} works with the orthogonal representative $\nu_\perp\in e^\perp$, for which the exact linear system there is well posed on the quotient by $e$. 
The two are linked by $\nu^{\mathrm{sk}}=\nu_\perp+t_\nu \,e$ with a scalar level $t_\nu$ shown to be $O_p(1)$ in Lemma~\ref{lem:fe-rate}. 
Unless stated otherwise, \eqref{eq:mult-fe} and all rate statements refer to the orthogonal representative; the exception is Lemma~\ref{lem:margins}(iii), which is stated in the Sinkhorn representative.
The index $d_\ell'\nu_n=\xi_i+\omega_j$ is common to both, and the exact identity \eqref{eq:exact-Rn} together with the score covariance and Hessian decompositions below hold with the multiplicative errors $(u,w)$ of either representative, because the annihilation of additive components through \eqref{eq:orth} is representative-free.

The appendix is organized as follows. 
This subsection proves Theorem~\ref{thm:stable}; Appendix~\ref{app:valid-proof} proves
Theorem~\ref{thm:valid}; Appendix~\ref{app:primitive-d} states the primitive design condition and verifies the high-level Assumption~\ref{ass:dgp}(d); Appendix~\ref{app:lemmas} collects all technical lemmas. 

\begin{proof}[Proof of Theorem~\ref{thm:stable}]
The proof has five steps: (1)~establish a multivariate Breiman lemma for deterministic but heterogeneous multipliers;
(2)~derive a triangular-array stable CLT for the score sum;
(3)~prove consistency of the PPML estimator;
(4)~combine consistency with the linearization to obtain the estimator's stable limit; and (5)~establish divergence of the sandwich variance.

\emph{Step~1: Conditional Breiman lemma.}
Under Assumption~\ref{ass:dgp}(a), the centered shock $\eta_\ell\equiv\varepsilon_\ell-1$ satisfies $\Pr(\eta_\ell>t)=c_+\,t^{-\alpha}\,\mathcal{L}(t)$ for a slowly varying
$\mathcal{L}$ and $c_+>0$, with $\Pr(\eta_\ell<-t)=0$ for all $t>1$. 
The upper tail of $|\eta_\ell|$ is therefore exactly the upper tail of $\eta_\ell$. By Lemma~\ref{lem:breiman}, for any deterministic $\kappa\in\R^p$ with $\|\kappa\|$ bounded
(Assumption~\ref{ass:dgp}(b)),
\begin{equation}\label{eq:breiman-app}
  \frac{\Pr(\|\eta\,\kappa\|>t)}{\Pr(|\eta|>t)}\;\to\;\|\kappa\|^\alpha\quad\text{as }t\to\infty.
\end{equation}
Since $\kappa_\ell$ is non-stochastic, \eqref{eq:breiman-app} holds for each $\ell$ separately, with the population moment
replaced by the empirical average
\begin{equation}\label{eq:Abar}
  \bar A_n\;\equiv\;\frac{1}{n}\sum_{\ell=1}^n\|\kappa_\ell\|^\alpha,
\end{equation}
which is bounded uniformly in $N$ by Assumption~\ref{ass:dgp}(b).

\emph{Step~2: Triangular-array stable CLT.}
Let $a_n$ be the normalizing sequence from Assumption~\ref{ass:dgp}(a), defined by $n\,\Pr(\eta_{ij}>a_n)\to 1$, so that $a_n=(c_+ c_*\,n)^{1/\alpha}(1+o(1))$. 
Define the finite spectral measure $\Lambda\equiv C_\alpha\bar A\,\Gamma$ on $\mathbb{S}^{p-1}$, with the universal stable constant $C_\alpha$ carried in $\Lambda$ and $a_n$ the
tail quantile. We claim that
\begin{equation}
  a_n^{-1}\sum_{\ell=1}^n s_\ell
  \;\convd\;S_\alpha(\Lambda),
  \label{eq:stable-clt-app}
\end{equation}
where $S_\alpha(\Lambda)$ is the multivariate $\alpha$-stable random vector with finite spectral measure $\Lambda$ in the parametrization of \citet[\S 2.3]{SamorodnitskyTaqqu1994}. We argue by the characteristic-function method, adapted to a triangular array of independent, non-identically distributed summands.

For any $u\in\R^p$, the characteristic function of
$u'\cdot a_n^{-1}\sum_\ell s_\ell$ is
\begin{equation}
  \varphi_n(u)\;=\;\prod_{\ell=1}^n\E\bigl[\exp\bigl(i\,a_n^{-1}\,\eta_\ell\, u'\kappa_\ell\bigr)\bigr],
  \label{eq:cf-app}
\end{equation}
where the factorization uses pair-level independence of $\{\eta_\ell\}$ (Assumption~\ref{ass:dgp}(a)); conditional on the deterministic array $\{\kappa_\ell\}$ the summands
$s_\ell=\eta_\ell\kappa_\ell$ are independent but not identically distributed, with scale $\|\kappa_\ell\|$.

Write $\chi(s)\equiv\E[\exp(is\eta_\ell)]$ for the characteristic function of the centered shock (common across $\ell$ by Assumption~\ref{ass:dgp}(a)), and let $\Pr(|\eta|>t)=(c_++c_-)t^{-\alpha}\mathcal{L}(t)$ with $\mathcal{L}$ slowly varying.
For a centered random variable in the domain of attraction of an $\alpha$-stable law with $\alpha\in(1,2)$ and tail-balance constants $(c_+,c_-)$, the standard Karamata expansion of the characteristic function gives, as $s\to 0$,
\begin{equation}
  \log\chi(s)\;=\;-C_\alpha(c_++c_-)\,|s|^\alpha\,\mathcal{L}(1/|s|)
  \bigl[1-i\beta_*\,\mathrm{sgn}(s)\tan(\tfrac{\pi\alpha}{2})\bigr]
  +o\bigl(|s|^\alpha\mathcal{L}(1/|s|)\bigr),
  \label{eq:cf-expansion}
\end{equation}
where $C_\alpha=\Gamma_E(1-\alpha)\cos(\pi\alpha/2)>0$ is the universal
stable constant (with $\Gamma_E$ the Euler gamma function) and
$\beta_*=(c_+-c_-)/(c_++c_-)\in[-1,1]$ is the skewness
parameter; see \citet[Theorem~XVII.5.3]{Feller1971} or
\citet[Chapter~2]{IbragimovLinnik1971}.
Under Assumption~\ref{ass:dgp}(a), $c_-=0$, so $\beta_*=1$.

Substituting \eqref{eq:cf-expansion} into \eqref{eq:cf-app}
with $s=a_n^{-1}u'\kappa_\ell$ and taking logs,
\begin{equation}
  \log\varphi_n(u)
  \;=\;-C_\alpha(c_++c_-)\sum_{\ell=1}^n
        \frac{|u'\kappa_\ell|^\alpha}{a_n^\alpha}\,
        \mathcal{L}\!\Bigl(\tfrac{a_n}{|u'\kappa_\ell|}\Bigr)
        \bigl[1-i\beta_*\,\mathrm{sgn}(u'\kappa_\ell)
                \tan(\tfrac{\pi\alpha}{2})\bigr]
       \cdot\bigl(1+o(1)\bigr),
  \label{eq:logcf-app}
\end{equation}
uniformly in $u$ on compact sets. Assumption~\ref{ass:dgp}(a) gives
$\mathcal{L}(t)\to c_*$, but the cancellation below uses only slow
variation: the slowly varying factor cancels against the tail-quantile
norming, after a truncation that controls the directions
$u'\kappa_\ell$ near zero. Fix $\epsilon>0$. On
$\{|u'\kappa_\ell|>\epsilon\}$ the ratio $\lambda_\ell=1/|u'\kappa_\ell|$
lies in the compact $[1/C,1/\epsilon]\subset(0,\infty)$ (using
$|u'\kappa_\ell|\le C$ from Assumption~\ref{ass:dgp}(b)), so Karamata's
uniform convergence theorem gives
$\mathcal{L}(a_n/|u'\kappa_\ell|)/\mathcal{L}(a_n)\to1$ uniformly over
these $\ell$ and over compact $u$. On $\{|u'\kappa_\ell|\le\epsilon\}$,
Potter's bound \citep[Ch.~2]{Resnick2007} gives, for any
$\delta\in(0,\alpha)$ and all $n$ large,
$\mathcal{L}(a_n/|u'\kappa_\ell|)/\mathcal{L}(a_n)\le2|u'\kappa_\ell|^{-\delta}$,
so these terms contribute at most
$2\,n^{-1}\sum_{|u'\kappa_\ell|\le\epsilon}|u'\kappa_\ell|^{\alpha-\delta}
\le2\epsilon^{\alpha-\delta}=O(\epsilon^{\alpha-\delta})$ to the
normalized average, vanishing as $\epsilon\downarrow0$. The normalization
$n\,\Pr(|\eta|>a_n)\to1$ reads
$n(c_++c_-)a_n^{-\alpha}\mathcal{L}(a_n)\to1$. Hence
\[
  C_\alpha(c_++c_-)\,a_n^{-\alpha}
  \sum_{\ell=1}^n|u'\kappa_\ell|^\alpha\,
  \mathcal{L}\!\Bigl(\tfrac{a_n}{|u'\kappa_\ell|}\Bigr)\,(\cdots)
  =C_\alpha\,\bigl[n(c_++c_-)a_n^{-\alpha}\mathcal{L}(a_n)\bigr]\,
  \frac1n\sum_{\ell=1}^n|u'\kappa_\ell|^\alpha\,(\cdots)\,(1+o(1)),
\]
so the universal constant $C_\alpha$ remains in the limit while the slowly varying factor $\mathcal{L}$ cancels.
The multiplier scale enters separately through the two empirical averages, which are deterministic since $\{\kappa_\ell\}$ is non-stochastic under Assumption~\ref{ass:dgp}(b). 
Their limits follow from Assumption~\ref{ass:dgp}(c) applied to the test functions
$\theta\mapsto|u'\theta|^\alpha$ and
$\theta\mapsto|u'\theta|^\alpha\mathrm{sgn}(u'\theta)$ on the
unit sphere $\mathbb{S}^{p-1}$:
\begin{align*}
  \frac{1}{n}\sum_{\ell=1}^n|u'\kappa_\ell|^\alpha
    &\;\longrightarrow\;\bar A
        \int_{\mathbb{S}^{p-1}}|u'\theta|^\alpha\,
        \Gamma(d\theta), \\
  \frac{1}{n}\sum_{\ell=1}^n|u'\kappa_\ell|^\alpha\,
                            \mathrm{sgn}(u'\kappa_\ell)
    &\;\longrightarrow\;\bar A
        \int_{\mathbb{S}^{p-1}}|u'\theta|^\alpha\,
        \mathrm{sgn}(u'\theta)\,\Gamma(d\theta).
\end{align*}
The first function is bounded and continuous on $\mathbb{S}^{p-1}$, so the convergence is immediate by \eqref{eq:spec-meas} together with the identity $|u'\kappa_\ell|^\alpha=
\|\kappa_\ell\|^\alpha|u'\theta_\ell|^\alpha$ where $\theta_\ell\equiv\kappa_\ell/\|\kappa_\ell\|$, combined with $\bar A_n\to\bar A$. 
The second function $\theta\mapsto|u'\theta|^\alpha\mathrm{sgn}(u'\theta)$ is in fact continuous on all of $\mathbb{S}^{p-1}$, including on the hyperplane $\{u'\theta=0\}$: the sign factor is discontinuous there, but it is multiplied by $|u'\theta|^\alpha\to0$ as $u'\theta\to0$ (using $\alpha>0$), which removes the discontinuity. 
Hence the second display converges for every $u\in\R^p$ by the continuous-mapping form of the portmanteau theorem.

The right-hand side of \eqref{eq:logcf-app} therefore converges, for every $u\in\R^p$, to
\begin{equation*}
  -C_\alpha\,\bar A
  \int_{\mathbb{S}^{p-1}}|u'\theta|^\alpha
  \bigl[1-i\beta_*\,\mathrm{sgn}(u'\theta)\tan(\tfrac{\pi\alpha}{2})\bigr]
  \,\Gamma(d\theta).
\end{equation*}
This is the logarithm of the characteristic function of a multivariate $\alpha$-stable random vector $S_\alpha(\Lambda)$ on $\mathbb{R}^p$ with finite spectral measure $\Lambda\equiv C_\alpha\,\bar A\,\Gamma$ in the parametrization of \citet[\S 2.3]{SamorodnitskyTaqqu1994}, where the universal stable constant $C_\alpha=\Gamma_E(1-\alpha)\cos(\pi\alpha/2)$
is carried in $\Lambda$; the directional skewness is inherited from $\beta_*=1$ ($c_-=0$). Pointwise convergence at every $u$, combined with continuity at $u=0$ of the limiting
characteristic function, implies convergence in distribution by L\'evy's continuity theorem:
\begin{equation}
  a_n^{-1}\sum_{\ell=1}^n s_\ell
  \;\convd\;
  S_\alpha(\Lambda),
\end{equation}
establishing \eqref{eq:stable-clt-app}.

\emph{Step~3: Consistency.}
The estimator maximizes the two-way fixed-effect PPML objective
\begin{equation}
  Q_n(\beta,\phi)=\frac{1}{n}\sum_{\ell=1}^n
    \bigl[y_\ell(x_\ell'\beta+d_\ell'\phi)-\exp(x_\ell'\beta+d_\ell'\phi)\bigr].
  \label{eq:Q-obj}
\end{equation} 
Note that each summand is a strictly concave function $t\mapsto y_\ell t-e^{t}$ of the affine index $t=x_\ell'\beta+d_\ell'\phi$, so
$Q_n$ is jointly concave in $(\beta,\phi)$. Concentrating preserves concavity:
with $\beta_\lambda=\lambda\beta_1+(1-\lambda)\beta_0$ and $\phi_\lambda=\lambda\hat\phi(\beta_1)+(1-\lambda)\hat\phi(\beta_0)$, $\hat Q_n(\beta_\lambda)\ge Q_n(\beta_\lambda,\phi_\lambda)\ge\lambda\hat Q_n(\beta_1)+(1-\lambda)\hat Q_n(\beta_0)$ by maximality and joint concavity, so $\hat Q_n(\beta)\equiv Q_n(\beta,\hat\phi(\beta))$ is concave in $\beta$ \citep[Section~5]{Rockafellar1970}. 
The fixed-effect block is strictly negative definite on the quotient at every finite $\phi$, not only at the pseudo-true weights of \eqref{eq:dense-design}:
the weights $\mu_\ell(\beta,\phi)$ are strictly positive, and on the off-diagonal complete bipartite support the null space of $\sum_\ell d_\ell d_\ell'$ is exactly $\mathrm{span}(e)$, so
$v'\bigl(\sum_\ell\mu_\ell(\beta,\phi)d_\ell d_\ell'\bigr)v =\sum_\ell\mu_\ell(\beta,\phi)(d_\ell'v)^2>0$ for every nonzero $v\in e^\perp$. 
Hence $\hat\phi(\beta)=\arg\max_\phi Q_n(\beta,\phi)$ is unique on the quotient; \eqref{eq:dense-design} is reserved for the quantitative eigenvalue bounds at the pseudo-true weights. 
By Danskin's theorem $\nabla\hat Q_n(\beta)=n^{-1}S_n^{\mathrm{profile}}(\beta)$.

Now we follow the consistency argument of \citet[Thm.~B.3]{FernandezVal2016}: since profiling out the fixed effects preserves concavity, the profile score $\nabla\hat Q_n$ is monotone, and $\hat\beta$ is is shown to lie in any fixed neighborhood of $\beta_0$ w.p.a.\ one by a sign condition on the profile score at the boundary of that neighborhood.

Let $\bar Q_n^{\mathrm{prof}}(\beta)=\bar Q_n(\beta,\bar\phi(\beta))$ be the profile of the population objective $\bar Q_n$ defined above, itself concave in $\beta$. By the envelope theorem, legitimate by the differentiability of $\beta\mapsto\bar\phi(\beta)$ of Lemma~\ref{lem:smooth-profile}, and the orthogonality $\sum_\ell\mu_\ell(\beta)\tilde x_\ell(\beta)d_\ell'=0$,
\begin{equation}
  \nabla\bar Q_n^{\mathrm{prof}}(\beta)
  =\bar s_n(\beta)\equiv\frac1n\sum_\ell\bigl(\mu_\ell^0-\mu_\ell(\beta)\bigr)\tilde x_\ell(\beta),
  \qquad
  \bar s_n(\beta_0)=0,\quad \nabla\bar s_n(\beta_0)\to -H\prec0,
  \label{eq:pop-prof-score}
\end{equation}
so $\beta_0$ is the unique stationary point of the concave population profile.
Definition~\eqref{eq:Rn-def} with $y_\ell=\mu_\ell^0(1+\eta_\ell)$ gives
\[
  \frac1n S_n^{\mathrm{profile}}(\beta)
  =\bar s_n(\beta)
  +\underbrace{\frac1n\sum_\ell\mu_\ell^0\eta_\ell\,\tilde x_\ell(\beta)}_{(\mathrm{I})}
  -\underbrace{\frac1n R_n(\beta)}_{(\mathrm{II})}.
\]
We show that the feasible profile score converges to its population counterpart uniformly on $\mathcal B$,
\begin{equation}
  \sup_{\beta\in\mathcal B}\bigl\|n^{-1}S_n^{\mathrm{profile}}(\beta)-\bar s_n(\beta)\bigr\|\convp0.
  \label{eq:grad-uc}
\end{equation}
Consider (I). Fix $q\in(1,\alpha)$. The $\eta_\ell$ are i.i.d.\ mean zero with $\E|\eta|^q<\infty$, and the multipliers $\mu_\ell^0\tilde x_\ell(\beta)$ are bounded uniformly in $\ell$ and $\beta$. The von Bahr--Esseen inequality \citep{vonBahrEsseen1965} therefore gives $\E\|\sum_\ell\mu_\ell^0\eta_\ell\tilde x_\ell(\beta)\|^q\le Cn$, so $(\mathrm{I})=O_p(n^{1/q-1})=o_p(1)$ pointwise. Then \eqref{eq:grad-uc} follows from the uniformity over the compact $\mathcal B$, which further follows from a finite set. In particular, the Lipschitz bound
\[
  \|(\mathrm I)(\beta)-(\mathrm I)(\beta')\|
  \le\Bigl(n^{-1}\sum_\ell\mu_\ell^0|\eta_\ell|\Bigr)C\,\|\beta-\beta'\|
  =O_p(1)\,\|\beta-\beta'\|
\]
holds by the uniform Lipschitz continuity of $\tilde x_\ell(\cdot)$ (Lemma~\ref{lem:smooth-profile}) and $\E|\eta_\ell|<\infty$. Fix $\delta>0$ and a finite $\delta$-net $\mathcal B_\delta\subset\mathcal B$. Then
$\sup_{\beta\in\mathcal B}\|(\mathrm I)(\beta)\|
\le\max_{\beta\in\mathcal B_\delta}\|(\mathrm I)(\beta)\|+O_p(1)\,\delta$.
The first term is $o_p(1)$ by pointwise convergence over the finite net. The second is arbitrarily small.
For (II), Assumption~\ref{ass:dgp}(d) gives $\sup_{\beta\in\mathcal B}\|R_n(\beta)\|=o_p(a_n)$, so $(\mathrm{II})=o_p(a_n/n)$, establishing \eqref{eq:grad-uc}.

We now bound $||\hat\beta-\beta_0||$. Fix $\epsilon>0$ with $\{\beta:\|\beta-\beta_0\|\le\epsilon\}\subset\mathcal B$. 
Expand the smooth population profile score to first order. 
Together with $\nabla\bar s_n(\beta_0)\to-H$ and the uniform equicontinuity of $\nabla\bar s_n$ near $\beta_0$ (Lemma~\ref{lem:smooth-profile}), this gives, for $n$ large and $\epsilon$ small, $u'\bar s_n(\beta_0+\epsilon u)=-\epsilon\,u'Hu+o(\epsilon)\le-\tfrac12\epsilon\,\lambda_{\min}(H)$ uniformly over unit vectors $u$. By \eqref{eq:grad-uc}, with probability approaching one,
\begin{equation}
  \sup_{\|u\|=1}\,u'\,n^{-1}S_n^{\mathrm{profile}}(\beta_0+\epsilon u) \;\le\;-\tfrac14\,\epsilon\,\lambda_{\min}(H)\;<\;0 .
  \label{eq:score-sign}
\end{equation}
Suppose $\|\hat\beta-\beta_0\|\ge\epsilon$ and set $u=(\hat\beta-\beta_0)/\|\hat\beta-\beta_0\|$.
Since $\hat Q_n$ is concave with global maximizer $\hat\beta$, the map $s\mapsto u'\nabla\hat Q_n(\beta_0+su)=u'n^{-1}S_n^{\mathrm{profile}}(\beta_0+su)$ is nonincreasing and nonnegative on $[0,\|\hat\beta-\beta_0\|]$; in particular $u'n^{-1}S_n^{\mathrm{profile}}(\beta_0+\epsilon u)\ge0$, contradicting \eqref{eq:score-sign}.
Hence $\|\hat\beta-\beta_0\|<\epsilon$ w.p.a. one, and since $\epsilon>0$ is arbitrary, $\hat\beta\convp\beta_0$, establishing part~(i) of Theorem~\ref{thm:stable}.

\emph{Step~4: From score to estimator.}
Recall the score blocks $S_{\beta,n},S_{\phi,n}$, the concentration map $\hat\phi(\beta)$, and the profile score $S_n^{\mathrm{profile}}(\beta)\equiv S_{\beta,n}(\beta,\hat\phi(\beta))$.
The concentrated PPML estimator solves $S_n^{\mathrm{profile}}(\hat\beta)=0$. This unconstrained first-order condition is available w.p.a.\ one: $\hat\beta\convp\beta_0\in\operatorname{int}(\mathcal B)$ places $\hat\beta$ in the interior. Because $\hat\phi(\hat\beta)$ satisfies the exporter/importer first-order conditions, it is equivalently the
partialled-out condition
\[
  \sum_\ell\bigl(y_\ell-\hat\mu_\ell(\hat\beta)\bigr)
  \hat{\tilde x}_\ell(\hat\beta)=0 .
\]
Evaluated at the truth, the profile-score condition of Assumption~\ref{ass:dgp}(d)
gives
\[
  S_n^{\mathrm{profile}}(\beta_0)=\sum_{\ell}s_\ell+o_p(a_n),
  \qquad s_\ell=\eta_\ell\kappa_\ell,
\]
so the heavy-tailed score sum of Step~2 is, to leading order, the profiled
score at $\beta_0$. The map $\beta\mapsto\hat\phi(\beta)$ is continuously
differentiable w.p.a.\ one. Hence the profile score is continuously
differentiable, with exact derivative
$\partial S_n^{\mathrm{profile}}(\beta)/\partial\beta'=-\,n\,H_n^{\ast}(\beta)$,
the feasible concentrated Hessian \eqref{eq:feasible-hessian}, as
recorded in the notation above. Apply the fundamental theorem of calculus to
$t\mapsto S_n^{\mathrm{profile}}(\beta_0+t(\hat\beta-\beta_0))$ and
carry the profile-score remainder $o_p(a_n)$. This gives the exact linearization
\begin{equation}
  \hat\beta-\beta_0
  \;=\;\bigl(n\bar H_n^{\ast}\bigr)^{-1}
       \Bigl(\sum_{\ell=1}^n s_\ell+o_p(a_n)\Bigr),
  \qquad
  \bar H_n^{\ast}
  \;\equiv\;\int_0^1H_n^{\ast}\bigl(\beta_0+t(\hat\beta-\beta_0)\bigr)\,dt,
  \label{eq:linearization-app}
\end{equation}
valid on the event, of probability approaching one, on which $\bar H_n^{\ast}$ is invertible. 
By Step~3 there is a sequence $\delta_n\downarrow0$ with $\Pr(\|\hat\beta-\beta_0\|\le\delta_n)\to1$. 
On that event the segment $\{\beta_0+t(\hat\beta-\beta_0):t\in[0,1]\}$ lies in the ball $\{\|\beta-\beta_0\|\le\delta_n\}$. 
The local-uniform profile-Hessian condition of Assumption~\ref{ass:dgp}(d) then gives
\begin{equation}
  \bigl\|\bar H_n^{\ast}-H\bigr\| \;\le\;\sup_{\|\beta-\beta_0\|\le\delta_n}\bigl\|H_n^{\ast}(\beta)-H\bigr\|\;\convp\;0 ,
  \label{eq:Hbar-conv}
\end{equation}
which also implies the invertibility w.p.a.\ one.

By Step~2, $\sum_\ell s_\ell=O_p(a_n)$. Multiplying both sides of \eqref{eq:linearization-app} by $n/a_n$, the remainder contributes $(\bar H_n^{\ast})^{-1}a_n^{-1}o_p(a_n)=o_p(1)$, so
\begin{equation}
  \frac{n}{a_n}(\hat\beta-\beta_0)\;=\;(\bar H_n^{\ast})^{-1}\Bigl(a_n^{-1}\sum_{\ell=1}^n s_\ell+o_p(1)\Bigr)\;\convd\;H^{-1}\,S_\alpha(\Lambda),
  \label{eq:rate-app}
\end{equation}
by Slutsky's lemma combined with \eqref{eq:Hbar-conv} and \eqref{eq:stable-clt-app}, where the finite spectral measure $\Lambda=C_\alpha\bar A\,\Gamma$ from Step~2 carries the empirical scale $\bar A$ and the universal stable constant $C_\alpha$, with $a_n$ kept equal to the tail quantile. 
This proves part~(ii) of Theorem~\ref{thm:stable}.

\emph{Step~5: Score-covariance divergence and failure of Gaussian sandwich inference.}
The truth-level squared-score sum is $\Omega_n^{\mathrm{truth}}\equiv n^{-1}\sum_\ell s_\ell s_\ell'$ with $s_\ell=\eta_\ell\kappa_\ell$. 
Each summand $s_\ell s_\ell'=\eta_\ell^2\,\kappa_\ell\kappa_\ell'$ has, conditional on the deterministic multiplier $\kappa_\ell$, the same upper-tail index as $\eta_\ell^2$, namely $\alpha/2\in(0,1)$ under Assumption~\ref{ass:dgp}(a) with $\alpha\in(1,2)$, and tail scale $\|\kappa_\ell\kappa_\ell'\|_F^{\alpha/2} =\|\kappa_\ell\|^\alpha$, where $\|\cdot\|_F$ is the Frobenius norm. 
The norming is $a_n^2$, since $\Pr(\eta_\ell^2>a_n^2)=\Pr(\eta_\ell>a_n)\sim n^{-1}$. 
Writing $M_\ell\equiv\kappa_\ell\kappa_\ell'/\|\kappa_\ell\|^2$ for the unit-Frobenius-norm matrix direction (rank one and positive semidefinite), the multivariate Breiman lemma
(Lemma~\ref{lem:breiman}, applied to the vectorized matrix multiplier) together with the characteristic-function argument of Step~2 give
\begin{equation}
  a_n^{-2}\,\sum_{\ell=1}^n s_\ell s_\ell'\;\convd\;\Xi\;\equiv\;S_{\alpha/2}(\Lambda_2),
  \label{eq:sum-meat}
\end{equation}
an $(\alpha/2)$-stable random matrix with finite spectral measure $\Lambda_2=C_{\alpha/2}\,\bar A\,\Gamma_2$ on the unit sphere of symmetric matrices, with the universal $(\alpha/2)$-stable constant $C_{\alpha/2}=\Gamma_E(1-\tfrac{\alpha}{2})\cos(\tfrac{\pi\alpha}{4})$ carried in $\Lambda_2$. 
For tail index $\alpha/2\in(1/2,1)$ the analogue of the expansion \eqref{eq:cf-expansion} holds without a compensating drift term \citep[Ch.~XVII.5]{Feller1971}, so the Step~2 argument applies to the uncentered array. Here
\begin{equation}
  \Gamma_2\;=\;\lim_{n\to\infty}\frac{\sum_\ell\delta_{M_\ell}\,\|\kappa_\ell\|^\alpha}{\sum_\ell\|\kappa_\ell\|^\alpha}
  \label{eq:spec-meas-2}
\end{equation}
is the weak limit of the $\|\kappa_\ell\|^\alpha$-weighted empirical distribution of the directions $M_\ell$. 
The weights are the same $\|\kappa_\ell\|^\alpha$ as in the linear array, so \eqref{eq:spec-meas-2} converges under Assumption~\ref{ass:dgp}(c): the map $\theta\mapsto\theta\theta'$
is continuous on $\mathbb{S}^{p-1}$ and replaces the directional limit $\Gamma$ with $\Gamma_2$. 
Because every summand $\eta_\ell^2\kappa_\ell\kappa_\ell'$ is positive semidefinite and the shock is one-sided ($c_-=0$), $\Xi$ is supported on the positive-semidefinite cone. 
Dividing both sides by $n$, equivalently multiplying by $a_n^2/n$,
\begin{equation}
  \frac{n}{a_n^2}\,\Omega_n^{\mathrm{truth}}\;=\;\frac{1}{a_n^2}\sum_{\ell=1}^n s_\ell s_\ell'\;\convd\;\Xi.
  \label{eq:meat-rate-app}
\end{equation}
The estimated score covariance $\hat\Omega_n=n^{-1}\sum_\ell\hat s_\ell(\hat\beta) \hat s_\ell(\hat\beta)'$ uses scores at the PPML estimates $(\hat\beta,\hat\phi(\hat\beta))$ rather than at the truth.
The truth-profiled score covariance condition of Assumption~\ref{ass:dgp}(d) gives $(n/a_n^2)\bigl(\hat\Omega_n^{\mathrm{prof}}(\beta_0) -\Omega_n^{\mathrm{truth}}\bigr)\convp0$, where
$\hat\Omega_n^{\mathrm{prof}}(\beta_0)$ is the score covariance at $(\beta_0,\hat\phi(\beta_0))$. 
The passage to $\hat\Omega_n$ uses the estimator rate now available together with a Lipschitz bound on the sample profiling map. 
The vector $\hat c^{(k)}(\beta):=-\partial\hat\phi(\beta)/\partial\beta_k$ solves $\hat H_{\phi\phi,n}(\beta)\,\hat c^{(k)} =\sum_\ell\hat\mu_\ell(\beta)d_\ell x^{(k)}_\ell$, the normal-equation
system of the $\hat\mu(\beta)$-weighted two-way fit of $x^{(k)}$, where $\hat\mu_\ell(\beta)=\mu_\ell(\beta)e^{\xi_i(\beta)} e^{\omega_j(\beta)}$ factors as instance~(ii) of Lemma~\ref{lem:linf-fit}, with $\epsilon_N=o_p(1)$ uniformly over $\beta\in\mathcal B$ by Corollary~\ref{cor:uniform-beta}. 
Hence, w.p.a.\ one, repeating the two applications of Lemma~\ref{lem:linf-fit} in the proof of Lemma~\ref{lem:smooth-profile} at the sample weights, we obtain
\begin{equation}
  \sup_{\beta\in\mathcal B}\max_{\ell,k}\bigl|d_\ell'\,\partial\hat\phi(\beta)/\partial\beta_k\bigr|\le C,\qquad
  \sup_{\beta\in\mathcal B}\max_{\ell,k} \bigl\|\partial\hat{\tilde x}_\ell(\beta)/\partial\beta_k\bigr\|\le C .
  \label{eq:sample-lipschitz}
\end{equation}
Integrating
\eqref{eq:sample-lipschitz} along the segment from $\beta_0$ to
$\hat\beta$ gives, uniformly in $\ell$,
\[
  \bigl|\log\hat\mu_\ell(\hat\beta)-\log\hat\mu_\ell(\beta_0)\bigr|
  \le C\|\hat\beta-\beta_0\|,
  \qquad
  \bigl\|\hat{\tilde x}_\ell(\hat\beta)-\hat{\tilde x}_\ell(\beta_0)\bigr\|
  \le C\|\hat\beta-\beta_0\| .
\]
Write $\zeta_\ell:=|u_i|+|w_j|+|u_iw_j|$, so that
$\hat\mu_\ell(\beta_0)=\mu_\ell(1+u_i)(1+w_j)\le C(1+\zeta_\ell)$ and
$\|\hat s_\ell(\beta_0)\|
=\|\mu_\ell(\eta_\ell-u_i-w_j-u_iw_j)\,\hat{\tilde x}_\ell(\beta_0)\|
\le C(|\eta_\ell|+\zeta_\ell)$, using
$\max_\ell\|\hat{\tilde x}_\ell(\beta_0)\|\le C$ w.p.a.\ one
(Lemma~\ref{lem:fitted-residual}). Since
$\|\hat\beta-\beta_0\|=O_p(a_n/n)\to0$ by \eqref{eq:rate-app}, the
decomposition
$\Delta_\ell:=\hat s_\ell(\hat\beta)-\hat s_\ell(\beta_0)
=\bigl(\hat\mu_\ell(\beta_0)-\hat\mu_\ell(\hat\beta)\bigr)
\hat{\tilde x}_\ell(\hat\beta)
+\bigl(y_\ell-\hat\mu_\ell(\beta_0)\bigr)
\bigl(\hat{\tilde x}_\ell(\hat\beta)-\hat{\tilde x}_\ell(\beta_0)\bigr)$ yields
\[
  \|\Delta_\ell\|\le C\|\hat\beta-\beta_0\|
  \bigl(1+|\eta_\ell|+\zeta_\ell\bigr).
\]
By Lemma~\ref{lem:fe-rate}, $\sum_\ell\zeta_\ell^2\le3\bigl(N\sum_iu_i^2+N\sum_jw_j^2 +\sum_iu_i^2\sum_jw_j^2\bigr) =O_p\bigl(N^{4/\alpha-1}+N^{8/\alpha-4}\bigr)=o_p(a_n^2)$, while
$\sum_\ell\eta_\ell^2=O_p(a_n^2)$ by Lemma~\ref{lem:rowsum}(iv): the i.i.d.\ nonnegative array $\{\eta_\ell^2\}$ is regularly varying with index $\alpha/2\in(0,1)$. 
Hence $\sum_\ell(1+|\eta_\ell|+\zeta_\ell)^2=O_p(n+a_n^2)$. 
The per-term difference of the outer products is $\Delta_\ell\hat s_\ell(\beta_0)'+\hat s_\ell(\beta_0)\Delta_\ell'
+\Delta_\ell\Delta_\ell'$, whence 
\[
  \frac{n}{a_n^2}\, \bigl\|\hat\Omega_n-\hat\Omega_n^{\mathrm{prof}}(\beta_0)\bigr\|\le\frac{C}{a_n^2}\Bigl(\|\hat\beta-\beta_0\|+\|\hat\beta-\beta_0\|^2\Bigr)
  \sum_\ell\bigl(1+|\eta_\ell|+\zeta_\ell\bigr)^2=O_p\Bigl(\frac{a_n}{n}+\frac1{a_n}\Bigr)\convp0
\]
for every $\alpha\in(1,2)$. 
Hence $(n/a_n^2)\,(\hat\Omega_n-\Omega_n^{\mathrm{truth}})\convp 0$, so the same limit applies to $\hat\Omega_n$:
\begin{equation}
  \frac{n}{a_n^2}\,\hat\Omega_n
  \;\convd\;\Xi.
\end{equation}
Equivalently, $\hat\Omega_n=O_p(a_n^2/n)$, and since
$a_n^2/n\to\infty$ for $\alpha\in(1,2)$, $\hat\Omega_n$ diverges
componentwise on coordinates with nonzero stable scale. This
proves part~(iii).

The estimated variance for $\hat\beta$ is $\widehat{\mathrm{Var}}(\hat\beta)=n^{-1}\hat H_n^{-1}\hat\Omega_n \hat H_n^{-1}=O_p(a_n^2/n^2)$, of the same order as
$(\hat\beta-\beta_0)^2$ by Step~4, so $\hat t_k=O_p(1)$. 
Its limit is identified by the joint law of the numerator and the denominator, which we obtain from a single point-process limit.

\emph{Joint stable limit.} Step~2 gives $a_n^{-1}\sum_\ell s_\ell\convd S_\alpha(\Lambda)$ and \eqref{eq:sum-meat} gives $a_n^{-2}\sum_\ell s_\ell s_\ell'\convd\Xi$. 
The points $\{a_n^{-1}s_\ell=a_n^{-1}\eta_\ell\kappa_\ell\}_{\ell=1}^n$ form a triangular array whose conditional law is regularly varying with index $\alpha$ and bounded directions ($\|\kappa_\ell\|\le C$ by Assumption~\ref{ass:dgp}(b)), so by heavy-tailed point-process convergence \citep[Ch.~7]{Resnick2007},
\begin{equation}
  \Pi_n\equiv\sum_{\ell=1}^n\delta_{a_n^{-1}s_\ell}\;\Rightarrow\;\Pi \qquad\text{on }\R^p\setminus\{0\},
  \label{eq:prm}
\end{equation}
with $\Pi$ a Poisson random measure whose mean measure $\varpi$ has, in polar coordinates, radial index $\alpha$ and angular part proportional to the score spectral measure $\bar A\,\Gamma$; since $\alpha>1$, $\int_{\|z\|>\epsilon}\|z\|\,\varpi(dz)<\infty$ for every $\epsilon>0$.
Joint convergence of the two functionals follows by a compensated truncation. 
For $\epsilon\in(0,1]$ define
\begin{align*}
  m_{n,\epsilon} &:=\sum_{\ell=1}^n \E\bigl[a_n^{-1}s_\ell\,\1\{\|s_\ell\|>\epsilon a_n\}\bigr],\\
  I_{n,\epsilon} &:=\int_{\|z\|>\epsilon}z\,\Pi_n(dz)-m_{n,\epsilon}\;=\;a_n^{-1}\sum_{\ell=1}^n \Bigl(s_\ell\1\{\|s_\ell\|>\epsilon a_n\}-\E\bigl[s_\ell\1\{\|s_\ell\|>\epsilon a_n\}\bigr]\Bigr),\\
  \mathcal Q_{n,\epsilon} &:=\int_{\|z\|>\epsilon}zz'\,\Pi_n(dz),
\end{align*}
the quadratic coordinate requiring no compensation because its tail index is $\alpha/2<1$. 
These functionals are unbounded at infinity, so the continuous-mapping step uses an upper truncation. 
Fix $M>\epsilon$ with $\varpi(\{\|z\|=M\})=0$. On the compact annulus
$\{\epsilon<\|z\|\le M\}$ the integrands are bounded and the limit
mean measure assigns no mass to either boundary sphere, so
\eqref{eq:prm} and the continuous mapping theorem give the joint
convergence of
$\bigl(\int_{\epsilon<\|z\|\le M}z\,\Pi_n(dz),\
\int_{\epsilon<\|z\|\le M}zz'\,\Pi_n(dz)\bigr)$
to the corresponding integrals against $\Pi$. The truncation is then
removed. First,
$\Pr(\max_\ell\|s_\ell\|>Ma_n)
\le\sum_\ell\Pr(\|s_\ell\|>Ma_n)\le CM^{-\alpha}+o(1)$,
so with probability at least $1-CM^{-\alpha}-o(1)$, uniformly in $n$, the truncated and untruncated point-process integrals coincide; the analogous bound holds for $\Pi$. 
Second, the centering tail is deterministic and satisfies $\sum_\ell\E[a_n^{-1}\|s_\ell\|\1\{\|s_\ell\|>Ma_n\}]\le CM^{1-\alpha}$ uniformly in $n$, by Karamata's theorem. 
Letting $M\to\infty$ after $n\to\infty$ yields the joint convergence of the uncentered pair over $\{\|z\|>\epsilon\}$:
$\bigl(\int_{\|z\|>\epsilon}z\,\Pi_n(dz),\,\mathcal Q_{n,\epsilon}\bigr) \convd\bigl(\int_{\|z\|>\epsilon}z\,\Pi(dz),\,\mathcal Q_\epsilon\bigr)$ with
$\mathcal Q_\epsilon=\int_{\|z\|>\epsilon}zz'\,\Pi(dz)$. The
deterministic centering converges as well:
$m_{n,\epsilon}\to m_\epsilon:=\int_{\|z\|>\epsilon}z\,\varpi(dz)$. This
follows from the regular-variation calculation of Step~2 and the
uniform integrability supplied by Karamata's theorem,
$\sup_n\sum_\ell\E[a_n^{-1}\|s_\ell\|\1\{\|s_\ell\|>Ma_n\}]\le
CM^{1-\alpha}\to0$ as $M\to\infty$. Hence, jointly,
\[
  (I_{n,\epsilon},\,\mathcal Q_{n,\epsilon})\;\convd\;(I_\epsilon,\,\mathcal Q_\epsilon),
  \qquad
  I_\epsilon:=\int_{\|z\|>\epsilon}z\,\bigl(\Pi(dz)-\varpi(dz)\bigr).
\]
As $\epsilon\downarrow0$ the compensated Poisson integral $I_\epsilon$ converges almost surely to the mean-zero $\alpha$-stable vector with spectral measure $\Lambda$ \citep[Section~3.12]{SamorodnitskyTaqqu1994};
its law coincides with that of $S_\alpha(\Lambda)$ from Step~2, both being infinitely divisible with L\'evy measure $\varpi$, no Gaussian component, and zero mean (the Step~2 characteristic function carries no shift term and $\alpha>1$). Likewise $\mathcal Q_\epsilon\to\Xi$ a.s.

On the approximation side, since $\E s_\ell=0$ the gap is exactly the compensated small-jump sum,
\[
  a_n^{-1}\sum_{\ell=1}^n s_\ell-I_{n,\epsilon}
  \;=\;a_n^{-1}\sum_{\ell=1}^n
  \Bigl(s_\ell\1\{\|s_\ell\|\le\epsilon a_n\}
  -\E\bigl[s_\ell\1\{\|s_\ell\|\le\epsilon a_n\}\bigr]\Bigr),
\]
whose variance is bounded, by Karamata's theorem, as
\[
  \operatorname{Var}\Bigl(u'\Bigl[a_n^{-1}\textstyle\sum_\ell
  \bigl(s_\ell\1\{\|s_\ell\|\le\epsilon a_n\}-\E s_\ell\1\{\cdot\}\bigr)
  \Bigr]\Bigr)
  \le a_n^{-2}\sum_\ell\E\bigl[(u's_\ell)^2
  \1\{\|s_\ell\|\le\epsilon a_n\}\bigr]
  \le C\epsilon^{2-\alpha}(1+o(1)),
\]
vanishing as $\epsilon\downarrow0$ uniformly in $n$; for the quadratic part, all summands are positive semidefinite and
$\E\|a_n^{-2}\sum_\ell s_\ell s_\ell'\1\{\|s_\ell\|\le\epsilon a_n\}\| \le C\epsilon^{2-\alpha}(1+o(1))\downarrow0$. 
Letting $\epsilon\downarrow0$ along the standard three-epsilon argument \citep[Theorem~3.2]{Billingsley1999} yields
\begin{equation}
  \biggl(a_n^{-1}\sum_{\ell=1}^n s_\ell,\; a_n^{-2}\sum_{\ell=1}^n s_\ell s_\ell'\biggr)\;\convd\; \bigl(S_\alpha(\Lambda),\;\Xi\bigr),
  \label{eq:joint-stable}
\end{equation}
the two coordinates sharing the single point process $\Pi$, yielding the joint convergence.

\emph{Conclusion of part~(iv).} Write the conventional
$t$-statistic as
\begin{equation}
  \hat t_k\;=\;\frac{\hat\beta_k-\beta_{0,k}}{\widehat{\mathrm{se}}_k}
  \;=\;\frac{(n/a_n)(\hat\beta_k-\beta_{0,k})} {\bigl[(n/a_n)^2\,\widehat{\mathrm{Var}}(\hat\beta)_{kk}\bigr]^{1/2}}.
\end{equation}
By the linearization \eqref{eq:linearization-app}, \eqref{eq:Hbar-conv}, and $\hat H_n=H_n^{\ast}(\hat\beta)\convp H$ (the local-uniform Hessian
condition of Assumption~\ref{ass:dgp}(d) evaluated at $\hat\beta$),
$(n/a_n)(\hat\beta-\beta_0)=\hat H_n^{-1}\,a_n^{-1}\sum_\ell s_\ell+o_p(1)$, while $(n/a_n)^2\widehat{\mathrm{Var}}(\hat\beta)=\hat H_n^{-1}\bigl(a_n^{-2}\hat\Omega_n\bigr)\hat H_n^{-1}$. Since $\hat H_n\convp H$, Slutsky's theorem and the joint limit
\eqref{eq:joint-stable} give
\begin{equation}
  \Bigl((n/a_n)(\hat\beta-\beta_0),\;(n/a_n)^2\widehat{\mathrm{Var}}(\hat\beta)\Bigr) \;\convd\; \bigl(H^{-1}S_\alpha(\Lambda),\;H^{-1}\Xi H^{-1}\bigr).
\end{equation}
On a coordinate $k$ with $[H^{-1}\Xi H^{-1}]_{kk}>0$ almost surely, the ratio map is almost surely continuous at the limit, so the continuous-mapping theorem yields
\begin{equation}
  \hat t_k\;\convd\;\frac{[H^{-1}S_\alpha(\Lambda)]_k}   {\bigl([H^{-1}\Xi H^{-1}]_{kk}\bigr)^{1/2}} \;\equiv\;J^\star_\alpha .
\end{equation}
Numerator and denominator are built from the same jointly stable pair \eqref{eq:joint-stable}; the law of $J^\star_\alpha$ is continuous by Lemma~\ref{lem:atomless}(i), and by
Lemma~\ref{lem:atomless}(ii), which projects the point process onto the coordinate $k$, it is the scalar self-normalized stable limit of \citet{LoganMallowsRiceShepp1973} with tail balance $(c_{k,+},c_{k,-})$ and is not $\mathcal{N}(0,1)$ for any $\alpha\in(1,2)$. This proves part~(iv).

Combining Steps~1--5 completes the proof.
\end{proof}

\subsection{Proof of Theorem~\ref{thm:valid}} \label{app:valid-proof}


\begin{proof}[Proof of Theorem~\ref{thm:valid}]
The proof proceeds in three steps:
(1)~establish the limiting distribution of the full-sample
self-normalized $t$-statistic; (2)~establish the same limit for the
subsample statistic (Lemma~\ref{lem:subnet-inherit}); and (3)~obtain
coverage by the vertex subsampling consistency of
Lemma~\ref{lem:vertex-subsample}. Throughout, we focus without loss of
generality on the coordinate $k=1$ and write $\delta_0\equiv\beta_{0,1}$,
$\hat\delta\equiv\hat\beta_1$, with analogous notation
$\hat\delta^{(j)}$ for the $j$-th subsample estimator on the induced
pair set $\mathcal S_j$ of size $n_G$.

\emph{Step~1: Limiting distribution of $T_n$.}
The PPML linearization \eqref{eq:linearization-app} gives
$\hat\beta-\beta_0=(n\bar H_n^{\ast})^{-1}(\sum_\ell s_\ell+o_p(a_n))$ with
$\bar H_n^{\ast}\convp H$ by \eqref{eq:Hbar-conv}. Define the self-normalized statistic
\begin{equation}
  \mathrm{SN}_{n}
  \;=\;\frac{\iota_1'(nH_n)^{-1}\sum_{\ell=1}^n s_\ell}
            {\bigl(\iota_1'(nH_n)^{-1}\sum_{\ell=1}^n
              \hat s_\ell(\hat\beta)\hat s_\ell(\hat\beta)'(nH_n)^{-1}\iota_1\bigr)^{1/2}},
  \label{eq:SN-defn}
\end{equation}
where $\iota_1$ is the first standard basis vector of $\R^p$ and
$\hat s_\ell(\hat\beta)$ is the estimated score process of
Appendix~\ref{app:profiling} at the estimates. Write $T_n\equiv T_n(\delta_0)$
for the full-sample statistic \eqref{eq:tstat} at $k=1$; then
$T_n=\mathrm{SN}_n+o_p(1)$. To see this, note three facts. By
\eqref{eq:linearization-app} the numerator of $T_n$ equals
$\iota_1'(n\bar H_n^{\ast})^{-1}(\sum_\ell s_\ell+o_p(a_n))$. The two
statistics share the score covariance, and
$\bar H_n^{\ast},\hat H_n,H_n\to_pH$. The common studentizer is
$\Theta_p(a_n/n)$ with an almost surely positive rescaled limit. The
Hessian substitutions and the $o_p(a_n)$ remainder therefore move the
ratio by $o_p(1)$.

The numerator and denominator of \eqref{eq:SN-defn} are the score
sum and the score covariance, normalized by $a_n^{-1}$ and $a_n^{-2}$. By the
joint stable limit \eqref{eq:joint-stable}, established in Step~5 of
the proof of Theorem~\ref{thm:stable},
\begin{equation*}
  \biggl(a_n^{-1}\sum_{\ell=1}^n s_\ell,\;
         a_n^{-2}\sum_{\ell=1}^n s_\ell s_\ell'\biggr)
  \;\convd\;\bigl(S_\alpha(\Lambda),\;\Xi\bigr),
\end{equation*}
with $\Lambda=C_\alpha\bar A\,\Gamma$ and
$\Xi=S_{\alpha/2}(\Lambda_2)$, $\Lambda_2=C_{\alpha/2}\bar A\,\Gamma_2$
the limits of Steps~2 and~5 there, jointly because both coordinates
are functionals of the same point process.
By the continuous mapping theorem applied to the ratio in \eqref{eq:SN-defn} (which is well defined because the denominator is almost-surely strictly positive),
\begin{equation}
  \mathrm{SN}_{n}\;\convd\;
  \frac{\iota_1'H^{-1}S_\alpha(\Lambda)}
       {\bigl(\iota_1'H^{-1}\Xi
              H^{-1}\iota_1\bigr)^{1/2}}
  \;\equiv\;J^\star_\alpha.
  \label{eq:SN-J}
\end{equation}
The limit $J^\star_\alpha$ has a continuous distribution function by Lemma~\ref{lem:atomless}, with $\int_{\mathbb S^{p-1}}(\iota_1'H^{-1}\theta)^2\,\Gamma(d\theta)>0$ by
Assumption~\ref{ass:dgp}(c) and $H\succ0$.
In the finite-variance case of Remark~\ref{rem:alpha2}, it follows that $T_n\convd\mathcal N(0,1)$ directly, and the argument of Step~3 below applies with $J^\star_2=\mathcal N(0,1)$.

\emph{Step~2: Limit of the subsample statistic.}
Two subsample statistics must be distinguished: the infeasible
truth-centered statistic
$T^{(j)}_0\equiv(\hat\delta^{(j)}-\delta_0)/\widehat{\mathrm{se}}^{(j)}$,
the subsample analog of $\mathrm{SN}_n$ on the induced country
subnetwork $\mathcal S_j$, and the feasible statistic
$T^{(j)}\equiv(\hat\delta^{(j)}-\hat\delta)/\widehat{\mathrm{se}}^{(j)}$
computed by the procedure of Section~\ref{sec:remedy-procedure}.
Under the subnetwork operator-norm condition \eqref{eq:design-op-sub}, Lemma~\ref{lem:subnet-inherit} applies. The within-subnetwork tail spectral measure and Hessian converge to $\Gamma$ and $H$. The subnetwork profile-score and score covariance remainders are negligible. Repeating the $T_n=\mathrm{SN}_n+o_p(1)$ argument of Step~1, with $(n,N)$ replaced by $(n_G,G)$ and the comparisons supplied by Lemma~\ref{lem:subnet-inherit}(iii), gives
\[
  T^{(j)}_0\;=\;T_{\mathcal S_j}+o_p(1),
\]
with $T_{\mathcal S_j}$ the self-normalized score statistic of Lemma~\ref{lem:subnet-inherit}(iv); consequently $T^{(j)}_0\convd J^\star_\alpha$, the same limit as in Step~1, where the convergence is with respect to the joint law of the shocks and the draw.

\emph{Step~3: Subsampling consistency and coverage.}
The subsampling units are the $N$ countries, drawn uniformly without replacement, with effective sampling fraction $n_G/n=(G/N)^2(1+o(1))\to0$.
Apply Lemma~\ref{lem:vertex-subsample} with $Z_{\mathcal C}$ equal to the truth-centered statistic $T^{(j)}_0$; its joint convergence over shocks and draw is precisely the conclusion of Step~2, and it is measurable with respect to the shocks on the drawn subnetwork. The Monte Carlo subsampling distribution $\widehat F_{n,G,M}$ of $\{T^{(j)}_0\}$ therefore converges uniformly in probability to the distribution function $F_{J^\star_\alpha}$ of $J^\star_\alpha$.
The centering difference between $T^{(j)}_0$ (centered at $\delta_0$) and the feasible statistic $T^{(j)}$ (centered at $\hat\delta$) is asymptotically negligible: since $\hat\delta-\delta_0=O_p(a_n/n)$ and the subsample studentizer is of order $\widehat{\mathrm{se}}^{(j)}\asymp a_{n_G}/n_G$ in probability,
\[
  \frac{\hat\delta-\delta_0}{\widehat{\mathrm{se}}^{(j)}}
  =O_p\Bigl(\frac{a_n/n}{a_{n_G}/n_G}\Bigr)
  =O_p\bigl((n_G/n)^{1-1/\alpha}\bigr)
  =O_p\bigl((G/N)^{2(1-1/\alpha)}\bigr)=o_p(1)
\]
for $\alpha>1$, using $a_n/a_{n_G}=(n/n_G)^{1/\alpha}(1+o(1))$ under
Assumption~\ref{ass:dgp}(a). To transfer the recentering to the
distribution level, define the exceptional fraction
\[
  \pi_{n,G}(\varepsilon)
  \;:=\;\binom NG^{-1}\sum_{\mathcal C:|\mathcal C|=G}
  \1\Bigl\{\bigl|(\hat\delta-\delta_0)\big/
  \widehat{\mathrm{se}}_{\mathcal C}\bigr|>\varepsilon\Bigr\};
\]
its expectation is $\Pr_{\eta,\mathcal C}\bigl(|(\hat\delta-\delta_0)/ \widehat{\mathrm{se}}_{\mathcal C}|>\varepsilon\bigr)\to0$ by the display above, so $\pi_{n,G}(\varepsilon)\convp0$ by Markov's inequality. 
Let $F_{n,G}^c$ and $F_{n,G}^{c,\mathrm{feas}}$ denote the complete subset-average cdf of Lemma~\ref{lem:vertex-subsample} applied to $T_0^{(j)}$ and to $T^{(j)}$, respectively. 
Since the feasible and truth-centered statistics differ exactly by $(\hat\delta-\delta_0)/\widehat{\mathrm{se}}_{\mathcal C}$, the complete conditional cdfs satisfy, for every $t$ and
$\varepsilon>0$,
\[
  F^{c}_{n,G}(t-\varepsilon)-\pi_{n,G}(\varepsilon)
  \;\le\;F^{c,\mathrm{feas}}_{n,G}(t)
  \;\le\;F^{c}_{n,G}(t+\varepsilon)+\pi_{n,G}(\varepsilon),
\]
and letting $\varepsilon\downarrow0$ along the continuity of
$F_{J^\star_\alpha}$ gives
$\sup_t|F^{c,\mathrm{feas}}_{n,G}(t)-F_{J^\star_\alpha}(t)|\convp0$; the Monte Carlo step of Lemma~\ref{lem:vertex-subsample} applies to the feasible statistics, so the subsampling distribution of $\{T^{(j)}\}$ converges uniformly in probability to $F_{J^\star_\alpha}$.
Coverage now follows from three ingredients: the uniform convergence just established, $T_n\convd J^\star_\alpha$ from Step~1, and the continuity of $F_{J^\star_\alpha}$. The standard subsampling argument \citep[Theorem~2.1]{PolitisRomano1994} combines them. Hence $\Pr(\hat q_{\tau/2}\leq T_n\leq\hat q_{1-\tau/2})\to 1-\tau$ as $n\to\infty$, with $\hat q_{\tau/2},\hat q_{1-\tau/2}$ the empirical quantiles of $\{T^{(j)}\}$.
Translating back from $T_n$ to $\beta_{0,1}$ establishes the coverage statement of the theorem.
\end{proof}

\subsection{Firm-Level Primitives for Assumption~\ref{ass:dgp}(a)}\label{app:primitive-a}

Assumption~\ref{ass:dgp}(a) is a reduced-form condition on the pair-level shock, and this appendix records firm-level primitives under which it holds.
In the \citet{Melitz2003} member of the \citet{arkolakis&costinot&rodriguezclare2012} class, firm sales in a destination have a Pareto upper tail with index $\alpha=\vartheta/(\sigma-1)$, where $\vartheta$ is the productivity exponent and $\sigma$ the elasticity of substitution.
Market entry left-truncates the productivity draw at a pair-specific cutoff, and left truncation preserves a Pareto index, so the entry margins of \citet{chaney2008distorted} and \citet{arkolakis2010market} move the scale of pair $ij$'s sales distribution and the number of entrants, both of which enter $\mu_{ij}$, and leave $\alpha$ unchanged; the same holds when the productivity body is lognormal with a Pareto upper tail \citep{nigai2017tale}, since the cutoff then falls in the body.
The integrability condition $\vartheta>\sigma-1$ of that literature is exactly $\alpha>1$, and $\alpha<2$ is $\vartheta<2(\sigma-1)$, so $\vartheta=5$ and $\sigma=5$ give $\alpha=1.25$.

Aggregation to the country-pair level preserves the index.
Regular variation is closed under convolution (\citealp{Feller1971}, Section~VIII.8; \citealp{EmbrechtsGoldie1980}), so if pair $ij$ is served by $\mathcal M_{ij}$ exporters with independent sales in the domain of attraction of $\alpha$, then
\begin{equation}
  \Pr(\varepsilon_{ij}-1>t)\;=\;c\,\mathcal M_{ij}^{\,1-\alpha}\,t^{-\alpha}\,(1+o(1)),
  \label{eq:granular-tail}
\end{equation}
with $c$ free of $\mathcal M_{ij}$: aggregation rescales the shock by $\mathcal M_{ij}^{1/\alpha-1}$ without lightening its tail, and Assumption~\ref{ass:dgp}(a) takes that scale common across pairs.
Because $\alpha>1$, a continuum of firms with an exact law of large numbers would give $\varepsilon_{ij}\equiv1$, so the pair-level shock is a finite-$\mathcal M_{ij}$ phenomenon in the sense of \citet{Gabaix2011}, with $\mathcal M_{ij}$ bounded uniformly in $N$ by Assumption~\ref{ass:dgp}(b).
The attenuation is slow at the estimated indices, $\mathcal M_{ij}^{1/\alpha-1}$ equalling $0.59$ at $\alpha=1.06$ and $\mathcal M_{ij}=10^{4}$, and firm-level estimates place $\alpha$ near one \citep{Axtell2001,diGiovanniLevchenkoRanciere2011}.

\subsection{Primitive Conditions for Assumption~\ref{ass:dgp}(d)}
\label{app:primitive-d}

Assumption~\ref{ass:dgp}(d) is stated at a high level. 
This subsection records one further primitive design condition, combines it with the Gram conditioning \eqref{eq:dense-design} of Assumption~\ref{ass:dgp}(b) and the
fixed-effect rates of Appendix~\ref{app:lemmas}, and shows in Proposition~\ref{prop:primitive-d} that together they imply~(d) for every $\alpha\in(1,2)$.

\medskip\noindent\emph{Non-spiked residualized design.} Let
$A_n^{(k)}(\beta)$ be the $N\times N$ residualized weighted design with off-diagonal entries $\mu_{ij}(\beta)\tilde x_{ij}^{(k)}(\beta)$ and zero diagonal, the matrix through which regressor $k$ enters the fixed-effect remainder $u'A_n^{(k)}w$; write $A_n^{(k)}=A_n^{(k)}(\beta_0)$. For some
$\gamma<1$,
\begin{equation}
  \sup_{\beta\in\mathcal B}\max_{1\le k\le p}\bigl\|A_n^{(k)}(\beta)\bigr\|_{\mathrm{op}}=O\bigl(N^{\gamma}\bigr),
  \label{eq:design-op}
\end{equation}
with $\|\cdot\|_{\mathrm{op}}$ the spectral norm; the generic non-spiked rate is $\gamma=1/2$.

The operator-norm condition \eqref{eq:design-op} isolates the only component of the profiling remainder that is bilinear in the two fixed-effect error vectors. 
By the exact identity \eqref{eq:exact-Rn}, the $k$-th component of the fixed-effect remainder is the single bilinear form $R_n^{(k)}=u'A_n^{(k)}w$, so the fixed-effect errors $u$ and $w$ are controlled by the $\ell_q$ rates of Lemmas~\ref{lem:fe-rate} and~\ref{lem:fe-lp} and the design enters only through the operator norms of $A_n^{(k)}$.
The resulting bound holds for arbitrarily large $\max_\ell|d_\ell'\nu_n|$, which under Assumption~\ref{ass:dgp} is $O_p(\log N)$.
Note that this condition constrains the residualized covariates $\tilde x^{(k)}$, not the weights: even under $\mu_{ij}\equiv1$, a spiked covariate ($\tilde x_{ij}=v_iv_j$-type variation with $\sum_iv_i=0$) has $\|A_n^{(k)}\|_{\mathrm{op}}\asymp N$.
Conditions \eqref{eq:design-op} and \eqref{eq:design-op-sub} are therefore separate non-spiking restrictions on the residualized covariates, on the full network and on induced subnetworks respectively.

Relative to the finite-variance literature, the controlling mechanism differs.
\citet{FernandezVal2016} and \citet{WeidnerZylkin2021} isolate a deterministic $O(N^{-1})$ incidental-parameter bias through expectations of second-order remainder terms.
Under infinite variance those expectation-based expansions are generally unavailable, and the remainder is instead controlled probabilistically through the interpolation bound \eqref{eq:score-interp}: for every $\epsilon>0$ it is a stochastic $O_p\bigl(N^{2\gamma(1-1/\alpha)+4/\alpha-2+\epsilon}\bigr)$ term, so that relative to the score scale $a_n\asymp N^{2/\alpha}$,
\[
  R_n/a_n=O_p\bigl(N^{-2(1-\gamma)(1-1/\alpha)+\epsilon}\bigr),
\]
which vanishes throughout $\alpha\in(1,2)$ whenever $\gamma<1$. 
As $\alpha\uparrow2$ the relative remainder is $O_p(N^{-(1-\gamma)+\epsilon})$, recovering the $O(N^{-1})$ incidental-parameter rate when $\gamma=0$.

The remainder bounds themselves, Lemmas~\ref{lem:remainder} and~\ref{lem:general-beta}, are stated and proved in Appendix~\ref{app:lemmas}; the following proposition collects them.

\begin{proposition}[Primitive conditions for Assumption~\ref{ass:dgp}(d)] \label{prop:primitive-d}
Under Assumption~\ref{ass:dgp}(a),(b) and condition \eqref{eq:design-op} with $\gamma<1$, Assumption~\ref{ass:dgp}(d) holds in full for every $\alpha\in(1,2)$: the truth-profiled score covariance
condition in \eqref{eq:hl-remainder}, the profile-score condition $\sup_{\beta\in\mathcal B}\|R_n(\beta)\|=o_p(a_n)$, and the local-uniform profile-Hessian condition $\sup_{\|\beta-\beta_0\|\le\delta_n}\|H_n^{\ast}(\beta)-H\|\convp0$ for every $\delta_n\downarrow0$.
\end{proposition}

\begin{proof}
\emph{Score covariance.} The score covariance condition is Lemma~\ref{lem:remainder}(c), proved under \eqref{eq:dense-design} for every $\alpha\in(1,2)$.

\emph{Score.} The profile-score condition at $\beta_0$ is Lemma~\ref{lem:remainder}(d) and, uniformly in $\beta$, Lemma~\ref{lem:general-beta}; both follow from \eqref{eq:design-op} with $\gamma<1$ and the conditioning~\eqref{eq:dense-design}, for every $\alpha\in(1,2)$.

\emph{Hessian.} Fix $\delta_n\downarrow0$. For $\beta\in\mathcal B$ write $H^{\mathrm{hyb}}_n(\beta)=n^{-1}\sum_\ell\hat\mu_\ell(\beta)\tilde x_\ell(\beta) \tilde x_\ell(\beta)'$ (fitted weights, pseudo-true residuals) and $H^{\mathrm{pt}}_n(\beta)=n^{-1}\sum_\ell\mu_\ell(\beta)\tilde x_\ell(\beta) \tilde x_\ell(\beta)'$ (pseudo-true weights and residuals), and decompose
\[
  H_n^{\ast}(\beta)-H =\bigl(H_n^{\ast}(\beta)-H^{\mathrm{hyb}}_n(\beta)\bigr)
   +\bigl(H^{\mathrm{hyb}}_n(\beta)-H^{\mathrm{pt}}_n(\beta)\bigr)
   +\bigl(H^{\mathrm{pt}}_n(\beta)-H^{\mathrm{pt}}_n(\beta_0)\bigr)
   +\bigl(H^{\mathrm{pt}}_n(\beta_0)-H\bigr).
\]
The first term is the fitted-residualization error: $\sup_{\beta\in\mathcal B}\|H_n^{\ast}(\beta)-H^{\mathrm{hyb}}_n(\beta)\| =O_p(N^{1/q-1})$ for any fixed $q\in(1,\alpha)$, by
Lemma~\ref{lem:fitted-residual}. 
The second is the fitted-weight error: with $u_i(\beta),w_j(\beta)$ the multiplicative fixed-effect errors of Corollary~\ref{cor:uniform-beta},
\[
  H^{\mathrm{hyb}}_n(\beta)-H^{\mathrm{pt}}_n(\beta) =\frac1n\sum_{i\neq j}\mu_{ij}(\beta) \bigl(u_i(\beta)+w_j(\beta)+u_i(\beta)w_j(\beta)\bigr)\tilde x_{ij}(\beta)\tilde x_{ij}(\beta)',
\]
so, by $\|\tilde x_{ij}\tilde x_{ij}'\|\le C$ and the uniform $\ell_1$ rates $\sup_{\beta\in\mathcal B}(\|u(\beta)\|_1+\|w(\beta)\|_1)=O_p(N^{1/q})$ of Corollary~\ref{cor:uniform-beta},
\[
  \sup_{\beta\in\mathcal B}\bigl\|H^{\mathrm{hyb}}_n(\beta)-H^{\mathrm{pt}}_n(\beta)\bigr\|=O_p\bigl(N^{1/q-1}+N^{2/q-2}\bigr)=o_p(1).
\]
The third term is bounded by $C\delta_n$ uniformly over $\|\beta-\beta_0\|\le\delta_n$, by the uniform Lipschitz continuity of the bounded pseudo-true map $\beta\mapsto(\mu_\ell(\beta),\tilde x_\ell(\beta))$ of Lemma~\ref{lem:smooth-profile}; the fourth is $o(1)$ by Assumption~\ref{ass:dgp}(b), since $H^{\mathrm{pt}}_n(\beta_0)=H_n\to H$. Taking $\sup_{\|\beta-\beta_0\|\le\delta_n}$ across the decomposition gives the local-uniform Hessian condition, for every $\alpha\in(1,2)$.
\end{proof}

\subsection{Technical Lemmas}\label{app:lemmas}

This subsection collects the technical lemmas, in an order such that each proof invokes only results stated earlier.
Lemmas~\ref{lem:breiman}--\ref{lem:smooth-profile} and Corollary~\ref{cor:uniform-beta} are established under Assumption~\ref{ass:dgp}(a) and~(b) (Lemmas~\ref{lem:linf-fit}
and~\ref{lem:psi-facts} are deterministic); 
Lemmas~\ref{lem:remainder} and~\ref{lem:general-beta} additionally impose the design condition \eqref{eq:design-op} of Appendix~\ref{app:primitive-d};
Lemmas~\ref{lem:subnet-scaling}--\ref{lem:vertex-subsample} concern the induced subnetwork of Theorem~\ref{thm:valid}
(Lemma~\ref{lem:subnet-scaling} is deterministic; the rest concern the uniform draw).

We first record the multivariate form of Breiman's lemma, which extends the univariate result of \citet{Breiman1965} to vector-valued products of a heavy-tailed scalar and a light-tailed multiplier.

\begin{lemma}[Multivariate Breiman lemma]\label{lem:breiman}
Let $Z$ be a real-valued random variable with regularly varying right
tail, $\Pr(Z>t)=c\,t^{-\alpha}\mathcal L(t)$ for some $\alpha>0$, $c>0$,
and slowly varying $\mathcal L$, and with negligible left tail in the
sense that $\Pr(Z<-t)=o(\Pr(Z>t))$ as $t\to\infty$. Let
$\varkappa\in\R^p$ be a deterministic vector with
$0<\|\varkappa\|\leq C_\varkappa<\infty$. Then
\begin{equation}
  \frac{\Pr(\|Z\varkappa\|>t)}{\Pr(|Z|>t)}\;\longrightarrow\;\|\varkappa\|^\alpha
  \qquad\text{as }t\to\infty,
  \label{eq:breiman}
\end{equation}
and, more completely, $Z\varkappa$ is multivariate regularly varying
with limit measure concentrated on the ray through
$\varkappa/\|\varkappa\|$:
\begin{equation*}
  \lim_{t\to\infty}\frac{\Pr(Z\varkappa\in tA)}{\Pr(|Z|>t)}
  \;=\;\nu_\varkappa(A)
  \;:=\;\|\varkappa\|^\alpha
  \int_0^\infty\1\bigl\{r\,\varkappa/\|\varkappa\|\in A\bigr\}\,
  \alpha\,r^{-\alpha-1}\,dr
\end{equation*}
for every Borel set $A\subset\R^p$ bounded away from the origin with
$\nu_\varkappa(\partial A)=0$; in polar coordinates,
$\nu_\varkappa(dr,d\theta)
=\|\varkappa\|^\alpha\,\alpha r^{-\alpha-1}dr\,
\delta_{\varkappa/\|\varkappa\|}(d\theta)$. If both tails of $Z$ are regularly varying of
index $\alpha$ with constants $c_+$ and $c_-$, the limit measure is the
corresponding mixture supported on the two rays
$\pm\varkappa/\|\varkappa\|$; the one-sided statement is the case
$c_-=0$.
\end{lemma}

\begin{proof}
For the norm tail,
$\Pr(\|Z\varkappa\|>t)=\Pr(|Z|>t/\|\varkappa\|)$, and regular variation of the tail of $|Z|$ (the negative tail being negligible) gives $\Pr(|Z|>t/\|\varkappa\|)/\Pr(|Z|>t)\to\|\varkappa\|^\alpha$, which is \eqref{eq:breiman}. For the measure statement, $Z\varkappa\in tA$ if and only if $Z\in tB$ with $B:=\{s\in\R:s\varkappa\in A\}$; since $A$ is bounded away from the origin, $B$ is bounded away from zero, and univariate regular variation with negligible left tail gives $\Pr(Z\in tB)/\Pr(|Z|>t)\to\int_{B\cap(0,\infty)}\alpha s^{-\alpha-1}ds$ whenever the limit assigns no mass to $\partial B$. 
The substitution $s=r/\|\varkappa\|$ identifies this limit with $\nu_\varkappa(A)$. 
The two-sided version follows by splitting on the sign of $Z$. 
See \citet[Section~7.3]{Resnick1987} for the general form. 
In our application (Assumption~\ref{ass:dgp}(a)), $Z=\eta_\ell\geq -1$, so the left tail is bounded and the negligibility condition holds trivially.
\end{proof}

The following $\ell_\infty$ bound for weighted two-way fits is used repeatedly: for the smoothness of the population and sample profiling maps, for the score covariance bridge in Step~5 of the proof of Theorem~\ref{thm:stable}, and for the subnetwork comparison of Lemma~\ref{lem:subnet-residual}. 
The constant depends only on the ratio bounds of the kernel $r_{ij}$ and not on the multiplicative factors $p_i,q_j$, which may be unbounded.

\begin{lemma}[$\ell_\infty$ bound for weighted two-way fits]\label{lem:linf-fit}
Let $\{\omega_{ij}:i\neq j\}$ be positive weights admitting a factorization $\omega_{ij}=p_iq_jr_{ij}$ with $p_i,q_j>0$ and $r_{ij}\in[c_r,C_r]$, $0<c_r\le C_r<\infty$, and suppose
$\max_ip_i/\sum_{i'}p_{i'}\le\epsilon_N$ and $\max_jq_j/\sum_{j'}q_{j'}\le\epsilon_N$ with $\epsilon_N\le\epsilon_\ast:=c_r/(8C_r)$.
Let $g=(g_{ij})$ satisfy $|g_{ij}|\le C_g$, and let $(a,b)$ solve the $\omega$-weighted two-way normal equations
\[
  a_i=\frac{\sum_{j\neq i}\omega_{ij}(g_{ij}-b_j)}{\sum_{j\neq i}\omega_{ij}},
  \qquad
  b_j=\frac{\sum_{i\neq j}\omega_{ij}(g_{ij}-a_i)}{\sum_{i\neq j}\omega_{ij}},
\]
normalized by $b_N=0$. Then
\[
  \max_i|a_i|+\max_j|b_j|\;\le\;K_\ast\,C_g,
\]
where $K_\ast=K_\ast(c_r/C_r)<\infty$ does not depend on $N$, $g$, or the factors $(p_i,q_j)$.
\end{lemma}

\begin{proof}
Eliminating $a$ from the $b$-equation gives $b=Pb+h$, where
$P_{jl}=\sum_{i\notin\{j,l\}}
(\omega_{ij}/\tilde m^\omega_j)(\omega_{il}/m^\omega_i)$
with $m^\omega_i=\sum_{l\neq i}\omega_{il}$,
$\tilde m^\omega_j=\sum_{i\neq j}\omega_{ij}$ (the index $i=j$ is
excluded by the definition of $\tilde m^\omega_j$ and the index $i=l$
because the diagonal weight $\omega_{ll}$ does not exist), and $h_j$ a
two-fold $\omega$-weighted average of the $g$'s; $P$ is row stochastic,
since $\sum_lP_{jl}=\sum_{i\neq j}(\omega_{ij}/\tilde m^\omega_j)
\sum_{l\neq i}(\omega_{il}/m^\omega_i)=1$, and
$\|h\|_\infty\le2C_g$. For $i\neq l$ the factorization gives
\[
  \frac{\omega_{il}}{m^\omega_i}\ge\frac{c_r}{C_r}\,\pi_l, \qquad \pi_l:=\frac{q_l}{\sum_{l'}q_{l'}},
\]
hence, with the convention $\omega_{jj}:=0$,
\[
  P_{jl}\;\ge\;\frac{c_r}{C_r}\,\pi_l \sum_{i\notin\{j,l\}}\frac{\omega_{ij}}{\tilde m^\omega_j} \;=\;\frac{c_r}{C_r}\,\pi_l \Bigl(1-\frac{\omega_{lj}}{\tilde m^\omega_j}\Bigr) .
\]
The excluded-diagonal correction is uniformly small: for $l\neq j$, $\omega_{lj}/\tilde m^\omega_j \le C_rp_lq_j\big/\bigl(c_rq_j\sum_{i\neq j}p_i\bigr) \le(C_r/c_r)\,\epsilon_N/(1-\epsilon_N) \le2(C_r/c_r)\,\epsilon_N\le\tfrac14$ by $\epsilon_N\le\epsilon_\ast$. Therefore
\[
  P_{jl}\;\ge\;\tau_0\,\pi_l,\qquad \tau_0:=\tfrac34\,(c_r/C_r),
\]
a Doeblin minorization by the probability vector $\pi$ in which $(p_i,q_j)$ enter only through $\pi$ itself. 
Hence the Dobrushin coefficient satisfies $\delta(P)\le1-\tau_0$, and since $\osc(Pv)\le\delta(P)\,\osc(v)$ \citep[Section~4.3]{Seneta1981}, any solution of $b=Pb+h$ satisfies
$\osc(b)\le\osc(h)/\tau_0\le4C_g/\tau_0\le6(C_r/c_r)C_g$. With $b_N=0$, $\|b\|_\infty\le\osc(b)$, and the $a$-equation gives $\|a\|_\infty\le C_g+\|b\|_\infty$, so the claim holds with
$K_\ast=1+12(C_r/c_r)$.
\end{proof}

Three instances are used: (i)~$\omega_{ij}=\mu_{ij}(\beta)$ with $p_i\equiv q_j\equiv1$, $r=\mu(\beta)\in[c,C]$ uniformly on $\mathcal B$ by Assumption~\ref{ass:dgp}(b) and $\epsilon_N=O(1/N)$; (ii)~$\omega_{ij}=\hat\mu_{ij}(\beta) =\mu_{ij}(\beta)e^{\xi_i(\beta)}e^{\omega_j(\beta)}$ with $p_i=e^{\xi_i(\beta)}$, $q_j=e^{\omega_j(\beta)}$, $r=\mu(\beta)$, and
$\epsilon_N=o_p(1)$ uniformly over $\beta\in\mathcal B$ because $\max_ie^{\xi_i(\beta)}=O_p(N^{2/\alpha-1})$ while $\sum_ie^{\xi_i(\beta)}\ge Ne^{-C_0}$, by Lemma~\ref{lem:margins}(iii)
and Corollary~\ref{cor:uniform-beta}; (iii)~the restriction of $\mu$ to a $G$-country subnetwork, with $G$ in place of $N$.

The row and column sums of the shock array are triangular-array sums of independent but not identically distributed terms, because the deterministic weights vary within a row. The following uniform one-large-jump bound, in the spirit of \citet{Nagaev1979}, is the single tail estimate from which all block rates below are derived.

\begin{lemma}[Uniform row-sum tail bound]\label{lem:rowsum}
Let $\{\eta_m\}_{m\le M}$ be i.i.d.\ centered random variables with
$\Pr(|\eta_1|>t)\le C_\eta t^{-\alpha}$ for all $t\ge1$ and some
$\alpha\in(1,2)$, and let
$\{b_m\}$ be deterministic weights with
$|b_m|\le C_b$. There are constants $c_0,C_S$, depending only on
$(\alpha,C_\eta,C_b)$, such that
\begin{equation*}
  \Pr\Bigl(\Bigl|\sum_{m\le M}b_m\eta_m\Bigr|>s\Bigr)
  \;\le\;C_S\,M\,s^{-\alpha}
  \qquad\text{for all }s\ge c_0M^{1/\alpha}.
\end{equation*}
Consequently, for the score blocks
$S_i=\sum_{j\neq i}\mu_{ij}\eta_{ij}$ and
$\tilde S_j=\sum_{i\neq j}\mu_{ij}\eta_{ij}$ under
Assumption~\ref{ass:dgp}(a),(b):
\emph{(i)}~$\max_{i\le N}|S_i|+\max_{j\le N}|\tilde S_j|
=O_p(N^{2/\alpha})$;
\emph{(ii)}~$\sum_{i\le N}S_i^2+\sum_{j\le N}\tilde S_j^2
=O_p(N^{4/\alpha})$;
\emph{(iii)}~$\Pr(|S_i|>s)\le C_SNs^{-\alpha}$ for all $s\ge N$,
uniformly in $i$ and $N$.
Separately,
\emph{(iv)}~if $Z_1,\dots,Z_n\ge0$ are i.i.d.\ regularly varying with
index $\rho\in(0,1)$ and $b_n$ satisfies $n\Pr(Z_1>b_n)\to1$, then
$\sum_{m\le n}Z_m=O_p(b_n)$.
\end{lemma}

\begin{proof}
Write $X_m=b_m\eta_m$ and truncate at $s/2$:
$\Pr(|\sum_mX_m|>s)\le\sum_m\Pr(|X_m|>s/2)
+\Pr(|\sum_m\check X_m|>s/2)$ with
$\check X_m=X_m\1\{|X_m|\le s/2\}$. The first sum is at most
$MC_\eta(2C_b)^\alpha s^{-\alpha}$. For the second, integrating the
polynomial tail bound gives
$\E[X_m^2\1\{|X_m|\le u\}]\le2\int_0^ut\Pr(|X_m|>t)\,dt\le Cu^{2-\alpha}$
and
$\E[|X_m|\1\{|X_m|>u\}]\le u\Pr(|X_m|>u)
+\int_u^\infty\Pr(|X_m|>t)\,dt\le Cu^{1-\alpha}$
for all $u\ge1$, so the
centering shift satisfies
$|\sum_m\E\check X_m|=|\sum_m\E X_m\1\{|X_m|>s/2\}|\le CMs^{1-\alpha}
\le s/4$ whenever $s\ge c_0M^{1/\alpha}$ with $c_0$ large, and
Chebyshev's inequality yields
$\Pr(|\sum_m(\check X_m-\E\check X_m)|>s/4)
\le CM s^{2-\alpha}/s^{2}=CMs^{-\alpha}$.

\emph{(i)} The rows $S_i$ use disjoint shocks across $i$, and each is a weighted sum of $M=N-1$ terms; a union bound gives $\Pr(\max_i|S_i|>tN^{2/\alpha})\le N\cdot C_SN(tN^{2/\alpha})^{-\alpha} =C_St^{-\alpha}$ for $t\ge c_0'$, and likewise for the columns.

\emph{(ii)} Fix $t\ge1$ and truncate at $B=tN^{4/\alpha}$. 
By the tail bound, $\Pr(\max_iS_i^2>B)\le N\cdot C_SN\,B^{-\alpha/2}
=C_St^{-\alpha/2}$, while
\[
  \E\bigl[S_i^2\1\{S_i^2\le B\}\bigr] =\int_0^B\Pr(S_i^2>x)\,dx \le c_0^2N^{2/\alpha}+\int_{c_0^2N^{2/\alpha}}^{B}C_SNx^{-\alpha/2}\,dx\le C\,t^{1-\alpha/2}N^{4/\alpha-1},
\]
so $\sum_i\E[S_i^2\1\{S_i^2\le B\}]\le Ct^{1-\alpha/2}N^{4/\alpha}$ and Markov's inequality gives $\Pr(\sum_iS_i^2\1\{\cdot\}>tN^{4/\alpha})\le Ct^{-\alpha/2}$; combining the two events proves~(ii).

\emph{(iii)} is the displayed bound with $M=N-1$, since $N\ge c_0N^{1/\alpha}$ for $N$ large ($\alpha>1$).

\emph{(iv)} Fix $t\ge1$ and truncate at $tb_n$. 
First, $\Pr(\max_mZ_m>tb_n)\le n\Pr(Z_1>tb_n)\to t^{-\rho}$. Second, by Karamata's theorem, $\E[Z_1\1\{Z_1\le tb_n\}] \le C\,tb_n\Pr(Z_1>tb_n)/(1-\rho)$, so
$\sum_m\E[Z_m\1\{Z_m\le tb_n\}]\le Ct^{1-\rho}b_n/(1-\rho)$ by $n\Pr(Z_1>b_n)\to1$ and regular variation; Markov's inequality bounds the truncated sum. 
Combining the two events gives $\sum_mZ_m=O_p(b_n)$.
\end{proof}

\begin{lemma}[Margin regularity]\label{lem:margins}
Under Assumption~\ref{ass:dgp}(a) and~(b), w.p.a one:
\begin{enumerate}
\item[\emph{(i)}] the observed margins $R_i=\sum_{j\neq i}y_{ij}$, $\tilde R_j=\sum_{i\neq j} y_{ij}$ satisfy $\min_i R_i\ge c'N$ and $\min_j \tilde R_j\ge c'N$ for some $c'>0$;
\item[\emph{(ii)}] the total mass $T=\sum_{i\neq j}y_{ij}$ satisfies $T=\Theta_p(N^2)$, and the block maxima satisfy $\max_iR_i,\max_j\tilde R_j=O_p(N^{2/\alpha})$;
\item[\emph{(iii)}] under $\sum_j e^{\omega_j}=N-1$ the fixed-effect errors are bounded below uniformly, $\min_i\xi_i\ge-C_0$ and $\min_j\omega_j\ge-C_0$, and bounded above by $\max_ie^{\xi_i},\max_je^{\omega_j}=O_p(N^{2/\alpha-1})=o_p(N)$, so that $m_i^{(\omega)}\!:=\!\sum_j\mu_{ij}e^{\omega_j}\asymp N$ and $\tilde m_j^{(\xi)}\!:=\!\sum_i\mu_{ij}e^{\xi_i}\asymp N$ uniformly.
\end{enumerate}
\end{lemma}
\begin{proof}
\emph{(i) Margin lower bounds.}
The $\varepsilon_{ij}=1+\eta_{ij}\ge0$ are i.i.d.\ with $\E\varepsilon=1$ (finite since $\alpha>1$). As $\E[\min(\varepsilon,K)]\uparrow1$, fix $K$ with $\E[\min(\varepsilon,K)]\ge3/4$. The $\min(\varepsilon_{ij},K)\in[0,K]$ are bounded and independent across $j$, so Hoeffding's inequality gives, for each $i$, $\Pr\bigl(\sum_j\min(\varepsilon_{ij},K)<\tfrac{N-1}{2}\bigr)
\le\exp(-(N-1)/(8K^2))$. 
A union bound over the $2N$ rows and columns sends the probability that any margin falls below $(N-1)/2$ to zero; hence $\min_iR_i\ge c_\mu(N-1)/2=:c'N$ and $\min_j\tilde R_j\ge c'N$ w.p.a.\ one.

\emph{(ii) Block maxima and total mass.}
Write $R_i=m_i+S_i$ with $m_i=\sum_j\mu_{ij}=O(N)$ and $S_i=\sum_j\mu_{ij}\eta_{ij}$. The weights are deterministic and bounded, so Lemma~\ref{lem:rowsum}(i) gives
$\max_i|S_i|=O_p(N^{2/\alpha})$, whence $\max_iR_i,\max_j\tilde R_j=O(N)+O_p(N^{2/\alpha})=O_p(N^{2/\alpha})$ since $2/\alpha>1$. 
For the total mass, write $T=\sum_im_i+\sum_iS_i$. 
The deterministic part lies in $[c_\mu N(N-1),\,C_\mu N(N-1)]$. 
The shock part is $O_p(N^{2/\alpha})=o_p(N^2)$, by Lemma~\ref{lem:rowsum} applied to the single weighted sum of all $n$ shocks. 
Hence there are fixed constants $0<c_T<C_T<\infty$ with $\Pr(c_TN^2\le T\le C_TN^2)\to1$; in particular $T=\Theta_p(N^2)$, and the two-sided bound holds with fixed constants w.p.a.\ one, which is used in part~(iii).

\emph{(iii) Fixed-effect bounds.}
The fixed-effect first-order conditions reproduce the margins, $e^{\xi_i}m_i^{(\omega)}=R_i$ and $e^{\omega_j}\tilde m_j^{(\xi)}=\tilde R_j$, with $m_i^{(\omega)}=\sum_j\mu_{ij}e^{\omega_j}$ and $\tilde m_j^{(\xi)}=\sum_i\mu_{ij}e^{\xi_i}$. We argue in three steps.

\emph{Lower bound on exporter effects.} Under the normalization $\sum_je^{\omega_j}=N-1$, $m_i^{(\omega)}\le C_\mu(N-1)$, so
\begin{equation}
  e^{\xi_i}=\frac{R_i}{m_i^{(\omega)}}\ge\frac{c'N}{C_\mu(N-1)}\ge\frac{c'}{C_\mu}=:e^{-C_0}>0\quad\text{for every }i,
  \label{eq:xi-lower}
\end{equation}
giving $\min_i\xi_i\ge-C_0$.

\emph{No importer weight dominates: $\max_je^{\omega_j}=O_p(N^{2/\alpha-1})$.}
Since $\sum_je^{\omega_j}=N-1$, at most one index $j_0$ can have $e^{\omega_{j_0}}>(N-1)/2$. For every importer $j$ (including $j_0$), the denominator uses only the $N-1$ retained exporters and the lower bound \eqref{eq:xi-lower}, $\tilde m_j^{(\xi)}=\sum_{i\neq j}\mu_{ij}e^{\xi_i} \ge c_\mu(N-1)e^{-C_0}=\Theta(N)$, so
$e^{\omega_j}=\tilde R_j/\tilde m_j^{(\xi)}\le\max_j\tilde R_j/\Theta(N) =O_p(N^{2/\alpha-1})$ by~(ii). For $\alpha>1$ this is $o(N)<(N-1)/2$ w.p.a.\ one, contradicting $e^{\omega_{j_0}}>(N-1)/2$; hence no dominant index exists and the bound holds for all $j$.

\emph{Mass equivalences and importer lower bound.} 
With $\max_je^{\omega_j}=o(N)$, for every exporter $i$, $m_i^{(\omega)}=\sum_{j\neq i}\mu_{ij}e^{\omega_j} \ge c_\mu\bigl(N-1-\max_je^{\omega_j}\bigr)\ge c_\mu(N-1)/2=\Theta(N)$
w.p.a.\ one, so $m_i^{(\omega)}\asymp N$; symmetrically $\tilde m_j^{(\xi)}\asymp N$. 
Then $e^{\omega_j}=\tilde R_j/\tilde m_j^{(\xi)}\ge c'N/(C_\mu A)$ with $A=\sum_ie^{\xi_i}\le T/\min_im_i^{(\omega)}$. 
By~(ii), $T\le C_TN^2$ w.p.a.\ one with a fixed $C_T$, and $\min_im_i^{(\omega)}\ge c_mN$ w.p.a.\ one from the previous display, so $A\le C_AN$ w.p.a.\ one with the fixed constant $C_A=C_T/c_m$.
Hence $e^{\omega_j}\ge c'/(C_\mu C_A)=:c''>0$ for all $j$ w.p.a.\ one, giving $\min_j\omega_j\ge-C_0$. Finally, the upper bounds $\max_ie^{\xi_i},\max_je^{\omega_j}=O_p(N^{2/\alpha-1})$ follow from $e^{\xi_i}=R_i/m_i^{(\omega)}\le\max_iR_i/\Theta(N)$ and the symmetric importer bound, using (ii).
\end{proof}

The implicit function theorem requires the sample profiling map $\hat\phi(\beta)$ to exist as a finite maximizer; strict concavity gives uniqueness on the quotient, not existence. Existence is a matrix-scaling property of the observed margins and follows from Lemma~\ref{lem:margins}.

\begin{corollary}[Existence of the sample profiling solution] \label{cor:sample-existence}
Under Assumption~\ref{ass:dgp}(a),(b),
\[
  \Pr\Bigl(\,\forall\beta\in\mathcal B:\ \exists\, \hat\phi(\beta)\in\R^{2N}\ \text{with}\ S_{\phi,n}\bigl(\beta,\hat\phi(\beta)\bigr)=0\Bigr) \;\longrightarrow\;1 ,
\]
and on this event the solution set at each $\beta$ is exactly $\hat\phi(\beta)+\mathrm{span}(e)$.
\end{corollary}

\begin{proof}
Fix $\beta$. 
The equation matches the observed margins: it requires $\sum_{j\neq i}\mu_{ij}(\beta,\phi)=R_i$ for every $i$ and $\sum_{i\neq j}\mu_{ij}(\beta,\phi)=\tilde R_j$ for every $j$, with
base kernel $e^{x_{ij}'\beta}>0$. 
The margins do not depend on $\beta$. 
The support $K_{N,N}\setminus\{(i,i)\}$ is fully indecomposable for $N\ge3$. For such a support, a finite scaling matching prescribed margins exists if and only if the margins are
strictly positive and lie in the relative interior of the margin cone; here the binding cuts are the single-vertex ones, so the strict conditions reduce to $\min_iR_i\wedge\min_j\tilde R_j>0$ and $R_i+\tilde R_i<T$ for every $i$ \citep{Sinkhorn1967,RothblumSchneider1989,Idel2016}. 
By Lemma~\ref{lem:margins}, w.p.a.\ one $\min_iR_i\wedge\min_j\tilde R_j\ge c'N>0$ and $\max_i(R_i+\tilde R_i)=O_p(N^{2/\alpha})=o_p(T)$, since $T=\Theta_p(N^2)$ and $2/\alpha<2$. On this event, which does not depend on $\beta$, a finite solution exists for every $\beta\in\mathcal B$; uniqueness on the quotient follows from strict concavity of $\phi\mapsto Q_n(\beta,\phi)$ there.
\end{proof}

\begin{lemma}[Truncated margin moments]\label{lem:tailmass}
Under Assumption~\ref{ass:dgp}(a) and~(b), for every fixed $q\in[1,\alpha)$ there exist constants $C_q<\infty$ and $T_0<\infty$ such that, for all $T_\ast\ge T_0$,
\begin{equation}
  \E\Bigl[\,\sum_{i}R_i^{\,q}\,\1\{R_i>T_\ast N\}\Bigr] \;\le\;C_q\,T_\ast^{\,q-\alpha}\,N^{2+q-\alpha},
  \label{eq:tailmass}
\end{equation}
and the same bound holds for the importer margins $\tilde R_j$.
\end{lemma}
\begin{proof}
Take $T_0 = \max\{4C_\mu,2\}$, so that $R_i>T_\ast N$ forces $S_i=R_i-m_i>T_\ast N/2$ and $R_i\le2S_i$ on that event. 
Each $S_i$ is a centered sum of $N-1$ independent summands $\mu_{ij}\eta_{ij}$ with bounded weights and common regularly varying right tail of index $\alpha$, so Lemma~\ref{lem:rowsum}(iii) gives $\Pr(S_i>s)\le C\,N\,s^{-\alpha}$ for all $s\ge N$, uniformly in $i$. Then, for $s=T_\ast N/2$,
\[
  \E\bigl[S_i^{\,q}\,\1\{S_i>s\}\bigr]=s^q\,\Pr(S_i>s)+q\int_s^\infty x^{q-1}\Pr(S_i>x)\,dx\le C\Bigl(1+\frac{q}{\alpha-q}\Bigr)N\,s^{\,q-\alpha},
\]
and summing $\E[R_i^q\1\{R_i>T_\ast N\}]\le2^q\,\E[S_i^q\1\{S_i>T_\ast N/2\}]$ over the $N$ exporters gives \eqref{eq:tailmass}. 
The importer bound is symmetric.
\end{proof}

\begin{lemma}[Fixed-effect conditioning and rate]\label{lem:fe-rate} Under Assumption~\ref{ass:dgp}(a) and~(b), w.p.a.\ one the integral-averaged Hessian $\bar H_{\phi,n}=\sum_\ell\mu_\ell\psi_\ell\,d_\ell d_\ell'$ with $\psi_\ell=\int_0^1 e^{s\,d_\ell'\nu_n}ds=(e^{d_\ell'\nu_n}-1)/(d_\ell'\nu_n)$ satisfies $\lambda_{\min}(N^{-1}\bar H_{\phi,n}\restriction_{e^\perp})\ge c_4>0$, the exact identity $\bar H_{\phi,n}\nu_\perp=S_\phi^0$ holds on $e^\perp$ with $S_\phi^0:=S_{\phi,n}(\beta_0,\phi_0)=\sum_\ell\mu_\ell\eta_\ell d_\ell$, the level of the Sinkhorn representative satisfies $t_\nu=O_p(1)$, and consequently, for the orthogonal representative,
\begin{equation}
  \sum_i\xi_i^2+\sum_j\omega_j^2=O_p(N^{4/\alpha-2}), \qquad \sum_iu_i^2+\sum_jw_j^2=O_p(N^{4/\alpha-2}).
  \label{eq:fe-rates}
\end{equation}
\end{lemma}
\begin{proof}
Throughout the proof, $(\xi^{\mathrm{sk}},\omega^{\mathrm{sk}})$ and $(\xi_\perp,\omega_\perp)$ denote the exporter and importer components of $\nu^{\mathrm{sk}}$ and $\nu_\perp$, with multiplicative errors $u_{\perp,i}=e^{\xi_{\perp,i}}-1$ and $w_{\perp,j}=e^{\omega_{\perp,j}}-1$.

\emph{Exact linear system on the quotient.} The fixed-effect first-order condition $\sum_\ell\mu_\ell(e^{d_\ell'\nu_n}-1)d_\ell=S_\phi^0$, combined with
the identity $e^{z}-1=z\int_0^1e^{sz}ds$ and the invariance of $d_\ell'\nu_n$ to the representative, gives $\bar H_{\phi,n}\nu_\perp=S_\phi^0$ exactly. 
Both sides are orthogonal to $e$: $\bar H_{\phi,n}e=0$ since $d_\ell'e=0$, and $\langle S_\phi^0,e\rangle=\sum_iS_i-\sum_j\tilde S_j=0$.

\emph{Conditioning.} By Lemma~\ref{lem:margins}(iii), $d_\ell'\nu_n=\xi_i+\omega_j\ge-2C_0$, and since $z\mapsto\int_0^1e^{sz}ds=\psi(z)$ is positive and increasing, $\psi_\ell\ge\psi(-2C_0)>0$, so $\mu_\ell\psi_\ell\ge c_\mu\psi(-2C_0)=:c_3$. 
Hence $\bar H_{\phi,n}\succeq c_3\sum_\ell d_\ell d_\ell' \succeq c_3C_\mu^{-1}H_{\phi\phi,n}$ on $e^\perp$, and the dense-design condition~\eqref{eq:dense-design} gives
$\lambda_{\min}(N^{-1}\bar H_{\phi,n}\restriction_{e^\perp})\ge c_3C_\mu^{-1}c=:c_4>0$, so $\|\nu_\perp\|\le c_4^{-1}N^{-1}\|S_\phi^0\|$.

\emph{Score-norm bound.} The score
$S_\phi^0=\sum_\ell\mu_\ell\eta_\ell d_\ell$
has exporter and importer blocks $S_i=\sum_j\mu_{ij}\eta_{ij}$ and
$\tilde S_j=\sum_i\mu_{ij}\eta_{ij}$. By Lemma~\ref{lem:rowsum}(ii),
$\sum_iS_i^2+\sum_j\tilde S_j^2=O_p(N^{4/\alpha})$. Thus
$\|S_\phi^0\|^2=O_p(N^{4/\alpha})$,
and the conditioning bound yields
$\|\nu_\perp\|=O_p(N^{2/\alpha-1})$, hence the first rate in
\eqref{eq:fe-rates}.

\emph{Level of the Sinkhorn representative.} Write $\nu^{\mathrm{sk}}=\nu_\perp+t_\nu \,e$, so $2Nt_\nu=\sum_i\xi^{\mathrm{sk}}_{i}-\sum_j\omega^{\mathrm{sk}}_{j}$. 
By the margin identities and Lemma~\ref{lem:margins}(iii), $e^{\xi^{\mathrm{sk}}_{i}}=R_i/m_i^{(\omega^{\mathrm{sk}})}$ with $m_i^{(\omega^{\mathrm{sk}})}\asymp N$
and $R_i=m_i+S_i\in[c'N,\,C_\mu N+|S_i|]$, so $|\xi^{\mathrm{sk}}_{i}|\le C+|S_i|/(cN)$. 
Since $\E|S_i|\le\sum_j\mu_{ij}\E|\eta|\le C_\mu N\,\E|\eta|=O(N)$, Markov's inequality gives $\sum_i|\xi^{\mathrm{sk}}_{i}|=O_p(N)$, and symmetrically $\sum_j|\omega^{\mathrm{sk}}_{j}|=O_p(N)$; hence $t_\nu=O_p(1)$. 
Fix $\varepsilon>0$ and choose $C_t=C_t(\varepsilon)$ with
$\limsup_N\Pr(|t_\nu|>C_t)<\varepsilon$; since the conclusions below are $O_p$ statements and $\varepsilon$ is arbitrary, it suffices to work on $\{|t_\nu|\le C_t\}$, with constants depending on $C_t$. 
On this event the margin bounds transfer to the orthogonal representative up to the factor $e^{C_t}$: $\min_i\xi_{\perp,i}\ge-C_0-C_t$, $\min_j\omega_{\perp,j}\ge-C_0-C_t$,
$m_i^{(\omega_\perp)}=e^{t_\nu}m_i^{(\omega^{\mathrm{sk}})}\asymp N$, and $1+u_{\perp,i}=e^{-t_\nu}R_i/m_i^{(\omega^{\mathrm{sk}})}\le e^{C_t}CR_i/N$.

\emph{From additive to multiplicative errors.} Work on $\{|t_\nu|\le C_t\}$ and split exporters at $\xi_{\perp,i}\le K'$ and $\xi_{\perp,i}>K'$, with $K'=2C_t+K_0$ and $K_0$ chosen so that $\xi_{\perp,i}>K'$ forces $R_i>T_\ast N$ with $T_\ast\ge T_0$ of Lemma~\ref{lem:tailmass} and $S_i=R_i-m_i\ge c'''N$. 
On $\{\xi_{\perp,i}\le K'\}$, $\xi_{\perp,i}\in[-C_0-C_t,K']$ and $|u_{\perp,i}|=|e^{\xi_{\perp,i}}-1|\le C_1|\xi_{\perp,i}|$ with $C_1=\sup_{[-C_0-C_t,K']}|e^x-1|/|x|<\infty$, so
$\sum_{\xi_{\perp,i}\le K'}u_{\perp,i}^2\le C_1^2\sum_i\xi_{\perp,i}^2=O_p(N^{4/\alpha-2})$. On $\{\xi_{\perp,i}>K'\}$, $u_{\perp,i}\le e^{C_t}CR_i/N$ and $\sum_{\xi_{\perp,i}>K'}u_{\perp,i}^2\le CN^{-2}\sum_{i:S_i\ge c'''N}S_i^2 \le CN^{-2}\sum_iS_i^2=O_p(N^{4/\alpha-2})$. 
The importer bound is symmetric, giving the second rate in \eqref{eq:fe-rates}.
\end{proof}

\begin{lemma}[Shift comparability of $\psi$]\label{lem:psi-facts}
Let $\psi(z)=(e^z-1)/z$ for $z\neq0$, $\psi(0)=1$.
For every $C_0>0$ and $K>0$ there exist constants $0<c_{K}\le C_{K}<\infty$ and $C_\psi<\infty$, depending only on $C_0$ and $K$, such that, for all $z\ge-2C_0$:
\emph{(i)}~$c_{K}\,\psi(z)\le\psi(z+k)\le C_{K}\,\psi(z)$ for every $k\in[-C_0,K]$;
\emph{(ii)}~$\psi(z+k)\le C_\psi\,e^{k}\,\psi(z)$ for every $k\ge-C_0$.
\end{lemma}

\begin{proof}
Since $\psi(z)=\int_0^1e^{sz}\,ds$, it is positive, continuous, and strictly increasing on $\R$, with derivative $\int_0^1se^{sz}\,ds>0$.

For~(i), the map $(z,k)\mapsto\psi(z+k)/\psi(z)$ is continuous and positive on $[-2C_0,\infty)\times[-C_0,K]$, and as $z\to\infty$ it converges to $e^{k}\in[e^{-C_0},e^{K}]$ uniformly in $k$, since $\psi(z+k)/\psi(z)=e^{k}\bigl(z/(z+k)\bigr)(1+o(1))$.
It is therefore bounded above, and bounded below by a positive constant, on the whole domain: by compactness on $[-2C_0,Z]\times[-C_0,K]$, and by the uniform limit on $[Z,\infty)\times[-C_0,K]$ for $Z$ large.
Take $c_K$ and $C_K$ to be these bounds.

For~(ii), suppose first that $k\in[-C_0,0)$.
Monotonicity gives $\psi(z+k)\le\psi(z)$, and $e^{-k}\le e^{C_0}$, so $\psi(z+k)\le e^{C_0}e^{k}\psi(z)$.
Now suppose $k\ge0$.
If $z+k\le1$, then $\psi(z+k)\le\psi(1)$ and $\psi(z)\ge\psi(-2C_0)$, so the ratio is at most $\psi(1)/\psi(-2C_0)$, which is at most $C_\psi e^{k}$ because $e^{k}\ge1$.
If $z+k>1$, then $z+k\ge\max(z,1)$, since $k\ge0$ gives $z+k\ge z$ and the case condition gives $z+k>1$, so $\psi(z+k)\le e^{z+k}/(z+k)\le e^{k}e^{z}/\max(z,1)\le C_\psi e^{k}\psi(z)$, using $\psi(z)\ge c\,e^{z}/\max(z,1)$ for $z\ge-2C_0$.
That last bound holds because $z\psi(z)/e^{z}=1-e^{-z}\ge1-e^{-1}$ for $z\ge1$, while $\psi(z)/e^{z}$ is continuous and positive on the compact interval $[-2C_0,1]$ and hence bounded below there.
Taking $C_\psi$ to be the largest of the constants produced gives~(ii).
\end{proof}

For the $\ell_q$ theory, let $P_e=I-ee'/\|e\|^2$ and define the quotient inverse $\mathcal G_\phi$ of $\bar H_{\phi,n}$ by: $\mathcal G_\phi b\in e^\perp$ solves
$\bar H_{\phi,n}(\mathcal G_\phi b)=P_eb$. $\mathcal G_\phi$ is symmetric, and $\nu_\perp=\mathcal G_\phi S_\phi^0$.

\begin{lemma}[$\ell_q$ conditioning of the averaged Hessian] \label{lem:linf-hessian}
Under Assumption~\ref{ass:dgp}(a),(b), with probability approaching one,
\[
  \|\mathcal G_\phi\|_{\infty\to\infty}\le\frac{C_J}{N},
  \qquad\text{and consequently}\qquad
  \|\mathcal G_\phi\|_{q\to q}\le\frac{C_J}{N}\quad\text{for every }q\in[1,2],
\]
for a constant $C_J=C_J(c,C,C_0)$.
\end{lemma}

\begin{proof}
Write $\tilde\omega_{ij}=\mu_{ij}\psi(\xi_i+\omega_j)>0$ for the weights of $\bar H_{\phi,n}$, where $(\xi,\omega)$ denote the Sinkhorn representative; the index $\xi_i+\omega_j=d_\ell'\nu_n$ is representative-free.
By Lemma~\ref{lem:margins}(iii), $\xi_i,\omega_j\ge-C_0$, so with $c_3':=c_\mu\psi(-2C_0)/2$ we have $m_i^{\tilde\omega}:=\sum_{j\neq i}\tilde\omega_{ij}\ge c_\mu\psi(-2C_0)(N-1)\ge c_3'N$ for $N\ge2$, and symmetrically $\tilde m_j^{\tilde\omega}\ge c_3'N$.

\emph{Step~1: reduction to an averaging fixed point.}
Fix $b\in\R^{2N}$ and let $v=\mathcal G_\phi b$, that is, $v=(v^{\mathrm{ex}},v^{\mathrm{im}})\in e^\perp$ solves $\bar H_{\phi,n}v=P_eb$.
The exporter and importer blocks read
\[
  v^{\mathrm{ex}}_i=\frac{(P_eb)_i}{m_i^{\tilde\omega}}
  -\sum_{l\neq i}\frac{\tilde\omega_{il}}{m_i^{\tilde\omega}}\,v^{\mathrm{im}}_l,
  \qquad
  v^{\mathrm{im}}_j=\frac{(P_eb)_j}{\tilde m_j^{\tilde\omega}}
  -\sum_{i\neq j}\frac{\tilde\omega_{ij}}{\tilde m_j^{\tilde\omega}}\,v^{\mathrm{ex}}_i .
\]
Substituting the first into the second gives $v^{\mathrm{im}}=h+Pv^{\mathrm{im}}$ with $P=P^{\mathrm{im}}P^{\mathrm{ex}}$, where $P^{\mathrm{im}}_{ji}=\tilde\omega_{ij}/\tilde m_j^{\tilde\omega}$ and $P^{\mathrm{ex}}_{il}=\tilde\omega_{il}/m_i^{\tilde\omega}$ are row stochastic, and $\|h\|_\infty\le4\|b\|_\infty/(c_3'N)$, using only the lower bounds on the masses and $\|P_eb\|_\infty\le2\|b\|_\infty$.

\emph{Step~2: Doeblin minorization on the exporters with $\xi_i\le K$.}
Choose a fixed $K\ge1$ large enough that $c_me^{K}\ge T_0$, where $m_i^{(\omega)}\ge c_mN$ w.p.a.\ one (Lemma~\ref{lem:margins}(iii)) and $T_0$ is the threshold of Lemma~\ref{lem:tailmass}; then $\xi_i>K$ forces $R_i=e^{\xi_i}m_i^{(\omega)}>T_\ast N$ with $T_\ast:=c_me^{K}\ge T_0$.
Let $I_0=\{i:\xi_i\le K\}$ and let $I_0^c=\{i\le N:\xi_i>K\}$ be its complement, with $\#$ denoting cardinality. 
Denote $N_0=\#I_0^c$. 
For $i\in I_0$ and every $l\neq i$, Lemma~\ref{lem:psi-facts}(i) with shift $k=\xi_i\in[-C_0,K]$ gives $\tilde\omega_{il}\ge c_\mu c_K\psi(\omega_l)$ and $m_i^{\tilde\omega}\le C_\mu C_K\sum_{l'}\psi(\omega_{l'})$, hence
\[
  P^{\mathrm{ex}}_{il}\;\ge\;\tau_1\,\sigma_l
  \qquad(l\neq i),
  \qquad
  \sigma_l:=\frac{\psi(\omega_l)}{\sum_{l'}\psi(\omega_{l'})},
  \qquad
  \tau_1=\frac{c_\mu c_K}{C_\mu C_K},
\]
with the minorizing probability vector $\sigma$ common to all $i\in I_0$; the diagonal entry $P^{\mathrm{ex}}_{ii}$ does not exist, which is accounted for below.
For the column weights, Lemma~\ref{lem:psi-facts}(i) gives $\tilde\omega_{ij}\ge c_\mu c_{K}\psi(\omega_j)$ for $i\in I_0$, while Lemma~\ref{lem:psi-facts}(ii), whose shift range covers $\xi_i\ge-C_0$, gives $\tilde\omega_{ij}\le C_\mu C_\psi e^{\xi_i}\psi(\omega_j)$ for all $i$; therefore, uniformly in $j$,
\[
  \sum_{i\in I_0}P^{\mathrm{im}}_{ji}
  =\frac{\sum_{i\in I_0,\,i\neq j}\tilde\omega_{ij}}{\sum_{i\neq j}\tilde\omega_{ij}}\;\ge\;1-\frac{C_\mu C_\psi\sum_{i\notin I_0}e^{\xi_i}}{c_\mu c_K\,(N-N_0-1)} .
\]
By the margin identity $e^{\xi_i}=R_i/m_i^{(\omega)}\le CR_i/N$ and Lemma~\ref{lem:tailmass} with $q=1$, $\sum_{i:\xi_i>K}e^{\xi_i}\le(C/N)\sum_iR_i\1\{R_i>T_\ast N\}=O_p(N^{2-\alpha})=o_p(N)$ for $\alpha>1$, and $N_0 \le e^{-K}\sum_{i\notin I_0}e^{\xi_i}=o_p(N)$.
Hence w.p.a.\ one $\sum_{i\in I_0}P^{\mathrm{im}}_{ji}\ge\tfrac12$ for all $j$.
The composition $P=P^{\mathrm{im}}P^{\mathrm{ex}}$ runs over $i\notin\{j,l\}$, because the diagonal entries of the two factors do not exist; the index $i=l$, at which the minorization is unavailable, must therefore be removed:
\[
  P_{jl}\;\ge\;\sum_{i\in I_0\setminus\{l\}}
  P^{\mathrm{im}}_{ji}P^{\mathrm{ex}}_{il}
  \;\ge\;\tau_1\sigma_l
  \Bigl(\sum_{i\in I_0}P^{\mathrm{im}}_{ji}-P^{\mathrm{im}}_{jl}\Bigr).
\]
The omitted transition probability is uniformly negligible: for all $j,l$ with $l\neq j$,
\[
  P^{\mathrm{im}}_{jl}
  \;=\;\frac{\tilde\omega_{lj}}{\tilde m_j^{\tilde\omega}}
  \;\le\;\frac{C_\mu C_\psi\,e^{\xi_l}\,\psi(\omega_j)}
              {c_3''\,N\,\psi(\omega_j)}
  \;\le\;\frac{C\max_ie^{\xi_i}}{N}
  \;=\;O_p\bigl(N^{2/\alpha-2}\bigr)\;=\;o_p(1),
\]
using Lemma~\ref{lem:psi-facts}(ii) for the numerator, the lower bound over $I_0$, $\tilde m_j^{\tilde\omega}\ge\sum_{i\in I_0,i\neq j}\tilde\omega_{ij}\ge c_\mu c_K(N-N_0-1)\psi(\omega_j)\ge c_3''N\psi(\omega_j)$ w.p.a.\ one for the denominator, and $\max_ie^{\xi_i}=O_p(N^{2/\alpha-1})$ (Lemma~\ref{lem:margins}(iii)).
For $l=j$ no transition needs to be removed, and directly $P_{jj}\ge\tau_1\sigma_j\sum_{i\in I_0}P^{\mathrm{im}}_{ji}\ge\tfrac12\tau_1\sigma_j$.
Hence w.p.a.\ one $P_{jl}\ge\tau_1\sigma_l(\tfrac12-o_p(1))\ge\tfrac14\tau_1\sigma_l$ for all $j,l$, and the Dobrushin coefficient satisfies $\delta(P)\le1-\tau_0$ with $\tau_0=\tau_1/4$.

\emph{Step~3: oscillation and level.}
From $v^{\mathrm{im}}=h+Pv^{\mathrm{im}}$ and $\osc(Pv^{\mathrm{im}})\le(1-\tau_0)\osc(v^{\mathrm{im}})$, $\osc(v^{\mathrm{im}})\le\osc(h)/\tau_0\le8\|b\|_\infty/(c_3'\tau_0 N)$.
The exporter equation exhibits $v^{\mathrm{ex}}_i$ as $(P_eb)_i/m_i^{\tilde\omega}$ minus a weighted average of $v^{\mathrm{im}}$, so with $c_0:=\min_jv^{\mathrm{im}}_j$, $\max_i|v^{\mathrm{ex}}_i+c_0|\le\osc(v^{\mathrm{im}})+2\|b\|_\infty/(c_3'N)$.
The constraint $\langle v,e\rangle=\sum_iv^{\mathrm{ex}}_i-\sum_jv^{\mathrm{im}}_j=0$ then forces $|c_0|\le C\|b\|_\infty/N$, whence $\|v\|_\infty\le C_J\|b\|_\infty/N$.
Since $b\in\R^{2N}$ is arbitrary, $\|\mathcal G_\phi\|_{\infty\to\infty}\le C_J/N$.

\emph{Step~4: duality and interpolation.}
$\mathcal G_\phi$ is symmetric, so $\|\mathcal G_\phi\|_{1\to1}=\|\mathcal G_\phi\|_{\infty\to\infty}\le C_J/N$, while $\|\mathcal G_\phi\|_{\mathrm{op}}\le c_4^{-1}/N$ by the conditioning of Lemma~\ref{lem:fe-rate}.
Riesz--Thorin interpolation, applied on the event where both endpoint bounds hold, yields $\|\mathcal G_\phi\|_{q\to q}\le C_J/N$ for every $q\in[1,2]$.
\end{proof}

\begin{lemma}[$\ell_q$ fixed-effect rates]\label{lem:fe-lp}
Under Assumption~\ref{ass:dgp}(a),(b), for every fixed
$q\in(1,\alpha)$, the orthogonal-representative errors satisfy
\[
  \textstyle\sum_i|\xi_i|^q+\sum_j|\omega_j|^q=O_p(N^{2-q}),
  \qquad
  \sum_i|u_i|^q+\sum_j|w_j|^q=O_p(N^{2-q}),
\]
and consequently
$\sum_i|u_i|+\sum_j|w_j|\le(2N)^{1-1/q}\bigl(\sum_i|u_i|^q
+\sum_j|w_j|^q\bigr)^{1/q}=O_p(N^{1/q})=o_p(N)$.
\end{lemma}
\begin{proof}
\emph{Additive errors.} $\nu_\perp=\mathcal G_\phi S_\phi^0$ with
$\langle S_\phi^0,e\rangle=0$. By
the von Bahr--Esseen inequality \citep{vonBahrEsseen1965},
$\E|S_i|^q\le2\sum_j\E|\mu_{ij}\eta_{ij}|^q\le CN$ for $q\in(1,\alpha)$,
so $\E\|S_\phi^0\|_q^q=\sum_i\E|S_i|^q+\sum_j\E|\tilde S_j|^q\le CN^2$, and
Lemma~\ref{lem:linf-hessian} gives
$\|\nu_\perp\|_q\le(C_J/N)\|S_\phi^0\|_q=O_p(N^{2/q-1})$, i.e.\
$\sum_i|\xi_i|^q+\sum_j|\omega_j|^q=O_p(N^{2-q})$.

\emph{Multiplicative errors.} Split at $\xi_i\le K'$ as in the last step of the proof of Lemma~\ref{lem:fe-rate}, on the event $\{|t_\nu|\le C_t(\varepsilon)\}$ chosen there; as there, $\varepsilon$ is arbitrary and the conclusions are $O_p$ statements. 
On $\{\xi_i\le K'\}$, $|u_i|\le C_1|\xi_i|$ and the additive bound applies. On $\{\xi_i>K'\}$, $1+u_i\le e^{C_t}CR_i/N$ and $\xi_i>K'$ forces $R_i>T_\ast N$, so by Lemma~\ref{lem:tailmass}, $\sum_{\xi_i>K'}(1+u_i)^q\le CN^{-q}\sum_iR_i^q\1\{R_i>T_\ast N\} =O_p(N^{2-\alpha})=o_p(N^{2-q})$ since $q<\alpha$. 
The importer bound is symmetric, and the $\ell_1$ bound follows from H\"older's inequality.
\end{proof}

\begin{corollary}[Uniformity over $\mathcal B$]\label{cor:uniform-beta}
Let $\nu_n(\beta)=\hat\phi(\beta)-\bar\phi(\beta)$, with components $\xi_i(\beta),\omega_j(\beta)$ and multiplicative errors $u_i(\beta),w_j(\beta)$ as in \eqref{eq:mult-fe}. 
Under Assumption~\ref{ass:dgp}(a),(b), the conclusions of Lemmas~\ref{lem:margins}(iii), \ref{lem:fe-rate}, \ref{lem:linf-hessian}, and~\ref{lem:fe-lp} hold with the pseudo-true
weights $\mu_\ell(\beta)$ in place of $\mu_\ell$ and $\nu_n(\beta)$ in place of $\nu_n$, with constants uniform over $\beta\in\mathcal B$; in particular
\[
  \sup_{\beta\in\mathcal B}
  \bigl(\|u(\beta)\|_2+\|w(\beta)\|_2\bigr)=O_p(N^{2/\alpha-1}),
  \qquad
  \sup_{\beta\in\mathcal B}
  \bigl(\|u(\beta)\|_q+\|w(\beta)\|_q\bigr)=O_p(N^{2/q-1})
\]
for every fixed $q\in(1,\alpha)$, and $\sup_{\beta\in\mathcal B}\max_ie^{\xi_i(\beta)}\big/ \sum_{i'}e^{\xi_{i'}(\beta)}=o_p(1)$, with the symmetric importer statements.
\end{corollary}

\begin{proof}
Since $\bar\phi(\beta)$ solves $\sum_\ell(\mu_\ell^0-\mu_\ell(\beta,\bar\phi(\beta)))d_\ell=0$ (Assumption~\ref{ass:dgp}(b)), the population part of the fixed-effect score cancels,
\[
  S_{\phi,n}(\beta,\bar\phi(\beta))
  =\sum_\ell(y_\ell-\mu_\ell^0)d_\ell
  +\sum_\ell\bigl(\mu_\ell^0-\mu_\ell(\beta,\bar\phi(\beta))\bigr)d_\ell
  =\sum_\ell\mu_\ell^0\eta_\ell d_\ell=S_\phi^0,
\]
the same $\beta$-free heavy-tailed vector for every $\beta$. 
The proofs of the cited lemmas use only this driving vector, the bounds $\mu_\ell(\beta)\in[c,C]$ and the Gram conditioning \eqref{eq:dense-design} at $\mu_\ell(\beta)$, both uniform over
$\beta\in\mathcal B$ by Assumption~\ref{ass:dgp}(b), the one-sided support $\eta\ge-1$, and the margin identities at the weights $\mu_\ell(\beta)$; every estimate therefore holds with constants uniform in $\beta$. The final ratio bound follows from $\max_ie^{\xi_i(\beta)}=O_p(N^{2/\alpha-1})$ and $\sum_ie^{\xi_i(\beta)}\ge Ne^{-C_0}$, as in Lemma~\ref{lem:margins}(iii).
\end{proof}

The next lemma quantifies the difference between the two residualizations that appear in the sandwich objects: the pseudo-true residual $\tilde x_\ell(\beta)$, built from the weights
$\mu_\ell(\beta)$, and the fitted-weight residual $\hat{\tilde x}_\ell(\beta)$ of Appendix~\ref{app:profiling}, built from $\hat\mu_\ell(\beta)=\mu_\ell(\beta)(1+u_i(\beta))(1+w_j(\beta))$. 
The feasible concentrated Hessian and the estimated scores are functions of $\hat{\tilde x}_\ell$, whereas the population objects are functions of $\tilde x_\ell$; the lemma shows the two agree uniformly at the fixed-effect $\ell_1$ rate.

\begin{lemma}[Fitted-weight residualization]\label{lem:fitted-residual}
Under Assumption~\ref{ass:dgp}(a),(b), for any fixed $q\in(1,\alpha)$,
\[
  \sup_{\beta\in\mathcal B}\max_\ell
  \bigl\|\hat{\tilde x}_\ell(\beta)-\tilde x_\ell(\beta)\bigr\|
  \;\le\;\frac{C}{N}\,
  \sup_{\beta\in\mathcal B}\bigl(\|u(\beta)\|_1+\|w(\beta)\|_1\bigr)
  \;=\;O_p\bigl(N^{1/q-1}\bigr)\;=\;o_p(1)
\]
w.p.a.\ one; consequently
$\sup_{\beta\in\mathcal B}\max_\ell\|\hat{\tilde x}_\ell(\beta)\|\le C'$
w.p.a.\ one, and
\[
  \sup_{\beta\in\mathcal B}
  \Bigl\|\frac1n\sum_\ell\hat\mu_\ell(\beta)\bigl[
  \hat{\tilde x}_\ell(\beta)\hat{\tilde x}_\ell(\beta)'
  -\tilde x_\ell(\beta)\tilde x_\ell(\beta)'\bigr]\Bigr\|
  \;=\;O_p\bigl(N^{1/q-1}\bigr).
\]
\end{lemma}

\begin{proof}
Fix $\beta$ and a coordinate $k$, and drop $\beta$ from the notation.
Both residuals are two-way fit residuals of the bounded dyadic variable $x^{(k)}$: $\tilde x^{(k)}$ at the weights $\mu_\ell(\beta)$ and $\hat{\tilde x}^{(k)}$ at the fitted weights $\hat\mu_{ij}=\mu_{ij}(1+u_i)(1+w_j)$.
Since a two-way fit reproduces additive components exactly, the difference of the two residuals is the $\hat\mu$-weighted two-way fit of the data $\tilde x^{(k)}$: $\hat{\tilde x}^{(k)}_{ij}-\tilde x^{(k)}_{ij}=-(\Delta a_i+\Delta b_j)$, where $(\Delta a,\Delta b)$ solves
\[
  \Delta a_i=F_i-\sum_{j\neq i}\frac{\hat\mu_{ij}}{\hat m_i}\,\Delta b_j,
  \qquad
  \Delta b_j=\tilde F_j
  -\sum_{i\neq j}\frac{\hat\mu_{ij}}{\hat{\tilde m}_j}\,\Delta a_i,
\]
with $\hat m_i=\sum_{j\neq i}\hat\mu_{ij}$, $\hat{\tilde m}_j=\sum_{i\neq j}\hat\mu_{ij}$, $F_i=\hat m_i^{-1}\sum_{j\neq i}\hat\mu_{ij}\tilde x^{(k)}_{ij}$, and its importer analogue $\tilde F_j$.
The forcing terms are centered by the $\mu(\beta)$-orthogonality $\sum_{j\neq i}\mu_{ij}\tilde x^{(k)}_{ij}=0$ (the row form of \eqref{eq:orth} at $\beta$):
\[
  \sum_{j\neq i}\hat\mu_{ij}\tilde x^{(k)}_{ij}
  =(1+u_i)\sum_{j\neq i}\mu_{ij}(1+w_j)\tilde x^{(k)}_{ij}
  =(1+u_i)\sum_{j\neq i}\mu_{ij}\,w_j\,\tilde x^{(k)}_{ij},
\]
while $\hat m_i=(1+u_i)\sum_{j\neq i}\mu_{ij}(1+w_j)=(1+u_i)\,m_i^{(\omega)}(\beta)$ with $m_i^{(\omega)}(\beta)\ge c_mN$ w.p.a.\ one uniformly in $i$ and $\beta$, by Lemma~\ref{lem:margins}(iii) and Corollary~\ref{cor:uniform-beta}.
The level factors $(1+u_i)$ cancel, so uniformly in $i$ and $\beta$,
\[
  |F_i|\le\frac{C}{N}\sum_j|w_j|=\frac{C}{N}\|w\|_1,
  \qquad
  |\tilde F_j|\le\frac{C}{N}\|u\|_1 .
\]
Eliminating $\Delta a$ as in Lemma~\ref{lem:linf-fit}, instance~(ii) (weights $\hat\mu$ with factors $p_i=1+u_i$, $q_j=1+w_j$, kernel $\mu(\beta)\in[c,C]$, and $\epsilon_N=o_p(1)$ uniformly in $\beta$ by Corollary~\ref{cor:uniform-beta}), gives $\Delta b=P\,\Delta b+h$ with $h$ built from the forcing terms, $\|h\|_\infty\le2\max\bigl(\max_i|F_i|,\max_j|\tilde F_j|\bigr)$, and Dobrushin coefficient $\delta(P)\le1-\tau_0$, so $\osc(\Delta b)\le2\|h\|_\infty/\tau_0$.
The $\Delta a$-equation exhibits $\Delta a_i+\Delta b_j$ as $F_i$ plus $\Delta b_j$ minus a weighted average of $\Delta b$, in which the level of $\Delta b$ cancels, whence $\max_{ij}|\Delta a_i+\Delta b_j|\le C(\|u\|_1+\|w\|_1)/N$.
The rate follows from $\sup_{\beta\in\mathcal B}(\|u(\beta)\|_1+\|w(\beta)\|_1)=O_p(N^{1/q})$ (Lemma~\ref{lem:fe-lp} and Corollary~\ref{cor:uniform-beta}), and the boundedness of $\hat{\tilde x}$ follows from $\max_\ell\|\tilde x_\ell(\beta)\|\le C$.
For the Gram comparison, $\|\hat{\tilde x}\hat{\tilde x}'-\tilde x\tilde x'\|\le(\|\hat{\tilde x}\|+\|\tilde x\|)\|\hat{\tilde x}-\tilde x\|\le C\max_\ell\|\hat{\tilde x}_\ell-\tilde x_\ell\|$, while $n^{-1}\sum_\ell\hat\mu_\ell(\beta)\le C+CN^{-1}(\|u\|_1+\|w\|_1)+Cn^{-1}\|u\|_1\|w\|_1=O_p(1)$ uniformly in $\beta$.
\end{proof}

\begin{lemma}[Smoothness of the profiling map]\label{lem:smooth-profile}
Under Assumption~\ref{ass:dgp}(b), the population profiling map $\beta\mapsto\bar\phi(\beta)$ solving $\sum_\ell\{\mu_\ell^0-\mu_\ell(\beta,\bar\phi(\beta))\}d_\ell=0$ is continuously differentiable on the compact set $\mathcal B$, with
\[
  \frac{\partial\bar\phi}{\partial\beta_k}(\beta)=-\,H_{\phi\phi,n}(\beta)^{-1}\sum_\ell\mu_\ell(\beta)\,d_\ell\,x^{(k)}_\ell,
  \qquad
  \sup_{\beta\in\mathcal B}\max_{\ell,k} \bigl|d_\ell'\,\partial\bar\phi(\beta)/\partial\beta_k\bigr|\le C_\infty,
\]
for a constant $C_\infty=C_\infty(c,C)$. 
Consequently $\beta\mapsto\mu_\ell(\beta)$ and $\beta\mapsto\tilde x_\ell(\beta)$ are continuously differentiable with derivatives bounded uniformly over $\ell$ and $\beta\in\mathcal B$, hence Lipschitz uniformly over $\ell$.
\end{lemma}

\begin{proof}
The pseudo-true effects solve the population score equation: with $\bar S_{\phi,n}(\beta,\phi):=\sum_\ell\{\mu_\ell^0-\mu_\ell(\beta,\phi)\}d_\ell$, $\bar\phi(\beta)$ solves $\bar S_{\phi,n}(\beta,\bar\phi(\beta))=0$; the map $(\beta,\phi)\mapsto\bar S_{\phi,n}(\beta,\phi)$ is $C^\infty$, and it is to be distinguished from the sample score $S_{\phi,n}$, which defines $\hat\phi(\beta)$.
The Jacobian of $\bar S_{\phi,n}$ in $\phi$ is $-H_{\phi\phi,n}(\beta,\phi)$, the additive term $\sum_\ell\mu^0_\ell d_\ell$ not depending on $\phi$, and it satisfies $\lambda_{\min}\bigl(N^{-1}H_{\phi\phi,n}(\beta,\phi)\restriction_{e^\perp}\bigr)\ge c>0$ uniformly over $\beta\in\mathcal B$ by \eqref{eq:pseudo-bdd} and \eqref{eq:dense-design}, hence is invertible on $e^\perp$.
The implicit function theorem on the quotient then gives a $C^1$ solution $\bar\phi(\beta)$, in the representative fixed by Assumption~\ref{ass:dgp}(b), with the stated derivative.

For the entrywise bound, note that $c^{(k)}(\beta):=-\partial\bar\phi(\beta)/\partial\beta_k$ solves $H_{\phi\phi,n}(\beta)\,c^{(k)}=\sum_\ell\mu_\ell(\beta)d_\ell x^{(k)}_\ell$, the normal-equation system \eqref{eq:xtilde-foc} of the $\mu(\beta)$-weighted two-way fit of the bounded dyadic variable $x^{(k)}$; equivalently $d_\ell'c^{(k)}=a^{(k)}_i(\beta)+b^{(k)}_j(\beta)$.
Lemma~\ref{lem:linf-fit}, instance~(i), gives $\max_i|a^{(k)}_i(\beta)|+\max_j|b^{(k)}_j(\beta)|\le K_\ast C$ uniformly over $\beta\in\mathcal B$, which is the entrywise bound with $C_\infty=2K_\ast C$.

For $\tilde x_\ell(\beta)$, differentiating the normal equations of the fit $(a^{(k)}(\beta),b^{(k)}(\beta))$ in $\beta_m$ shows that the derivative pair $(\dot a,\dot b)$ solves the same $\mu(\beta)$-weighted two-way normal equations with data $\bigl(\partial\log\mu_{ij}(\beta)/\partial\beta_m\bigr)\tilde x^{(k)}_{ij}(\beta)$.
Since $|\partial\log\mu_{ij}(\beta)/\partial\beta_m|=|x^{(m)}_{ij}+d_{ij}'\partial\bar\phi/\partial\beta_m|\le C+C_\infty$ by the first step and $|\tilde x^{(k)}_{ij}(\beta)|\le C$, a second application of Lemma~\ref{lem:linf-fit}, now with $C_g=(C+C_\infty)C$, bounds $\|\dot a\|_\infty+\|\dot b\|_\infty$ by a constant, so $\partial\tilde x^{(k)}_\ell(\beta)/\partial\beta_m$ is bounded uniformly over $\ell$ and $\beta$.
Uniform Lipschitz continuity of $\mu_\ell(\beta)$ and $\tilde x_\ell(\beta)$ on the compact convex set $\mathcal B$ follows.
\end{proof}


\begin{lemma}[Exact remainder identity and rates] \label{lem:remainder}
Under Assumption~\ref{ass:dgp}(a),(b), with probability approaching one, parts~(a)--(c) hold for every $\alpha\in(1,2)$; part~(d) holds for every $\alpha\in(1,2)$ under the additional condition \eqref{eq:design-op} with $\gamma<1$.
\begin{enumerate}
  \item[\emph{(a)}] \emph{(Exact remainder identity)} The profiled score admits the exact decomposition
    \begin{equation}
      S_n^{\mathrm{profile}}(\beta_0)\;=\;\sum_\ell s_\ell\;-\;R_n, \qquad
      R_n\;=\;\sum_{i\neq j}\mu_{ij}\,u_i w_j\,\tilde x_{ij}, \label{eq:exact-Rn}
    \end{equation}
    with $u_i,w_j$ as in \eqref{eq:mult-fe}.
  \item[\emph{(b)}] \emph{(Fixed-effect rate)}
    $\sum_i\xi_i^2,\sum_j\omega_j^2,\sum_i u_i^2,\sum_j w_j^2=O_p\bigl(N^{4/\alpha-2}\bigr)$ and, for every $q\in(1,\alpha)$, $\sum_i|u_i|^q,\sum_j|w_j|^q=O_p(N^{2-q})$.
  \item[\emph{(c)}] \emph{(score covariance remainder)}
    \begin{equation}
      \frac{n}{a_n^2}\Bigl\|\,\hat\Omega_n^{\mathrm{prof}}(\beta_0)-n^{-1}\textstyle\sum_{i\neq j}s_{ij}s_{ij}'\,\Bigr\|\;\convp\;0,
      \label{eq:meat-remainder}
    \end{equation}
    where $\hat\Omega_n^{\mathrm{prof}}(\beta_0)$ is the score covariance at $(\beta_0,\hat\phi(\beta_0))$ and $s_{ij}=\eta_{ij}\kappa_{ij}$ are the truth-level scores.
  \item[\emph{(d)}] \emph{(Score remainder)} $\|R_n\|=o_p(a_n)$, equivalently
    \begin{equation}
      \frac{1}{a_n}\Bigl[\,S_n^{\mathrm{profile}}(\beta_0) -\textstyle\sum_{i\neq j}\eta_{ij}\,\kappa_{ij}\,\Bigr]\;\convp\;0.
      \label{eq:fe-remainder}
    \end{equation}
\end{enumerate}
\end{lemma}
\begin{proof}
\emph{Step~1: remainder \eqref{eq:exact-Rn}.}
Since $y_\ell=\mu_\ell(1+\eta_\ell)$ and $\mu_\ell(\beta_0,\hat\phi_0)=\mu_\ell e^{d_\ell'\nu_n}$,
\[
  y_\ell-\mu_\ell(\beta_0,\hat\phi_0)
  \;=\;\mu_\ell\eta_\ell-\mu_\ell\bigl(e^{d_\ell'\nu_n}-1\bigr).
\]
Write $e^{d_\ell'\nu_n}-1=d_\ell'\nu_\perp+\rho_\ell$ with $\rho_\ell=e^{d_\ell'\nu_n}-1-d_\ell'\nu_\perp$.
Substituting into $S_{\beta,n}(\beta_0,\hat\phi_0)=\sum_\ell\{y_\ell-\mu_\ell(\beta_0,\hat\phi_0)\}x_\ell$ gives $S_{\beta,n}(\beta_0,\hat\phi_0)=S_{\beta,n}(\beta_0,\phi_0)-H_{\beta\phi,n}\nu_\perp-\sum_\ell\mu_\ell\rho_\ell x_\ell$, while the fixed-effect first-order condition $S_{\phi,n}(\beta_0,\hat\phi_0)=0$ gives $H_{\phi\phi,n}\nu_\perp=S_{\phi,n}(\beta_0,\phi_0)-\sum_\ell\mu_\ell\rho_\ell d_\ell$ on the quotient by $e$.
Solving for $\nu_\perp$ and substituting,
\[
  S_{\beta,n}(\beta_0,\hat\phi_0)
  =\underbrace{S_{\beta,n}(\beta_0,\phi_0)
  -H_{\beta\phi,n}H_{\phi\phi,n}^{-1}S_{\phi,n}(\beta_0,\phi_0)}_{=\,\sum_\ell s_\ell}
  \;-\;\sum_\ell\mu_\ell\rho_\ell\tilde x_\ell .
\]
Using $e^{\xi_i+\omega_j}=(1+u_i)(1+w_j)$ with $(u,w)$ the orthogonal-representative errors,
\[
  e^{d_\ell'\nu_n}-1=u_i+w_j+u_iw_j,
  \qquad
  \rho_\ell=(u_i-\xi_i)+(w_j-\omega_j)+u_iw_j .
\]
The first two terms are additive in $(i,j)$, so by \eqref{eq:orth} they contribute zero to $\sum_\ell\mu_\ell\rho_\ell\tilde x_\ell$; only the cross term remains, giving \eqref{eq:exact-Rn}.

\emph{Step~2: fixed-effect rates (b).}
These are Lemmas~\ref{lem:fe-rate} and~\ref{lem:fe-lp}.
In particular, the exact integral identity $\bar H_{\phi,n}\nu_\perp=S_\phi^0$ of Lemma~\ref{lem:fe-rate} is solved using the conditioning $\lambda_{\min}(N^{-1}\bar H_{\phi,n}\restriction_{e^\perp})\ge c_4$, which in turn rests on the uniform lower bound $\xi_i,\omega_j\ge-C_0$ of Lemma~\ref{lem:margins}: the one-sided shock support $\eta_{ij}\ge-1$ keeps every fitted margin above $c'N$, so no exporter or importer weight degenerates and the averaged Hessian is well conditioned.

\emph{Step~3: score covariance remainder (c).}
Set $\check s_\ell:=\{y_\ell-\mu_\ell(\beta_0,\hat\phi_0)\}\,\tilde x_\ell$, the estimated score with the truth-weight residual, and $\check\Omega_n:=n^{-1}\sum_\ell\check s_\ell\check s_\ell'$.
From $\check s_\ell-s_\ell=-\mu_\ell(u_i+w_j+u_iw_j)\tilde x_\ell$ and Assumption~\ref{ass:dgp}(b), $\|\check s_\ell-s_\ell\|\le C\zeta_\ell$ w.p.a.\ one, where $\zeta_\ell:=|u_i|+|w_j|+|u_iw_j|$, so
\[
  \Bigl\|\check\Omega_n
  -n^{-1}\textstyle\sum_\ell s_\ell s_\ell'\Bigr\|
  \le\underbrace{\frac2n\sum_\ell\|s_\ell\|\,\|\check s_\ell-s_\ell\|}_{=:E_1}
  +\underbrace{\frac1n\sum_\ell\|\check s_\ell-s_\ell\|^2}_{=:E_2}.
\]
We bound $E_1$ and $E_2$ by splitting $\zeta_\ell$ into its additive part $|u_i|+|w_j|$, whose terms depend on a single index, and its cross part $|u_iw_j|$, which depends on both.
The two parts are bounded separately: an additive term factorizes over the dyadic sum, each value being replicated once per free index, whereas the cross term weighted by the shocks is a bilinear form in $B_n=(|\eta_{ij}|)$ and requires an operator-norm bound.
The centered array $|\eta_{ij}|-\E|\eta|$ is i.i.d.\ with the same tail index $\alpha$, so Lemma~\ref{lem:rowsum}(i) applied to it gives $\max_i|\sum_j(|\eta_{ij}|-\E|\eta|)|=O_p(N^{2/\alpha})$, and adding back the $O(N)$ mean yields
\begin{equation}
  \max_i\sum_j|\eta_{ij}|,\ \max_j\sum_i|\eta_{ij}|=O_p(N^{2/\alpha}),
  \|B_n\|_{\mathrm{op}}\le\Bigl(\max_i\textstyle\sum_j|\eta_{ij}|\cdot \max_j\sum_i|\eta_{ij}|\Bigr)^{1/2}=O_p(N^{2/\alpha}),
  \label{eq:Bn-op}
\end{equation}
where the operator-norm bound follows from the Schur test.

\emph{Additive contributions.}
By part~(b), the additive part of $E_2$ is $\tfrac{C}{n}\sum_{i\neq j}(u_i^2+w_j^2)=\tfrac{C}{n}(N-1)\bigl(\sum_iu_i^2+\sum_jw_j^2\bigr)=O_p(N^{4/\alpha-3})$, the factor $N-1$ being the replication count.
By Cauchy--Schwarz $\sum_i|u_i|\le N^{1/2}(\sum_iu_i^2)^{1/2}=O_p(N^{2/\alpha-1/2})$, and likewise for $w$, so the additive part of $E_1$ is $\tfrac{C}{n}(\sum_i|u_i|)\max_i\sum_j|\eta_{ij}|+(\mathrm{sym.})=O_p(N^{-2}\cdot N^{2/\alpha-1/2}\cdot N^{2/\alpha})=O_p(N^{4/\alpha-5/2})$.

\emph{Product contributions.}
Writing $|u|=(|u_i|)_i$ and $|w|=(|w_j|)_j$, the product part of $E_1$ does not factorize and is bounded through \eqref{eq:Bn-op}: $\tfrac{C}{n}\sum_{i\neq j}|\eta_{ij}||u_iw_j|=\tfrac{C}{n}|u|'B_n|w|\le\tfrac{C}{n}\|u\|_2\|B_n\|_{\mathrm{op}}\|w\|_2=O_p(N^{-2}\cdot N^{2/\alpha-1}\cdot N^{2/\alpha}\cdot N^{2/\alpha-1})=O_p(N^{6/\alpha-4})$ by part~(b).
The product part of $E_2$ carries no shock weight and therefore does factorize: $\tfrac{C}{n}\sum_{i\neq j}u_i^2w_j^2=\tfrac{C}{n}\|u\|_2^2\|w\|_2^2=O_p(N^{8/\alpha-6})$.

Collecting the four terms and multiplying by $n/a_n^2\asymp N^{2-4/\alpha}$,
\[
  \frac{n}{a_n^2}\Bigl\|\check\Omega_n-n^{-1}\sum_\ell s_\ell s_\ell'\Bigr\|
  =O_p\bigl(N^{-1/2}+N^{-1}+N^{2/\alpha-2}+N^{4/\alpha-4}\bigr)\convp0
\]
for every $\alpha\in(1,2)$.

It remains to pass from the truth-residual score covariance $\check\Omega_n$ to $\hat\Omega_n^{\mathrm{prof}}(\beta_0)=n^{-1}\sum_\ell\hat s_\ell(\beta_0)\hat s_\ell(\beta_0)'$, whose scores carry the fitted-weight residual $\hat{\tilde x}_\ell(\beta_0)$.
By Lemma~\ref{lem:fitted-residual} at $\beta_0$, $\varrho_N:=\max_\ell\|\hat{\tilde x}_\ell(\beta_0)-\tilde x_\ell\|=O_p(N^{1/q-1})=o_p(1)$ for any fixed $q\in(1,\alpha)$, and
\[
  \|\hat s_\ell(\beta_0)-\check s_\ell\|
  =\bigl|y_\ell-\mu_\ell(\beta_0,\hat\phi_0)\bigr|\,
  \bigl\|\hat{\tilde x}_\ell(\beta_0)-\tilde x_\ell\bigr\|
  \le C(|\eta_\ell|+\zeta_\ell)\,\varrho_N,
  \qquad
  \|\check s_\ell\|\le C(|\eta_\ell|+\zeta_\ell).
\]
Hence
\[
  \frac{n}{a_n^2}
  \bigl\|\hat\Omega_n^{\mathrm{prof}}(\beta_0)-\check\Omega_n\bigr\|
  \le\frac{C}{a_n^2}\bigl(\varrho_N+\varrho_N^2\bigr)
  \sum_\ell(|\eta_\ell|+\zeta_\ell)^2
  =O_p(\varrho_N)\,O_p\Bigl(1+\frac{n}{a_n^2}\Bigr)=o_p(1),
\]
using $\sum_\ell(|\eta_\ell|+\zeta_\ell)^2=O_p(n+a_n^2)$, by Lemma~\ref{lem:rowsum}(iv) applied to $\{\eta_\ell^2\}$ and Lemma~\ref{lem:fe-rate} as in Step~5 of the proof of Theorem~\ref{thm:stable}, and $n/a_n^2\to0$ for $\alpha\in(1,2)$.
Combining the two comparisons gives \eqref{eq:meat-remainder}.
The passage to the score covariance $\hat\Omega_n$ at the PPML estimates $(\hat\beta,\hat\phi(\hat\beta))$ is carried out in Step~5 of the proof of Theorem~\ref{thm:stable}, using only consistency, the estimator rate, and Lemmas~\ref{lem:fe-rate} and~\ref{lem:fitted-residual}.

\emph{Step~4: score remainder (d).}
By the exact identity \eqref{eq:exact-Rn}, the $k$-th component of the remainder is the bilinear form
\[
  R_n^{(k)}=\sum_{i\neq j}\mu_{ij}\tilde x_{ij}^{(k)}u_iw_j
  =u'A_n^{(k)}w,
\]
with $A_n^{(k)}$ the residualized weighted design matrix of condition~\eqref{eq:design-op}, $(A_n^{(k)})_{ij}=\mu_{ij}\tilde x_{ij}^{(k)}$ ($i\neq j$, zero diagonal).
To bound $|R_n^{(k)}|$, note that $A_n^{(k)}$ maps $\ell_1\to\ell_\infty$ with norm $\max_{ij}|\mu_{ij}\tilde x^{(k)}_{ij}|\le C$ by Assumption~\ref{ass:dgp}(b) and $\ell_2\to\ell_2$ with norm $O(N^\gamma)$ by \eqref{eq:design-op}; by the Riesz--Thorin theorem, for every $\vartheta\in[0,1]$ and $1/q=(1+\vartheta)/2$, $A_n^{(k)}$ maps $\ell_q\to\ell_{q'}$ ($q'$ the conjugate exponent of $q$) with norm $O(N^{\gamma(1-\vartheta)})$.
H\"older's inequality and the $\ell_q$ rates $\|u\|_q,\|w\|_q=O_p(N^{2/q-1})=O_p(N^{\vartheta})$ of Lemma~\ref{lem:fe-lp}, admissible whenever $q=2/(1+\vartheta)<\alpha$, i.e.\ $\vartheta>2/\alpha-1$, give
\begin{equation}
  \bigl|R_n^{(k)}\bigr|\le\bigl\|A_n^{(k)}\bigr\|_{q\to q'}\|u\|_q\|w\|_q =O_p\bigl(N^{\gamma(1-\vartheta)+2\vartheta}\bigr).
  \label{eq:score-interp}
\end{equation}
Since $a_n\asymp N^{2/\alpha}$ under Assumption~\ref{ass:dgp}(a), the right-hand side of \eqref{eq:score-interp} is $o_p(a_n)$ provided $\gamma(1-\vartheta)+2\vartheta<2/\alpha$ for some admissible $\vartheta$.
At $\vartheta=2/\alpha-1+\epsilon$ the requirement reads
\[
  2\gamma\Bigl(1-\frac1\alpha\Bigr)+\frac4\alpha-2+\epsilon(2-\gamma)
  \;<\;\frac2\alpha
  \quad\Longleftrightarrow\quad
  \gamma\;<\;1-\frac{\epsilon(2-\gamma)}{2(1-1/\alpha)},
\]
which holds for $\epsilon$ small whenever $\gamma<1$.
As the dimension $p$ is fixed, $\|R_n\|=o_p(a_n)$, which is \eqref{eq:fe-remainder}.
This proves (a)--(d).
\end{proof}
\begin{lemma}[Profile-score expansion]\label{lem:general-beta}
With the pseudo-true weights $\mu_\ell(\beta)$ and residual $\tilde x_\ell(\beta)$ of Appendix~\ref{app:profiling}, fixed-effect error $\nu_n(\beta)=\hat\phi(\beta)-\bar\phi(\beta)$,
$d_\ell'\nu_n(\beta)=\xi_i(\beta)+\omega_j(\beta)$, $u_i(\beta)=e^{\xi_i(\beta)}-1$, and $w_j(\beta)=e^{\omega_j(\beta)}-1$, the profile-score remainder \eqref{eq:Rn-def} has the closed form
\begin{equation}
  R_n(\beta)=\sum_\ell\mu_\ell(\beta)\,u_i(\beta)w_j(\beta)\,\tilde x_\ell(\beta),
  \label{eq:Rn-general}
\end{equation}
and, under Assumption~\ref{ass:dgp}(a),(b) and condition~\eqref{eq:design-op} with $\gamma<1$,
\begin{equation}
  \sup_{\beta\in\mathcal B}\bigl\|R_n(\beta)\bigr\|=o_p(a_n)\qquad\text{for every }\alpha\in(1,2).
  \label{eq:Rn-general-bound}
\end{equation}
\end{lemma}

\begin{proof}
The reduction of $R_n(\beta)$ to the cross term, as in \eqref{eq:exact-Rn} is algebraic and uses only (a) the multiplicative identity $e^{d_\ell'\nu_n(\beta)}-1=u_i(\beta)+w_j(\beta)+u_i(\beta)w_j(\beta)$,
(b) the fixed-effect first-order condition $S_{\phi,n}(\beta,\hat\phi(\beta))=0$, and (c) the residual orthogonality $\sum_\ell\mu_\ell(\beta)\tilde x_\ell(\beta)d_\ell'=0$.
All three hold at every $\beta$: (a) is identical to the $\beta_0$ case, (b) defines $\hat\phi(\beta)$, and (c) holds by construction of $\tilde x_\ell(\beta)$ as the $\mu(\beta)$-weighted projection residual. The additive separable terms $u_i(\beta)+w_j(\beta)$ are therefore annihilated, leaving \eqref{eq:Rn-general}.

For \eqref{eq:Rn-general-bound}, by Corollary~\ref{cor:uniform-beta}, $\|u(\beta)\|_2,\|w(\beta)\|_2=O_p(N^{2/\alpha-1})$ and $\|u(\beta)\|_q,\|w(\beta)\|_q=O_p(N^{2/q-1})$ for $q\in(1,\alpha)$, uniformly in $\beta$. For the $k$-th component, \eqref{eq:Rn-general} is the bilinear form $R_n^{(k)}(\beta)=u(\beta)'A_n^{(k)}(\beta)w(\beta)$ with
$(A_n^{(k)}(\beta))_{ij}=\mu_{ij}(\beta)\tilde x_{ij}^{(k)}(\beta)$.
Condition~\eqref{eq:design-op} gives $\sup_{\beta\in\mathcal B}\max_k\|A_n^{(k)}(\beta)\|_{\mathrm{op}} =O(N^\gamma)$, while
$\max_{ij}|\mu_{ij}(\beta)\tilde x^{(k)}_{ij}(\beta)|\le C$ uniformly by Assumption~\ref{ass:dgp}(b). 
The interpolated bound \eqref{eq:score-interp} therefore applies uniformly in $\beta$ and gives $\sup_{\beta\in\mathcal B}\|R_n(\beta)\| =O_p(N^{\gamma(1-\vartheta)+2\vartheta})=o_p(a_n)$ for $\gamma<1$, which is \eqref{eq:Rn-general-bound}.
\end{proof}

The subnetwork analysis requires the analogue of Assumption~\ref{ass:dgp}(b) on the induced subnetwork away from the truth. 
The relevant population weights there are built from the subnetwork pseudo-true fixed effects, which solve the fixed-effect equations restricted to $\mathcal S$ and are not in
general the restriction of the full-sample $\bar\phi(\beta)$. 
The next lemma shows that they are nevertheless uniformly bounded, so the full-sample argument applies with $N$ replaced by $G$. 
The statement is deterministic: it holds for every draw.

\begin{lemma}[Subnetwork pseudo-true scaling]\label{lem:subnet-scaling}
Let $\mathcal C_G$ be any set of $G$ countries with induced complete subnetwork $\mathcal S$, and for $\beta\in\mathcal B$ let $\bar\phi_{\mathcal S}(\beta)$ solve the subnetwork population fixed-effect equations $\sum_{\ell\in\mathcal S}\bigl(\mu^0_\ell-\mu_\ell(\beta,\bar\phi_{\mathcal S}(\beta))\bigr)d_{\ell,\mathcal S}=0$; write $\mu_{ij,\mathcal S}(\beta)=\mu_{ij}(\beta,\bar\phi_{\mathcal S}(\beta))$ and $\tilde x_{ij,\mathcal S}(\beta)$ for the $\mu_{\mathcal S}(\beta)$-weighted two-way residual on $\mathcal S$, and set $e_{\mathcal S}:=(\mathbf 1_G',-\mathbf 1_G')'$.
Under Assumption~\ref{ass:dgp}(b) there are a $G_0<\infty$ and constants $0<c_{\mathcal S}\le C_{\mathcal S}<\infty$, $C'_{\mathcal S}<\infty$, depending only on the constants of Assumption~\ref{ass:dgp}(b) such that a solution exists, is unique on the quotient, and, in the representative lying in $e_{\mathcal S}^\perp$, satisfies
\[
  c_{\mathcal S}\le\mu_{ij,\mathcal S}(\beta)\le C_{\mathcal S},
  \qquad
  \max_{i\in\mathcal C_G}\bigl|\bar\phi^{\mathrm{ex}}_{i,\mathcal S}(\beta)\bigr|
  \;\vee\;\max_{j\in\mathcal C_G}\bigl|\bar\phi^{\mathrm{im}}_{j,\mathcal S}(\beta)\bigr|
  \;\le\;C_{\mathcal S},
  \qquad
  \max_{\ell\in\mathcal S}\bigl\|\tilde x_{\ell,\mathcal S}(\beta)\bigr\|\le C'_{\mathcal S},
\]
while the subnetwork fixed-effect Gram matrix $H_{\phi\phi,\mathcal S}(\beta)=\sum_{\ell\in\mathcal S}\mu_{\ell,\mathcal S}(\beta)d_{\ell,\mathcal S}d_{\ell,\mathcal S}'$ satisfies
\[
  c_{\mathcal S}
  \;\le\;\lambda_{\min}\bigl(G^{-1}H_{\phi\phi,\mathcal S}(\beta)\restriction_{e_{\mathcal S}^\perp}\bigr)
  \;\le\;\lambda_{\max}\bigl(G^{-1}H_{\phi\phi,\mathcal S}(\beta)\restriction_{e_{\mathcal S}^\perp}\bigr)
  \;\le\;C_{\mathcal S}.
\]
Consequently the induced subnetwork satisfies the boundedness and Gram conditions \eqref{eq:pseudo-bdd} and \eqref{eq:dense-design} with $N$ replaced by $G$, with constants uniform over draws, over $G\ge G_0$, and over $\beta\in\mathcal B$; the Hessian-limit condition \eqref{eq:Hn-limit} on the subnetwork is established separately in Lemma~\ref{lem:subnet-inherit}(ii).
\end{lemma}

\begin{proof}
Write $k_{ij}(\beta)=\exp(x_{ij}'\beta)$; by the boundedness of $x_{ij}$ and compactness of $\mathcal B$, $k_{ij}(\beta)\in[c_x,C_x]$ with $0<c_x\le C_x<\infty$ uniform over $\beta$.

We first show that $\bar\phi_{\mathcal S}(\beta)$ exists and is unique on the quotient.
Writing $\bar a_i=e^{\bar\phi^{\mathrm{ex}}_{i,\mathcal S}(\beta)}$ and $\bar b_j=e^{\bar\phi^{\mathrm{im}}_{j,\mathcal S}(\beta)}$ (barred, to distinguish these Sinkhorn factors from the fit coefficients of Lemma~\ref{lem:linf-fit}), the fixed-effect equations are the margin conditions
\[
  \bar a_i\sum_{j\in\mathcal C_G\setminus\{i\}}k_{ij}\bar b_j=R_i^0,
  \qquad
  \bar b_j\sum_{i\in\mathcal C_G\setminus\{j\}}k_{ij}\bar a_i=\tilde R_j^0,
\]
where $R_i^0=\sum_{j\in\mathcal C_G\setminus\{i\}}\mu^0_{ij}$ and $\tilde R_j^0=\sum_{i\in\mathcal C_G\setminus\{j\}}\mu^0_{ij}$, so $\bar\phi_{\mathcal S}(\beta)$ exists exactly when the kernel $k(\beta)$ admits a scaling to these margins.
The pattern $K_{G,G}\setminus\{(i,i)\}$ is fully indecomposable for $G\ge3$ \citep{Sinkhorn1967}.
The target margins are generated by the strictly positive restricted matrix $\{\mu^0_{ij}:i,j\in\mathcal C_G,\ i\neq j\}$ on that support, so they are strictly feasible: they satisfy the strict Hall conditions of the margin cone.
A finite scaling to these prescribed margins therefore exists \citep{RothblumSchneider1989}; see \citet{Idel2016} for a survey.
Strict concavity of the subnetwork population objective on the quotient gives uniqueness there.

For the bounds, fix temporarily the representative with $\sum_j\bar b_j=G-1$.
Since $\mu^0_{ij}\in[c,C]$, both $R_i^0$ and $\tilde R_j^0$ lie in $[c(G-1),C(G-1)]$ for every draw.
Then $\sum_{j\neq i}k_{ij}\bar b_j\le C_x(G-1)$, so $\bar a_i\ge R_i^0/\{C_x(G-1)\}\ge c/C_x=:a_{\min}>0$ for every $i$.
Then $\sum_{i\neq j}k_{ij}\bar a_i\ge c_xa_{\min}(G-1)$, so $\bar b_j\le\tilde R_j^0/\{c_xa_{\min}(G-1)\}\le CC_x/(c_xc)=:b_{\max}$.
In particular no importer weight dominates, and $\sum_{j\neq i}k_{ij}\bar b_j\ge c_x(G-1-b_{\max})\ge\tfrac12c_x(G-1)$ for $G\ge G_0:=2b_{\max}+3$, which also ensures $G\ge3$.
Hence $\bar a_i\le R_i^0/\{\tfrac12c_x(G-1)\}\le2C/c_x=:a_{\max}$ and $\bar b_j\ge\tilde R_j^0/\{C_xa_{\max}(G-1)\}\ge cc_x/(2CC_x)=:b_{\min}>0$.
Therefore $\mu_{ij,\mathcal S}(\beta)=k_{ij}(\beta)\bar a_i\bar b_j\in[c_xa_{\min}b_{\min},\,C_xa_{\max}b_{\max}]=:[c_{\mathcal S},C_{\mathcal S}]$, and the fixed effects themselves obey $|\bar\phi^{\mathrm{ex}}_{i,\mathcal S}(\beta)|\le|\log a_{\min}|\vee|\log a_{\max}|$ and $|\bar\phi^{\mathrm{im}}_{j,\mathcal S}(\beta)|\le|\log b_{\min}|\vee|\log b_{\max}|$, uniformly over draws, $G\ge G_0$, and $\beta\in\mathcal B$.
Passing to the representative in $e_{\mathcal S}^\perp$ subtracts $t_\ast e_{\mathcal S}$ with $t_\ast=\bigl(\sum_i\bar\phi^{\mathrm{ex}}_{i,\mathcal S}-\sum_j\bar\phi^{\mathrm{im}}_{j,\mathcal S}\bigr)/(2G)$, which is bounded by the same constants, so the fixed-effect bounds persist there after enlarging $C_{\mathcal S}$, and $\mu_{ij,\mathcal S}(\beta)$ is unchanged.

For the residual, $\tilde x^{(k)}_{\ell,\mathcal S}(\beta)$ is the residual of the bounded data $x^{(k)}$ from the $\mu_{\mathcal S}(\beta)$-weighted two-way fit on the complete $G$-network.
Lemma~\ref{lem:linf-fit} applies with kernel $r=\mu_{\mathcal S}(\beta)\in[c_{\mathcal S},C_{\mathcal S}]$ and $p_i\equiv q_j\equiv1$, for which the balance parameter is exactly $\epsilon_G=1/G$, so $\epsilon_G\le\epsilon_\ast$ after enlarging $G_0$ if necessary.
It bounds the fit coefficients by $K_\ast C$, so $\max_\ell\|\tilde x_{\ell,\mathcal S}(\beta)\|\le C+2K_\ast C=:C'_{\mathcal S}$.

Finally, for $v=(\xi,\omega)\in\R^{2G}$ we have $v'H_{\phi\phi,\mathcal S}(\beta)v=\sum_{i\neq j}\mu_{ij,\mathcal S}(\beta)(\xi_i+\omega_j)^2$.
Since $\sum_{i\neq j}(\xi_i+\omega_j)^2\le2(G-1)\|v\|^2$, the upper bound $\lambda_{\max}\bigl(G^{-1}H_{\phi\phi,\mathcal S}(\beta)\bigr)\le2C_{\mathcal S}$ is immediate.
For the lower bound, $v\in e_{\mathcal S}^\perp$ means $\sum_i\xi_i=\sum_j\omega_j$, so $\sum_{i,j}(\xi_i+\omega_j)^2=G\|v\|^2+2\bigl(\sum_i\xi_i\bigr)^2\ge G\|v\|^2$, while the excluded diagonal contributes at most $\sum_i(\xi_i+\omega_i)^2\le2\|v\|^2$.
Hence $\sum_{i\neq j}(\xi_i+\omega_j)^2\ge(G-2)\|v\|^2$ and $\lambda_{\min}\bigl(G^{-1}H_{\phi\phi,\mathcal S}(\beta)\restriction_{e_{\mathcal S}^\perp}\bigr)\ge c_{\mathcal S}(G-2)/G\ge\tfrac13c_{\mathcal S}$ for $G\ge3$.
Both constants depend only on $(c_{\mathcal S},C_{\mathcal S})$, and are absorbed by decreasing $c_{\mathcal S}$ and enlarging $C_{\mathcal S}$.
\end{proof}

\begin{lemma}[Subnetwork residual comparison]\label{lem:subnet-residual}
Under Assumption~\ref{ass:dgp}(b), let $\mathcal C_G$ be a uniformly drawn set of $G$ countries, $G\to\infty$, and let $\tilde x_{\ell,\mathcal S}$ be the $\mu$-weighted two-way residual on the induced complete subnetwork $\mathcal S$. Then under $G\to\infty$,
\[
  \max_{\ell\in\mathcal S}\bigl\|\tilde x_{\ell,\mathcal S}-\tilde x_\ell\bigr\|\;=\;O_p\Bigl(\sqrt{\tfrac{\log G}{G}}\Bigr)\;=\;o_p(1).
\]
\end{lemma}

\begin{proof}
Fix a coordinate $k$ and write $x^{(k)}_{ij}=a^{(k)}_i+b^{(k)}_j+\tilde x^{(k)}_{ij}$ for the full-sample fit \eqref{eq:xtilde-additive}.
A two-way fit reproduces additive components exactly, so the subnetwork fit of $x^{(k)}$ equals $(a^{(k)},b^{(k)})$ restricted to $\mathcal C_G$ plus the subnetwork fit $(\Delta a,\Delta b)$ of the data $\tilde x^{(k)}$, and $\tilde x^{(k)}_{\ell,\mathcal S}-\tilde x^{(k)}_\ell=-(\Delta a_i+\Delta b_j)$ for $\ell=(i,j)\in\mathcal S$.
The subnetwork normal equations for $(\Delta a,\Delta b)$ read
\[
  \Delta a_i =F_i-\sum_{j\in\mathcal C_G\setminus\{i\}}\frac{\mu_{ij}}{m^{\mathcal S}_i}\,\Delta b_j, \qquad
  \Delta b_j =\tilde F_j-\sum_{i\in\mathcal C_G\setminus\{j\}}\frac{\mu_{ij}}{\tilde m^{\mathcal S}_j}\,\Delta a_i,
\]
with $m^{\mathcal S}_i=\sum_{j\in\mathcal C_G\setminus\{i\}}\mu_{ij}$, $\tilde m^{\mathcal S}_j=\sum_{i\in\mathcal C_G\setminus\{j\}}\mu_{ij}$, and forcing terms $F_i=(m^{\mathcal S}_i)^{-1}\sum_{j\in\mathcal C_G\setminus\{i\}}\mu_{ij}\tilde x^{(k)}_{ij}$ and its importer analogue $\tilde F_j$.
By the full-sample orthogonality \eqref{eq:orth} applied to the exporter-$i$ and importer-$j$ indicators, $\sum_{j\neq i}\mu_{ij}\tilde x^{(k)}_{ij}=0$ for every $i$ and $\sum_{i\neq j}\mu_{ij}\tilde x^{(k)}_{ij}=0$ for every $j$, so each $F_i$ is a normalized sum, over a simple random sample drawn without replacement, of bounded population values with zero total.
Hoeffding's inequality for sampling without replacement \citep[Section~6]{Hoeffding1963} gives $\Pr(|F_i|>\varepsilon\mid i\in\mathcal C_G)\le2\exp(-c_H\varepsilon^2G)$ for a constant $c_H=c_H(c_\mu,C_\mu,C)>0$.
Only the drawn margins need to be controlled, and there are $2G$ of them: by exchangeability of the draw,
\begin{align*}
  \Pr\Bigl(\max_{i\in\mathcal C_G}|F_i|>\varepsilon\Bigr)
  &\;\le\;\E\Bigl[\sum_{i\in\mathcal C_G}\1\{|F_i|>\varepsilon\}\Bigr]\\
  &\;=\;\sum_{i=1}^N\Pr(i\in\mathcal C_G)\,
       \Pr\bigl(|F_i|>\varepsilon\mid i\in\mathcal C_G\bigr)\\
  &\;\le\;2G\exp(-c_H\varepsilon^2G),
\end{align*}
using $\Pr(i\in\mathcal C_G)=G/N$, and likewise for the importer margins.
Over the $p$ coordinates this yields $\max_{i\in\mathcal C_G}|F_i|+\max_{j\in\mathcal C_G}|\tilde F_j|=O_p(\sqrt{\log G/G})$, which vanishes under $G\to\infty$.

Eliminating $\Delta a$ as in Lemma~\ref{lem:linf-fit} (instance~(iii): weights $\mu\in[c,C]$ on the complete $G$-country subnetwork, $p_i\equiv q_j\equiv1$), $\Delta b=P_{\mathcal S}\Delta b+h_{\mathcal S}$ with $\|h_{\mathcal S}\|_\infty = \,O_p(\sqrt{\log G/G})$ and Dobrushin coefficient $\delta(P_{\mathcal S})\le1-\tau_0$, $\tau_0=\tau_0(c/C)>0$, so $\osc(\Delta b) = \,O_p(\sqrt{\log G/G})$.

Finally, write $\langle\Delta b\rangle_i=\sum_{j\in\mathcal C_G\setminus\{i\}}(\mu_{ij}/m^{\mathcal S}_i)\,\Delta b_j$ for the average appearing in the $\Delta a$-equation, whose weights are nonnegative and sum to one.
That equation reads $\Delta a_i=F_i-\langle\Delta b\rangle_i$, so $\Delta a_i+\Delta b_j=F_i+\bigl(\Delta b_j-\langle\Delta b\rangle_i\bigr)$, in which any common shift of $\Delta b$ cancels.
Since $\langle\Delta b\rangle_i$ is a convex combination of the $\Delta b_j$, it and $\Delta b_j$ both lie between $\min_j\Delta b_j$ and $\max_j\Delta b_j$, so $|\Delta b_j-\langle\Delta b\rangle_i|\le\osc(\Delta b)$.
Hence $\max_{\ell\in\mathcal S}|\Delta a_i+\Delta b_j|\le\max_i|F_i|+\osc(\Delta b)=O_p\bigl(\sqrt{\log G/G}\bigr)$.
\end{proof}

\begin{lemma}[Subnetwork inheritance]\label{lem:subnet-inherit}
Let $\mathcal C_G$ be a uniformly drawn set of $G$ countries with induced complete bilateral subnetwork $\mathcal S$ on $n_G=G(G-1)$ ordered pairs, and let $a_{n_G}\asymp G^{2/\alpha}$ solve $n_G\Pr(\eta>a_{n_G})\to1$. Suppose Assumption~\ref{ass:dgp} holds together with the subnetwork analogue of
\eqref{eq:design-op},
\begin{equation}
  \sup_{\beta\in\mathcal B}\max_{1\le k\le p}\bigl\|A^{k}_{(\mathcal S)}(\beta)\bigr\|_{\mathrm{op}}=O_p(G^{\gamma}),\qquad \gamma<1 ,
  \label{eq:design-op-sub}
\end{equation}
where $A^{k}_{(\mathcal S)}(\beta)$ is the $G\times G$ matrix with off-diagonal entries $\mu_{ij,\mathcal S}(\beta)\,\tilde x^{(k)}_{ij,\mathcal S}(\beta)$, built from the subnetwork pseudo-true weights and residuals of Lemma~\ref{lem:subnet-scaling}.
Then, as $G\to\infty$:
\begin{enumerate}
  \item[\emph{(i)}] the within-subnetwork tail spectral measure converges in
    probability over the draw, with
    $\kappa_{\ell,\mathcal S}=\mu_\ell\tilde x_{\ell,\mathcal S}$,
    \[
      \Bigl(\textstyle\sum_{\ell\in\mathcal S}\|\kappa_{\ell,\mathcal S}\|^\alpha\Bigr)^{-1}\sum_{\ell\in\mathcal S}\|\kappa_{\ell,\mathcal S}\|^\alpha\,
      \delta_{\kappa_{\ell,\mathcal S}/\|\kappa_{\ell,\mathcal S}\|}\;\Rightarrow\;\Gamma;
    \]
  \item[\emph{(ii)}] the within-subnetwork Hessian converges, $n_G^{-1}\sum_{\ell\in\mathcal S}\mu_\ell\tilde x_{\ell,\mathcal S}\tilde x_{\ell,\mathcal S}'\convp H$, where $\tilde x_{\ell,\mathcal S}$ partials out the $2G$ subnetwork fixed effects;
  \item[\emph{(iii)}] with $H^{\ast}_{\mathcal S}(\beta)$ and $\hat\Omega^{\mathrm{prof}}_{\mathcal S}(\beta_0)$ the subnetwork analogues of \eqref{eq:feasible-hessian} and the truth-profiled score covariance,
    \[
      \sup_{\beta\in\mathcal B}\|R_{\mathcal S}(\beta)\|=o_p(a_{n_G}), \qquad
      \frac{n_G}{a_{n_G}^{2}}\Bigl\|\hat\Omega^{\mathrm{prof}}_{\mathcal S}(\beta_0)-n_G^{-1}\textstyle\sum_{\ell\in\mathcal S}s_{\ell,\mathcal S} s_{\ell,\mathcal S}' \Bigr\|\convp0,
    \]
    and $\sup_{\|\beta-\beta_0\|\le\delta_G}\|H^{\ast}_{\mathcal S}(\beta)-H\|\convp0$ for every sequence $\delta_G\downarrow0$;
  \item[\emph{(iv)}] the subsample self-normalized statistic $T_{\mathcal S}$, defined as $\mathrm{SN}_n$ of \eqref{eq:SN-defn} computed on $\mathcal S$ with $(n,N)$ replaced by $(n_G,G)$, satisfies $T_{\mathcal S}\convd J^\star_\alpha$.
\end{enumerate}
\end{lemma}

\begin{proof}
The induced subnetwork is a complete bilateral network on $G$ countries, with per-country degree $G-1$ and $2G$ fixed effects.
It is therefore an instance of the full-sample design with $N$ replaced by $G$.
At the truth its weights are $\mu^0\in[c,C]$.
Away from the truth, the relevant population weights are the subnetwork pseudo-true weights $\mu_{ij,\mathcal S}(\beta)$ of Lemma~\ref{lem:subnet-scaling}, built from $\bar\phi_{\mathcal S}(\beta)$; this map is not in general the restriction of the full-sample $\bar\phi(\beta)$.
By Lemma~\ref{lem:subnet-scaling}, the fixed effects $\bar\phi_{\mathcal S}(\beta)$, the weights, the residuals $\tilde x_{\ell,\mathcal S}(\beta)$, and the subnetwork fixed-effect Gram matrix satisfy the analogues of \eqref{eq:pseudo-bdd} and \eqref{eq:dense-design}, with constants uniform over draws, over $G\ge G_0$, and over $\beta\in\mathcal B$.
The subnetwork therefore satisfies part~(a) and the boundedness and Gram parts of part~(b) of Assumption~\ref{ass:dgp} with $N$ replaced by $G$; the Hessian limit \eqref{eq:Hn-limit} on the subnetwork is part~(ii) below.
Once \eqref{eq:design-op-sub} is established, the arguments of Lemmas~\ref{lem:margins}--\ref{lem:general-beta} and the steps of the proof of Theorem~\ref{thm:stable} apply with $N$ replaced by $G$ and with \eqref{eq:design-op-sub} in the role of \eqref{eq:design-op}, with $\bar\phi_{\mathcal S},\mu_{\mathcal S},\tilde x_{\mathcal S}$ in the roles of $\bar\phi,\mu,\tilde x$.

\emph{(i).} Let $f$ be Lipschitz on $\mathbb S^{p-1}$ and set $\psi_f(\kappa)=\|\kappa\|^\alpha f(\kappa/\|\kappa\|)$, $\psi_f(0)=0$; for $\alpha\ge1$, $\psi_f$ is Lipschitz on $\{\|\kappa\|\le C\}$.
Since $\kappa_{\ell,\mathcal S}$ and $\kappa_\ell$ both lie in that set, Lemma~\ref{lem:subnet-residual} and this Lipschitz property give
\[
  X(f):=n_G^{-1}\sum_{\ell\in\mathcal S}\psi_f(\kappa_{\ell,\mathcal S})
  =n_G^{-1}\sum_{\ell\in\mathcal S}\psi_f(\kappa_\ell)
  +O_p\bigl(\sqrt{\log G/G}\bigr)
  =:X^0(f)+o_p(1).
\]
The full-sample array $\{\psi_f(\kappa_\ell)\}$ is non-random and the only randomness in $X^0(f)$ is the draw $\mathcal C_G$; the mean and variance below are computed in this conditional law.
$X^0(f)$ is exactly unbiased for the full-sample average: every ordered pair is included with probability $n_G/n$, so $\E X^0(f)=n^{-1}\sum_\ell\psi_f(\kappa_\ell)\to\bar A\int f\,d\Gamma$ by Assumption~\ref{ass:dgp}(c).
For its variance, we exploit exchangeability of the draw directly.
Let $\pi$ be a uniform permutation of $\{1,\dots,N\}$ with $\mathcal C_G=\{\pi(1),\dots,\pi(G)\}$, and let $M_k=\E[X^0(f)\mid\pi(1),\dots,\pi(k)]$, $k=0,\dots,G$, so that $M_0=\E X^0(f)$ and $M_G=X^0(f)$.
Swapping the $k$-th drawn country changes at most, $2(G-1)$ of the $n_G =G(G-1)$ ordered pairs, each summand bounded by $C_{\psi}/n_G$, hence $|M_k-M_{k-1}|\le 4C_{\psi}/G$ for each $k$, and orthogonality of martingale increments gives $\operatorname{Var}X^0(f)\le G\cdot(2C/G)^2=O(1/G)$.
Equivalently, symmetrizing $h_{ij}=\tfrac12\{\psi_f(\kappa_{ij})+\psi_f(\kappa_{ji})\}$ writes $X^0(f)=\binom G2^{-1}\sum_{\{i,j\}\subset\mathcal C_G}h_{ij}$ as a bounded
order-two U-statistic on a simple random sample drawn without replacement, whose finite-population variance is $O(1/G)$ by the Hoeffding decomposition. Hence $X(f)\convp\bar A\int f\,d\Gamma$.

By construction of $\psi_f$, the spectral measure on the left of~(i) satisfies $\int f\,d\Gamma_{\mathcal S}=X(f)/X(1)$.
Taking $f\equiv1$ gives $X(1)\convp\bar A>0$, so for every $f$ in a countable convergence-determining class of Lipschitz functions $\int f\,d\Gamma_{\mathcal S}\convp\int f\,d\Gamma$, which is the asserted weak convergence in probability $\Gamma_{\mathcal S}\Rightarrow\Gamma$.

\emph{(ii).} The same argument applies with the summands $\psi_f(\kappa_{\ell,\mathcal S})$ replaced by the entries of $\mu_\ell\tilde x_{\ell,\mathcal S}\tilde x_{\ell,\mathcal S}'$, which are Lipschitz in $\tilde x_{\ell,\mathcal S}$ on bounded sets and so admit the same replacement of $\tilde x_{\ell,\mathcal S}$ by $\tilde x_\ell$ up to $O_p(\sqrt{\log G/G})$.
Conditional on the shocks, the average of the non-random full-sample values $n^{-1}\sum_\ell\mu_\ell\tilde x_\ell\tilde x_\ell'$ has variance $O(1/G)$ by the same permutation-martingale argument.
Hence $n_G^{-1}\sum_{\ell\in\mathcal S}\mu_\ell\tilde x_{\ell,\mathcal S}\tilde x_{\ell,\mathcal S}'\convp H$, which is~(ii).

\emph{(iii).} Lemmas~\ref{lem:margins}, \ref{lem:tailmass}, \ref{lem:fe-rate}, \ref{lem:linf-hessian}, and~\ref{lem:fe-lp} with $G$ in place of $N$ give $\lambda_{\min}(G^{-1}\bar H_{\phi}^{(\mathcal S)}\restriction_{e_{\mathcal S}^\perp})\ge c_4$, $\|u\|_2,\|w\|_2=O_p(G^{2/\alpha-1})$, and $\|u\|_q,\|w\|_q=O_p(G^{2/q-1})$ for $q\in(1,\alpha)$, every country having degree $G-1$.
By the exact identity \eqref{eq:exact-Rn} of Lemma~\ref{lem:remainder}(a), and by the argument of Lemma~\ref{lem:general-beta} run with $(N,\eqref{eq:design-op})$ replaced by $(G,\eqref{eq:design-op-sub})$, the interpolated bound \eqref{eq:score-interp} applies on the subnetwork and gives $\sup_{\beta\in\mathcal B}\|R_{\mathcal S}(\beta)\|=O_p\bigl(G^{\gamma(1-\vartheta)+2\vartheta}\bigr)=o_p(a_{n_G})$ for every $\alpha\in(1,2)$, since $a_{n_G}\asymp G^{2/\alpha}$ and $\gamma<1$.
The subnetwork score covariance remainder satisfies \eqref{eq:meat-remainder} with $N\to G$ by Lemma~\ref{lem:remainder}(c), and the subnetwork profile-Hessian convergence follows as in the proof of Proposition~\ref{prop:primitive-d}.

\emph{(iv).} The shocks $\{\varepsilon_\ell\}_{\ell\in\mathcal S}$ are i.i.d.\ with the full-sample marginal, so by~(i) and Step~2 of the proof of Theorem~\ref{thm:stable}, $a_{n_G}^{-1}\sum_{\ell\in\mathcal S}s_{\ell,\mathcal S} \convd S_\alpha(\Lambda)$ where $s_{\ell,\mathcal S}=\sum_\ell \eta_\ell \mu_\ell \tilde{x}_{\ell,\mathcal S}$, and the joint analogue of \eqref{eq:joint-stable} holds on the subnetwork by the point-process argument there.
The passage from the truth-profiled subnetwork score covariance to the score covariance at the subsample estimates follows from the display in Step~5 of the proof of Theorem~\ref{thm:stable} with $N$ replaced by $G$, its inputs being supplied by~(iii).
Here Steps~2 and~5 are applied conditionally on the draw, for which the subnetwork array is deterministic; since the convergence in~(i) holds in probability, it holds almost surely along a further subsequence of any subsequence of draws, and bounded convergence over the draw then yields the convergence under the joint law of shocks and draw.
With~(ii)--(iii) and the continuous mapping theorem applied to the self-normalized ratio, whose limiting denominator is a.s.\ positive, $T_{\mathcal S}\convd J^\star_\alpha$.
\end{proof}

\begin{lemma}[Atomlessness of the self-normalized limit]\label{lem:atomless}
Let $\Pi$ be the Poisson random measure of \eqref{eq:prm} and let $(S_\alpha(\Lambda),\Xi)$ be the joint limit of \eqref{eq:joint-stable}, realized as the compensated linear and uncentered quadratic integrals of $\Pi$. Fix $k$, let $\iota_k$ be the $k$-th standard basis vector of $\R^p$, and set $g(\theta):=\iota_k'H^{-1}\theta$, $X:=\iota_k'H^{-1}S_\alpha(\Lambda)$, and $V:=\iota_k'H^{-1}\Xi H^{-1}\iota_k$.
\begin{enumerate}
\item[\emph{(i)}] If $\int_{\mathbb S^{p-1}}g(\theta)^2\,\Gamma(d\theta)>0$, then $V>0$ almost surely, $\Pr(X=t\sqrt V)=0$ for every $t\in\R$, and $J^\star_\alpha=X/\sqrt V$ has a continuous distribution function.
\item[\emph{(ii)}] Define the tail-balance constants
\[
  c_{k,+}:=\bar A\int_{\mathbb S^{p-1}}\bigl(g(\theta)_+\bigr)^\alpha\,\Gamma(d\theta)
  =C_\alpha^{-1}\int_{\mathbb S^{p-1}}\bigl(g(\theta)_+\bigr)^\alpha\,\Lambda(d\theta),
  \qquad
  c_{k,-}:=\bar A\int_{\mathbb S^{p-1}}\bigl(g(\theta)_-\bigr)^\alpha\,\Gamma(d\theta),
\]
with $x_+=\max(x,0)$ and $x_-=\max(-x,0)$, matching the polar representation $\varpi(dr,d\theta)=\alpha r^{-\alpha-1}dr\,\bar A\,\Gamma(d\theta)$ of the mean measure in \eqref{eq:prm}. Under the condition of~(i), $c_{k,+}+c_{k,-}>0$, and $J^\star_\alpha$ is distributed as the scalar self-normalized $\alpha$-stable limit of \citet{LoganMallowsRiceShepp1973} with tail balance $(c_{k,+},c_{k,-})$; in particular, $J^\star_\alpha\neq\mathcal N(0,1)$ for every $\alpha\in(1,2)$.
\end{enumerate}
\end{lemma}

\begin{proof}
\emph{(i) Atomlessness.}
We isolate the first jump $m^\star$ with $g(\theta_{m^\star})\neq0$ and show that, conditionally on all other jumps, the event $\{X=t\sqrt V\}$ constrains its arrival time to at most two values, which is a null event because that arrival time has an atomless conditional law.

By the series representation of the Poisson integrals \citep[Section~3.10]{SamorodnitskyTaqqu1994}, realize $\Pi=\sum_{m\ge1}\delta_{(\varsigma\,\Upsilon_m^{-1/\alpha},\,\theta_m)}$ in polar coordinates, with $\{\Upsilon_m\}_{m\ge1}$ the arrival times of a unit-rate Poisson process on $(0,\infty)$, $\{\theta_m\}$ i.i.d.\ with law $\tilde\Gamma:=\Gamma/\Gamma(\mathbb S^{p-1})$ independent of $\{\Upsilon_m\}$, and $\varsigma>0$ a scale constant. Then
\[
  X=\sum_{m\ge1}\bigl[g(\theta_m)\,\varsigma\,\Upsilon_m^{-1/\alpha}-b_m\bigr],
  \qquad
  V=\varsigma^2\sum_{m\ge1}\Upsilon_m^{-2/\alpha}g(\theta_m)^2,
\]
with $\{b_m\}$ deterministic centering constants; both series converge almost surely for $\alpha\in(1,2)$.

The condition of~(i), $\int_{\mathbb S^{p-1}}g(\theta)^2\,\Gamma(d\theta)>0$, is equivalent to $\tilde\Gamma(\{g\neq0\})>0$, so the i.i.d.\ sequence $\{\theta_m\}$ almost surely contains infinitely many indices with $g(\theta_m)\neq0$; hence $V>0$ a.s.\ and, with $m^\star:=\min\{m:g(\theta_m)\neq0\}$, $B:=V-\varsigma^2\Upsilon_{m^\star}^{-2/\alpha}g(\theta_{m^\star})^2>0$ a.s.

Fix $t\in\R$ and condition on $\mathcal F:=\sigma(\{\theta_m\}_{m\ge1},\{\Upsilon_m\}_{m\neq m^\star})$; note $m^\star$ is $\mathcal F$-measurable. Given $\mathcal F$, the conditional law of $\Upsilon_{m^\star}$ is uniform on $(\Upsilon_{m^\star-1},\Upsilon_{m^\star+1})$ (with $\Upsilon_0:=0$), hence atomless, and $(X,V)$ depend on $\Upsilon_{m^\star}$ only through $r:=\varsigma\,\Upsilon_{m^\star}^{-1/\alpha}$:
\[
  X=A+g(\theta_{m^\star})\,r,
  \qquad
  V=B+g(\theta_{m^\star})^2r^2,
\]
with $A,B$ $\mathcal F$-measurable, $g(\theta_{m^\star})\neq0$, and $B>0$ a.s. The equation $A+g(\theta_{m^\star})r=t\,(B+g(\theta_{m^\star})^2r^2)^{1/2}$ has at most two solutions in $r>0$: squaring yields the quadratic $g(\theta_{m^\star})^2(1-t^2)r^2+2Ag(\theta_{m^\star})r+A^2-t^2B=0$, which is non-degenerate unless $t^2=1$ and $A=0$, in which case the original equation reads $g(\theta_{m^\star})r=\pm(B+g(\theta_{m^\star})^2r^2)^{1/2}$, impossible for $B>0$. Since $r\mapsto\Upsilon_{m^\star}$ is a bijection and the conditional law of $\Upsilon_{m^\star}$ is atomless, $\Pr(X=t\sqrt V\mid\mathcal F)=0$ a.s.; integrating gives $\Pr(X=t\sqrt V)=0$. As $V>0$ a.s., the law of $J^\star_\alpha=X/\sqrt V$ assigns no mass to any point, i.e.\ its distribution function is continuous, proving~(i).

\emph{(ii) Nonnormality via the projected point process.}
We realize $(X,V)$ as functionals of a single scalar $\alpha$-stable point process $\Pi_k$ on the line, so that $J^\star_\alpha$ becomes the self-normalized limit of an ordinary i.i.d.\ sum with tail balance $(c_{k,+},c_{k,-})$. Then the Giné--Götze--Mason characterization of asymptotic normality fails for $\alpha<2$.

Map each point of $\Pi$ by $(r,\theta)\mapsto x=r\,g(\theta)$, discarding points with $g(\theta)=0$. By the mapping theorem for Poisson random measures, the image $\Pi_k$ is a Poisson random measure on $\R\setminus\{0\}$ with mean measure determined by
\[
  \varpi_k\bigl((x,\infty)\bigr)=c_{k,+}\,x^{-\alpha},
  \qquad
  \varpi_k\bigl((-\infty,-x)\bigr)=c_{k,-}\,x^{-\alpha},
  \qquad x>0 ,
\]
an $\alpha$-stable scalar intensity with tail balance $(c_{k,+},c_{k,-})$. The condition of~(i) gives $\tilde\Gamma(\{g\neq0\})>0$, hence $c_{k,+}+c_{k,-}>0$. Linearity carries the compensation through the projection:
\[
  X\;=\;\iota_k'H^{-1}S_\alpha(\Lambda)
  \;=\;\lim_{\epsilon\downarrow0}
  \Bigl[\int_{|x|>\epsilon}x\,\Pi_k(dx)
  -\int_{|x|>\epsilon}x\,\varpi_k(dx)\Bigr],
  \qquad
  V\;=\;\int x^2\,\Pi_k(dx),
\]
so the pair $(X,V)$ is a functional of the single scalar Poisson random measure $\Pi_k$. The truncation sets require one reconciliation: $X$ is constructed in Step~5 of the proof of Theorem~\ref{thm:stable} with the vector truncation $\{\|z\|>\epsilon\}$, whereas the display uses the scalar truncation $\{|x|=r|g(\theta)|>\epsilon\}$. The two compensated integrals differ by a compensated Poisson integral over the symmetric difference of the truncation sets, with variance at most
\[
  \int r^2g(\theta)^2\,
  \1\bigl\{\1\{r>\epsilon\}\neq\1\{r|g(\theta)|>\epsilon\}\bigr\}\,
  \varpi(dr,d\theta)
  \;\le\;C\,\epsilon^{2-\alpha}
  \int_{\mathbb S^{p-1}}
  \bigl(|g(\theta)|^{\alpha}+g(\theta)^{2}\bigr)\,\Gamma(d\theta)
  \;\longrightarrow\;0
\]
as $\epsilon\downarrow0$, by the polar form of $\varpi$ and the boundedness of $g$; both truncations therefore yield the same limit $X$. 

It remains to identify the law of $J^\star_\alpha=X/\sqrt V$ as a scalar self-normalized limit and to rule out normality.
Let $\{\chi_i\}$ be i.i.d.\ mean-zero with regularly varying tails $\Pr(\chi_1>t)\sim c\,c_{k,+}t^{-\alpha}$ and $\Pr(\chi_1<-t)\sim c\,c_{k,-}t^{-\alpha}$ for some $c>0$, so that its normalized point process is $\Pi_k$ up to the common scale $c$.
The two-sided analogue of the joint limit \eqref{eq:joint-stable} with $p=1$ and $\kappa_\ell\equiv1$, whose proof applies verbatim with the general tail balance carried through the characteristic-function expansion \eqref{eq:cf-expansion}, gives
\[
  \Bigl(b_m^{-1}\textstyle\sum_{i\le m}\chi_i,\;b_m^{-2}\sum_{i\le m}\chi_i^2\Bigr)
  \;\convd\;(X,V)
\]
up to a common positive scale, which cancels in the ratio; see also \citet{LoganMallowsRiceShepp1973}.
Hence $J^\star_\alpha$ is the distributional limit of the self-normalized sum $\sum_{i\le m}\chi_i\big/\bigl(\sum_{i\le m}\chi_i^2\bigr)^{1/2}$ of an i.i.d.\ centered sequence with tail index $\alpha\in(1,2)$.

By the characterization of \citet{GineGotzeMason1997}, such a self-normalized sum converges to $\mathcal N(0,1)$ if and only if the summand is centered and lies in the domain of attraction of the Gaussian law.
A variable with regularly varying tails of index $\alpha<2$ does not lie in that domain, so the limit of the self-normalized sum, which exists and is unique, is not $\mathcal N(0,1)$.
Therefore $J^\star_\alpha\neq\mathcal N(0,1)$ for every $\alpha\in(1,2)$.
\end{proof}

\begin{lemma}[Vertex subsampling consistency]\label{lem:vertex-subsample}
Let $Z_{\mathcal C}$ be a statistic of the induced subnetwork on a country set $\mathcal C$, measurable with respect to the shocks on $\mathcal C$, with $Z_{\mathcal C}\convd J^\star_\alpha$ under the joint law of the shocks and a uniform draw, where $F_{J^\star_\alpha}$ is continuous (Lemma~\ref{lem:atomless}). If $Z_{\mathcal C}$ is undefined for some $\mathcal C$, set $Z_{\mathcal C}=0$. Define
\[
  F^{c}_{n,G}(t):=\binom NG^{-1}\sum_{|\mathcal C|=G}\1\{Z_{\mathcal C}\le t\},
  \qquad
  \widehat F_{n,G,M}(t):=M^{-1}\sum_{m=1}^M\1\{Z_{\mathcal C^{(m)}}\le t\},
\]
the latter over $M$ independent uniform draws. If $G\to\infty$ and $G/N\to0$, then
\[
  \sup_t\bigl|F^{c}_{n,G}(t)-F_{J^\star_\alpha}(t)\bigr|\convp0,
  \qquad\text{and, as $M\to\infty$,}\qquad
  \sup_t\bigl|\widehat F_{n,G,M}(t)-F_{J^\star_\alpha}(t)\bigr|\convp0 .
\]
\end{lemma}

\begin{proof}
We first show that the undefined statistics are asymptotically negligible. Let
\[
  \pi^{\mathrm{fail}}_{n,G}:=\binom NG^{-1}\sum_{|\mathcal C|=G}\1\{Z_{\mathcal C}\text{ undefined on }\mathcal C\},
\]
whose shock-expectation equals the probability that $Z_{\mathcal C}$ is undefined on a uniform draw. For the PPML statistics this probability vanishes: Corollary~\ref{cor:sample-existence} applies with $G$ in place of $N$, its margin inputs (Lemma~\ref{lem:margins}(i)--(ii)) transferring to the induced subnetwork as in Lemma~\ref{lem:subnet-inherit} with a bound uniform over $\mathcal C$, so the strict Hall conditions hold w.p.a.\ one. Markov's inequality gives $\pi^{\mathrm{fail}}_{n,G}\convp0$, and replacing $Z_{\mathcal C}$ on this fraction of subsets moves each cdf by at most $\pi^{\mathrm{fail}}_{n,G}$ uniformly in $t$, so the convention $Z_{\mathcal C}=0$ does not affect the limits.

Write $F_{n,G}(t):=\Pr_{\eta,\mathcal C}(Z_{\mathcal C}\le t)$. By the maintained convergence $Z_{\mathcal C}\convd J^\star_\alpha$ and continuity of $F_{J^\star_\alpha}$, $F_{n,G}(t)\to F_{J^\star_\alpha}(t)$ for every $t$. Conditional on the shocks, the Monte Carlo error is $\sup_t|\widehat F_{n,G,M}(t)-F^{c}_{n,G}(t)|=O_p(M^{-1/2})$ by the Dvoretzky--Kiefer--Wolfowitz inequality. It therefore suffices to show $F^{c}_{n,G}(t)-F_{n,G}(t)\convp0$ pointwise over the shocks.

$F^{c}_{n,G}(t)$ is a complete subset average over the $N$ countries with summands $h_t(\mathcal C)=\1\{Z_{\mathcal C}\le t\}\in[0,1]$ and shock-mean $\E_\eta[F^{c}_{n,G}(t)]=F_{n,G}(t)$. Since the deterministic design varies with $\mathcal C$, center each summand separately, $\tilde h_{\mathcal C}:=h_t(\mathcal C)-\E_\eta[h_t(\mathcal C)]$. 
Centered summands on disjoint country blocks are independent: disjoint blocks induce disjoint sets of directed pairs, whose shocks are independent.
By Hoeffding's representation of a complete subset average as an average, over permutations $\pi$ independent of the shocks, of block means over $K=\lfloor N/G\rfloor$ disjoint $G$-country blocks $\mathcal C_1,\dots,\mathcal C_K$ \citep[Section~5]{Hoeffding1963},
\[
  F^{c}_{n,G}(t)-F_{n,G}(t)=\E_\pi\Bigl[K^{-1}\textstyle\sum_{b\le K}\tilde h_{\mathcal C_b}\Bigr],
\]
so Jensen's inequality and block independence give
\[
  \operatorname{Var}_\eta\bigl(F^{c}_{n,G}(t)\bigr)
  \;\le\;\E_\pi\Bigl[K^{-2}\textstyle\sum_{b\le K}\operatorname{Var}_\eta(\tilde h_{\mathcal C_b})\Bigr]
  \;\le\;\frac{1}{4K}\;=\;O\Bigl(\frac GN\Bigr)\;\longrightarrow\;0 .
\]
Hence $F^{c}_{n,G}(t)\convp F_{J^\star_\alpha}(t)$, and with the Monte Carlo error also $\widehat F_{n,G,M}(t)\convp F_{J^\star_\alpha}(t)$, for every $t$. Monotonicity of both cdfs and continuity of $F_{J^\star_\alpha}$ generalize these pointwise limits to the two uniform statements by P\'olya's argument.
\end{proof}

\section{Cluster Standard Errors under Heavy Tails}
\label{app:cluster}

\setcounter{theorem}{0}
\setcounter{proposition}{0}
\setcounter{lemma}{0}
\setcounter{equation}{0}

Clustering on the country or the country pair is the standard robustness device in applied gravity work: it permits arbitrary within-cluster dependence while requiring only independence across clusters. It does not, however, restore valid inference under heavy tails. The cluster-robust sandwich is justified by a finite-variance cluster central limit theorem, and under Assumption~\ref{ass:dgp}(a) the cluster score inherits the same infinite-variance tail as the pair-level score. The paper reports cluster confidence intervals in Section~\ref{sec:application} only as a conventional benchmark, not as a remedy; this appendix explains heuristically why that benchmark fails.

Partition the $n$ dyads into disjoint clusters $\mathcal{C}_1,\dots,\mathcal{C}_{n_c}$ (clustering on the exporter gives $n_c=N$ clusters; on the country pair, $n_c\asymp N^2$). 
The cluster score and the cluster-robust score covariance estimator are
\begin{equation}
  T_c=\sum_{\ell\in\mathcal{C}_c}s_\ell,\qquad\hat\Omega_n^{\mathrm{clus}}=\frac1n\sum_{c=1}^{n_c}\hat T_c\hat T_c',\qquad\hat T_c=\sum_{\ell\in\mathcal{C}_c}\hat s_\ell(\hat\beta),
  \label{eq:cluster-score-meat}
\end{equation}
with $s_\ell=\eta_\ell\kappa_\ell$ and $\kappa_\ell=\mu_\ell\tilde x_\ell$. Under Assumption~\ref{ass:dgp}, the cluster scores $\{T_c\}$ are independent across $c$. 
When within-cluster dependence is instead allowed, independence across clusters must be imposed separately. 
In either case the cluster-robust $t$-statistic is asymptotically standard normal only if the independent cluster array obeys a finite-variance CLT.

We now heuristically argue that $T_{c,k}=\sum_{\ell\in \mathcal{C}_c}\eta_\ell\kappa_{\ell k}$ is heavy-tailed, and a rigorous treatment is beyond the scope of this paper.
Fix a cluster $c$ with identifying variation in the coordinate $k$ of interest, i.e.\ $\kappa_{\ell k}\neq0$ for some $\ell\in\mathcal{C}_c$.
Extremes of heavy-tailed sums are governed by the one-large-jump principle: a single large pair shock $\eta_{\ell^*}$ dominates the cluster total, so that $T_{c,k}\approx\eta_{\ell^*}\kappa_{\ell^* k}$ on the relevant events and
\begin{equation}
  \Pr\bigl(|T_{c,k}|>t\bigr)\;\sim\; L_{c,k}\,t^{-\alpha},
  \qquad L_{c,k}>0,
  \label{eq:cluster-tail-heuristic}
\end{equation}
as $t\to\infty$, for a cluster-specific tail constant $L_{c,k}$.
Unlike the independent benchmark, in which $L_{c,k}=c_+\sum_{\ell}|\kappa_{\ell k}|^{\alpha}$ splits by sign, this constant depends on the joint tail of $(\eta_\ell)_{\ell\in\mathcal{C}_c}$ and may be
larger or smaller.
Formally, \eqref{eq:cluster-tail-heuristic} holds whenever the cluster shock vector is multivariate regularly varying with index $\alpha$ and the weights $\{\kappa_{\ell k}\}$ do not produce exact tail cancellation.

The implication does not depend on the value of $L_{c,k}$.
Because $\alpha<2$, \eqref{eq:cluster-tail-heuristic} gives $\E[T_{c,k}^2]=\infty$, so the finite-variance cluster CLT that justifies the conventional cluster-robust standard error does not apply, and asymptotic standard normality of the studentized cluster statistic is not justified; under suitable multivariate regular-variation and anti-cancellation conditions on the cluster shock vectors, a nonstandard self-normalized stable limit is expected instead.
The mechanism is generic: clustering changes the variance estimator but not the infinite-variance score from which it is built, and the choice of clustering level does not affect this conclusion.
Within-cluster dependence typically reinforces it by concentrating the cluster total on a single extreme, although dependence can also alter tail constants or, in special configurations, produce cancellation.
The $m$-out-of-$n$ bootstrap instead forms the self-normalized statistic \eqref{eq:SN-defn}. Its score-sum numerator and heteroskedasticity-robust denominator carry the same stable scale $a_n$, so the ratio converges to the law $J^\star_\alpha$ of Theorem~\ref{thm:stable}, which is not $\mathcal N(0,1)$. Country-level resampling then estimates the quantiles of $J^\star_\alpha$ consistently.
Validity comes not from removing the infinite variance from the $t$-statistic but from replacing the normal critical values with the subsampling ones.

\bibliographystyle{apalike}
\bibliography{references}

\end{document}